\documentclass[a4paper,10pt]{article}
\usepackage{textcomp}
\usepackage{a4wide}
\usepackage{amsmath}
\usepackage{amssymb}
\usepackage{xcolor}
\usepackage{algorithm,algorithmic-numalg}
\usepackage{amsthm}
\usepackage{multirow}
\usepackage{booktabs}
\usepackage[table]{colortbl}
\usepackage{url}
\usepackage{centernot}
\usepackage{subcaption}
\usepackage{varwidth}

\theoremstyle{definition}
\newtheorem{definition}{Definition}[section]
\newtheorem{invariant}[definition]{Invariant}
\newtheorem{example}[definition]{Example}
\theoremstyle{plain}

\newtheorem{theorem}[definition]{Theorem}
\newtheorem{lemma}[definition]{Lemma}

\newcommand{\bigo}[1]{{O\!\left(#1\right)}}
\newcommand{\bigO}{O}   
\newlength{\leftrightarrowwidth}
\newcommand{\bisauxiliary}{\raisebox{.3ex}{\makebox[\leftrightarrowwidth]{$\underline{\makebox[0.7\leftrightarrowwidth]{$\leftrightarrow$}}$}}}

\newcommand{\bbis}{\mathrel{\bisauxiliary_b\!}}
\newcommand{\pijl}[1]{\mathrel{\text{$\xrightarrow{\smash[t]{#1}}$}}}

\newcommand{\lijp}[1]{\mathrel{\text{$\xleftarrow{\smash[t]{#1}}$}}}
\newcommand{\R}[2]{{#1 \mathrel{R} #2}}
\newcommand{\size}[1]{\left\lvert #1 \right\rvert}

\newcommand{\remark}[1]{}

\definecolor{jfgcolor}{rgb}{0.2, 0, 1}

\definecolor{jkcolor}{rgb}{1, 0.5, 0}
\newcommand{\jk}[1]{\remark{\textcolor{jkcolor}{JK: #1.}}}
\definecolor{awcolor}{rgb}{0.8, 0.5, 0.8}

\newcommand{\TB}{T_{B\pijl{}}}

\newcommand{\TBaB}{T_{B\pijl{a} B'}}
\newcommand{\TaB}{T_{\pijl{a} B'}}
\newcommand{\BT}{B_{\pijl{T}}}

\newcommand{\nraB}{\#aB'}

\newcommand{\Bottom}{\mathit{Bottom}}
\newcommand{\Marked}{\mathit{Marked}}

\newcommand{\splittableBlocks}{\mathit{splittableBlocks}}

\newcommand{\ob}{$\blacktriangle$} 
\newcommand{\ns}{\phantom{$\blacktriangle$}} 
\newcommand{\og}{$\blacktriangledown$} 
\newcommand{\shortto}{\mathord{\tikz[baseline]{\clip(0em,-0.3ex) rectangle (0.5em,1.8ex); \node[draw=none,inner sep=0pt,anchor=base](0em,0ex){$\to$};}\!}}
\newcommand{\xshortto}[1]{\mathord{\tikz[baseline]{\clip(0em,-0.3ex) rectangle (0.5em,2.5ex); \node[draw=none,inner sep=0pt,anchor=base](0em,0ex){$\smash[t]{\stackrel{~#1}{\rightarrow}}$};}\!}}

\usepackage{tikz}
\usetikzlibrary{decorations.pathreplacing}
\usetikzlibrary{arrows.meta}
\usetikzlibrary{intersections}
\usetikzlibrary{bending}
\usetikzlibrary{backgrounds}
\usetikzlibrary{shapes.geometric}
\usetikzlibrary{shapes.misc}

\newlength{\halflineskip}
\newcommand{\rightbrace}[3][0]{\setlength{\halflineskip}{0.5\baselineskip}\hspace{0pt plus 1filll}\makebox[0pt][l]{\smash{\tikz[baseline=#1\baselineskip]{\draw[thick,color=gray,line cap=round,decorate,decoration={brace,amplitude=2.3pt}] (0pt,0.7\baselineskip) ++(0.1em,0) -- ++(0pt,-#2\baselineskip+1pt) (5pt,\halflineskip-#2\halflineskip) node[color=gray,right,inner sep=0.1em,anchor=base west]{\small #3};}}}}

\newlength{\statetextwidth}
\newlength{\statetextheight}
\newlength{\statetextdepth}
\tikzset{bunch/.style={very thin,draw,fill=gray!5}}
\tikzset{block/.style={thick,draw,fill=gray!15,rounded corners=4.95mm}}
\tikzset{state/.style={circle,draw,semithick,fill=white,inner sep=1.35pt,text height=\statetextheight,text depth=\statetextdepth}}
\tikzset{red state/.style={state,fill=red!10,path picture={\draw[red!40,very thick]
	(-0.25\statetextwidth,-1cm) -- (-0.25\statetextwidth,2cm)
	(0cm,-1cm) -- (0cm,2cm)
	(0.25\statetextwidth,-1cm) -- (0.25\statetextwidth,2cm)
	(0.5\statetextwidth,-1cm) -- (0.5\statetextwidth,2cm)
	(0.75\statetextwidth,-1cm) -- (0.75\statetextwidth,2cm)
	(\statetextwidth,-1cm) -- (\statetextwidth,2cm)
	(1.25\statetextwidth,-1cm) -- (1.25\statetextwidth,2cm);}}}
\tikzset{blue state/.style={state,fill=blue!10,path picture={\draw[blue!40,very thick]
	(-1cm,0.5\statetextheight-0.5\statetextdepth+0.75\statetextwidth) -- (2cm,0.5\statetextheight-0.5\statetextdepth+0.75\statetextwidth)
	(-1cm,0.5\statetextheight-0.5\statetextdepth+0.5\statetextwidth) -- (2cm,0.5\statetextheight-0.5\statetextdepth+0.5\statetextwidth)
	(-1cm,0.5\statetextheight-0.5\statetextdepth+0.25\statetextwidth) -- (2cm,0.5\statetextheight-0.5\statetextdepth+0.25\statetextwidth)
	(-1cm,0.5\statetextheight-0.5\statetextdepth) -- (2cm,0.5\statetextheight-0.5\statetextdepth)
	(-1cm,0.5\statetextheight-0.5\statetextdepth-0.25\statetextwidth) -- (2cm,0.5\statetextheight-0.5\statetextdepth-0.25\statetextwidth)
	(-1cm,0.5\statetextheight-0.5\statetextdepth-0.5\statetextwidth) -- (2cm,0.5\statetextheight-0.5\statetextdepth-0.5\statetextwidth)
	(-1cm,0.5\statetextheight-0.5\statetextdepth-0.75\statetextwidth) -- (2cm,0.5\statetextheight-0.5\statetextdepth-0.75\statetextwidth);}}}
\newcommand{\redstatetext}{\smash{\tikz[baseline=-0.37\statetextwidth]{\node[red state,thin,inner sep=0pt]{\makebox[0.5\statetextwidth]{}};}}}
\newcommand{\bluestatetext}{\smash{\tikz[baseline=-0.37\statetextwidth]{\node[blue state,thin,inner sep=0pt]{\makebox[0.5\statetextwidth]{}};}}}

\renewcommand{\algorithmiccomment}[1]{\hfill{\color{gray}// #1}}

\title{\bf A simpler \textit{O}(\textit{m}\,log\,\textit{n}) algorithm for branching bisimilarity on labelled transition systems}
\author{
David N.\ Jansen\thanks{This work is partly done during a visit of the first author at Eindhoven University of Technology, and
a visit of the second author at the Institute of Software, Chinese Academy of Sciences.}\\
\small Institute of Software, Chinese Academy of Sciences, Beijing, China\\[-0.5ex]
\small and Sino-Europe Institute of Dependable Smart Software\vspace{2ex}\\
Jan Friso Groote\footnotemark[1]\hspace{1cm}
Jeroen J.A.\ Keiren\hspace{1cm}
Anton Wijs\\
\small Department of Mathematics and Computer Science, Eindhoven University of Technology, The Netherlands\vspace{2ex}\\
\small\texttt{dnjansen@ios.ac.cn, \{J.F.Groote, J.J.A.Keiren, A.J.Wijs\}@tue.nl}
}

\date{}
\begin{document}

\maketitle
\begin{abstract}
\noindent%
Branching bisimilarity is a behavioural equivalence relation on labelled transition systems
that takes internal actions into account.
It has the traditional advantage that algorithms for branching bisimilarity
are more efficient than all algorithms for other weak behavioural equivalences, especially weak bisimilarity.
With $m$ the number of transitions and $n$ the number of states,
the classic $\bigo{m n}$ algorithm was recently replaced by an $\bigo{m (\log \size{\mathit{Act}} + \log n)}$ algorithm \cite{GJKW2017}, which is unfortunately rather complex.
This paper combines its ideas with the ideas from Valmari \cite{Valmari2009}.
This results in a simpler algorithm with complexity $\bigo{m \log n}$.
Benchmarks show that this new algorithm is also faster and often far more memory efficient
than its predecessors.
This makes it the best option for branching bisimulation minimisation and preprocessing for weak bisimulation of LTSs.
\end{abstract}

\section{Introduction}
Branching bisimilarity \cite{DBLP:journals/jacm/GlabbeekW96} is an alternative to weak bisimilarity \cite{DBLP:books/sp/Milner80}.
Both equivalences allow the reduction of labelled transition systems containing transitions labelled with internal actions, which are also referred to as silent, hidden or $\tau$-actions.

One of the distinct advantages of branching bisimilarity is that, from the outset, an efficient
algorithm has been available~\cite{GV90}.
This algorithm can be used to calculate whether two states in a labelled transition system are equivalent,
and to calculate a quotient transition system.
The algorithm had complexity $\bigo{mn}$ with $m$ the number of transitions and $n$ the number of states.
It is more efficient than classic algorithms for weak bisimilarity,
which use transitive closure (for instance, \cite{KannellakisSmolka1990} runs in $\bigo{n^2m\log n + mn^{2.376}}$, where $n^{2.376}$ is the time for computing the transitive closure), and algorithms for weak simulation (the strong simulation relation can be computed in $\bigo{mn}$ \cite{HHK1995}, and for weak simulation first the transitive closure of the transition relation needs to be computed).
The algorithm is also far more efficient than algorithms for trace-based equivalence
notions, such as (weak) trace equivalence or weak failure equivalence,
as these are generally PSPACE-complete on finite-state labelled transition systems \cite{KannellakisSmolka1990}.

Branching bisimilarity is interesting in several other respects.
Not only is it a useful notion to compare the
behaviour of labelled transition systems directly, as it exactly respects the branching structure of behaviour, it also enjoys a number
of nice mathematical properties such as the existence of a canonical quotient
with a minimal number of states and transitions modulo branching bisimilarity
(contrary to, for instance, trace-based equivalences).
Additionally, as branching bisimilarity is coarser than virtually any other conceivable
behavioural equivalence taking internal actions into account \cite{Gla1993}, it is ideal for preprocessing.
In order to calculate a desired equivalence, one can first reduce the behaviour modulo branching bisimilarity, before applying a dedicated
algorithm on the often substantially reduced transition system. In the mCRL2 toolset \cite{DBLP:conf/tacas/BunteGKLNVWWW19} this is common practice.

In \cite{GJKW2017,GrooteW16tacas} an algorithm to calculate stuttering equivalence on Kripke structures
with complexity $\bigo{m\log n}$ was proposed.
Stuttering equivalence essentially differs from branching bisimilarity in the fact that transitions do not have labels and
as such all transitions can be viewed as internal. In these papers it was shown that branching bisimilarity can be calculated by
translating labelled transition systems to Kripke structures, encoding the labels of transitions into labelled states following \cite{dNV1995,RSW2014}.
This led to an $\bigo{m (\log \size{\mathit{Act}} + \log n)}$ or $\bigo{m \log m}$ algorithm for branching bisimilarity.
In Appendix~\ref{app:complexity_via_stuttering} we include and example that shows this bound is, in fact, tight.

Besides the time complexity, the algorithm in \cite{GJKW2017,GrooteW16tacas} has two disadvantages.
First, the translation to Kripke structures introduces a new state and a new transition
per action label and target state of a transition,
and as such increases the memory required to calculate branching bisimilarity substantially,
depending on the structure of the transition system.
This made it far less memory efficient than the classical algorithm of \cite{GV90},
and this was actually perceived as a substantial practical hindrance.
For instance, when reducing systems consisting of tens of millions of states, such as \cite{BLW2018}, memory consumption is the bottleneck of the algorithm from \cite{GJKW2017,GrooteW16tacas}.
Second, the algorithm in \cite{GJKW2017,GrooteW16tacas} is very complex. To realise the targeted $\bigo{m\log n}$ complexity, several subtle situations that can occur while running the algorithm were handled using dedicated subalgorithmic steps.
To illustrate the complexity, implementing the algorithm of \cite{GJKW2017,GrooteW16tacas} took approximately half a man-year.

\paragraph{Contributions.}
We present an algorithm for branching bisimilarity that runs directly on labelled transition systems in $\bigo{m \log n}$ time and that is simpler than the algorithm of \cite{GJKW2017,GrooteW16tacas}.

To achieve this we use an idea
from Valmari and Lehtinen \cite{Valmari2009,ValmariL08}, which they apply in the context of strong bisimulation.
The standard Paige--Tarjan algorithm~\cite{PT87},
which has $\bigo{m \log n}$ time complexity for strong bisimilarity on Kripke structures,
registers work done in a separate partition of states.
Valmari~\cite{Valmari2009} observed that this leads to complexity $\bigo{m \log m}$ on LTSs
and proposed to use a partition of transitions,
whose elements he (and we) calls \textit{bunches,}
to register work done.
This reduces the time complexity on LTSs to $\bigo{m \log n}$.

Using this idea we design our more straightforward algorithm for branching bisimilarity on labelled transition
systems.
Essentially, this makes the maintenance of action labels particularly straightforward and allows to simplify
stability reassessment in case of new bottom states. It also leads to a novel main invariant,
which we formulate as Invariant~\ref{inv:bunches}. It allows us to prove the correctness of the algorithm in a far more straightforward way than before.

We provide a detailed proof of correctness of the algorithm
and demonstrate using benchmarks that it outperforms all
preceding algorithms both in time and space when the labelled transition systems become sizeable.
This is illustrated with more than 30 example LTSs.
This shows that the new algorithm pushes the state-of-the-art in comparing and minimising the behaviour of LTSs w.r.t.\@ weak equivalences,
either directly (branching bisimilarity) or in the form of a preprocessing step (weak bisimilarity).

Despite the fact that this new algorithm
is more straightforward than the $\bigo{m (\log \size{\mathit{Act}} + \log n)}$ algorithm~\cite{GJKW2017}, the implementation of the algorithm into code is still not easy.
To guard against implementation errors, we extensively applied random testing,
comparing the output with that of other algorithms.
The algorithms and their source code are freely available for use in
the mCRL2 toolset~\cite{DBLP:conf/tacas/BunteGKLNVWWW19}.

\paragraph{Historical overview.}
For those new to the area of algorithms for bisimulation equivalences on labelled transition systems, it might be useful
to review the major concepts that have been developed in this field throughout the years.
Following Kanellakis and Smolka \cite{KanellakisSmolka1983},
efficient algorithms use partition refinement.
The states of the transition system are partitioned into blocks,
such that equivalent states are in the same block.
These blocks are refined until non-equivalent states are in different blocks.
The original algorithm~\cite{KanellakisSmolka1983} calculated
strong bisimilarity and had time complexity $\bigo{mn}$.
The main idea is to find a \textit{splitter,}
i.e., a block that shows that some states are not equivalent,
and then move these states to separate blocks.

Subsequently,
the seminal article of Paige and Tarjan \cite{PT87} presented an efficient algorithm for strong bisimulation minimisation of Kripke structures.
Its main data structure consists of two partitions, a fine one into blocks and a coarse one into
(what is called in~\cite{GJKW2017}) \textit{constellations} of blocks.
The fine partition stores the current knowledge about inequivalence of states,
and the coarse partition stores the current knowledge on which blocks cannot act as splitters.
When the two partitions coincide, no more splitters exist,
so the blocks in the fine partition are the bisimulation equivalence classes.
Paige and Tarjan's algorithm repeatedly splits a constellation with multiple blocks into two parts
and splits the fine partition if the new constellations actually lead to some splits.
Their ingenious data structures ensure Hopcroft's ``Process the smaller half'' principle \cite{Hop1971}, guaranteeing
that the work done requires time proportional to the \textit{smaller} of the splitters.
This leads to a time complexity in $\bigo{m \log n}$ on Kripke structures.

The next key insight was by Valmari,
who introduced the idea to use a partition of transitions into bunches \cite{Valmari2009}
to store the knowledge about non-splitters.
This resulted in an $\bigo{m \log n}$ algorithm for strong bisimilarity on labelled transition systems,
even though $m$ may be larger than $n^2$ when there are enough transition labels.

For branching bisimilarity,
Groote and Vaandrager \cite{GV90} presented an algorithm with worst-case time complexity in $\bigo{mn}$.
They established that in order to determine that a block can be split it is only necessary to look
at its \textit{bottom states,} i.e., states that have no outgoing internal transition in the same block.
In addition to the algorithm of \cite{KanellakisSmolka1983} the algorithm of Groote and Vaandrager \cite{GV90}
had to determine which states can reach certain bottom states via internal transitions in the block to split that block.
They also observed that stability is not preserved
when new bottom states emerge. Therefore, stability of existing blocks had to be reassessed
by the algorithm.

More than 25 years later, \cite{GJKW2017,GrooteW16tacas} managed to merge the ideas of Paige and Tarjan with those of
Groote and Vaandrager, finding an algorithm for stuttering equivalence
that has time complexity $\bigo{m \log n}$ on Kripke structures
as well as time complexity $\bigo{m (\log \size{\mathit{Act}} + \log n)}$ on LTSs using a translation of LTSs to Kripke structures.
The first essential difficulty that had to be overcome was that calculating the reachability
of states through internal transitions must be done in time proportional to the smaller block that is split off, following
the ``Process the smaller half'' principle. This was done by
two coroutines that are executed in parallel to identify those states that can reach certain bottom states
via internal transitions,
and simultaneously identify the states from which those bottom states are not reachable. As soon as the smaller
set of states was found, the other coroutine was terminated.
The other essential contribution
of \cite{GJKW2017,GrooteW16tacas} is that reassessing stability
of the partition can be done in time proportional to $\log n$ times the number of (incoming and outgoing) transitions of new bottom states.
As each state becomes a bottom state at most once, this fits into the time bound.
Dealing with this was done in a rather delicate post-processing stage that would be executed whenever a new bottom state
was found.
We present an algorithm without postprocessing.

\paragraph{Overview of the article.}
In Section~\ref{sec:preliminaries} we provide the definition of labelled transition systems and branching bisimilarity.
In Section~\ref{sec:high} we provide the core algorithm with high-level data structures, correctness and complexity.
The next section presents the procedure for splitting blocks,
which can be presented as an independent pair of coroutines.
Section~\ref{sect:implementation} provides the details of the algorithm,
especially how the high-level data structures can be represented efficiently.
Section~\ref{sec:benchmarks} presents some benchmarks.




\section{Branching bisimilarity}
\label{sec:preliminaries}
In this section we define labelled transition systems and branching bisimilarity.


\begin{definition}[Labelled transition system]
\label{dn:lts}
A \emph{labelled transition system} ({LTS}) is a triple
$A=(S,\mathit{Act}, \mathord{\pijl{}})$ where
\begin{enumerate}
\item
$S$ is a finite set of \emph{states.} The number of states is denoted by $n$.
\item
$\mathit{Act}$ is a finite set of actions including the \emph{internal action} $\tau$.
\item
$\mathord{\pijl{}} \subseteq S\times \mathit{Act}\times S$ is
a \emph{transition relation.} The number of transitions is necessarily finite and denoted by $m$.
\end{enumerate}
\end{definition}
\noindent%
It is common to write $t\pijl{a}t'$ for $(t,a,t')\in\mathord{\pijl{}}$. Using a slight abuse of notation we write $t\pijl{a}t' \in T$ instead of $(t,a,t') \in T$ for $T \subseteq \mathord{\pijl{}}$.
We also write $t\pijl{a}T$ for the set of transitions $\{ t \pijl{a} t' \mid t' \in T \}$, and likewise $T \pijl{a}T'$ for the set $\{ t \pijl{a} t' \mid t \in T\ \text{and}\ t' \in T' \}$.
We refer to all actions except $\tau$ as
the visible actions. The transitions labelled with $\tau$ are the invisible or hidden transitions.
If $t\pijl{a} t'$, we say that from $t$, the state $t'$, the action $a$, and the transition $t\pijl{a} t'$ are \emph{reachable.}



\begin{definition}[Branching bisimilarity]
\label{def:branching}
Let $A=(S,\mathit{Act},\mathord{\pijl{}})$ be a labelled transition system.
We call a relation $R\subseteq S\times S$
a \emph{branching bisimulation relation} iff it is symmetric and
for all $s,t\in S$ such that $\R{s}{t}$
and all transitions $s \pijl{a} s'$ we have:
\begin{enumerate}
\item
$a=\tau$ and $\R{s'}{t}$, or
\item	\label{def:branching:simulating-transition}
there is a sequence
$t\pijl{\tau}\cdots\pijl{\tau} t'\pijl{a}t''$
such that $\R{s}{t'}$ and $\R{s'}{t''}$.
\end{enumerate}
Two states $s$ and $t$ are \emph{branching bisimilar,}
 denoted by $s\bbis t$,
iff there is a branching bisimulation relation $R$ such that
$\R{s}{t}$.
\end{definition}
Note that branching bisimilarity is an equivalence relation.
Given an equivalence relation $R$, a transition $s \pijl{a} t$ is called \textit{inert} iff $a = \tau$ and $\R{s}{t}$.
If $t \pijl{\tau} t_1 \pijl{\tau} \cdots \pijl{\tau} t_{n-1} \pijl{\tau} t_n \pijl{a} t'$ such that $\R{t}{t_i}$ for $1 \leq i \leq n$, we say that the state $t_n$, the action $a$, and the transition $t_n \pijl{a} t'$ are \emph{inertly reachable} from $t$.

\section{The algorithm}
\label{sec:high}

We now present our algorithm to calculate branching bisimilarity at an abstract level
and assign time budgets, i.e., indications of how much time each step is allowed to take.
The details of the implementation, which are essential to fit the time budgets, are given in Section~\ref{sect:implementation}.
We first describe the basic data structures and subsequently the algorithm, its correctness and complexity.
The algorithm depends on a block splitting procedure, which is explained in Section~\ref{sec:splitting_blocks}.
In this and the following sections, we use the labelled transition system $A=(S,\mathit{Act},\mathord{\pijl{}})$.

\subsection{The essential data types}
The algorithm relies heavily on partitions of sets, especially, sets of states and sets of transitions.

\begin{definition}[Partition]
For a set $X$ a partition $\Pi$ of $X$ is a disjoint cover of $X$, i.e., $\Pi=\{ B_i \subseteq X \mid B_i \not= \emptyset, 1\leq i\leq k\}$
such that $B_i\cap B_j=\emptyset$ for all $1\leq i<j\leq k$ and $X =\bigcup_{1\leq i\leq k}B_i$.

A partition $\Pi'$ is a \textit{refinement} of $\Pi$ iff for every $B' \in \Pi'$ there is some $B \in \Pi$ such that $B' \subseteq B$.
\end{definition}
\noindent%
Note that a partition induces an equivalence relation in the following way:
$s \equiv_\Pi t$ iff there is some $B \in \Pi$ containing both $s$ and $t$.

The algorithm uses two main partitions.
Partition $\Pi_s$ is a partition of states in $S$ that is coarser than branching bisimilarity.
We refer to the elements of $\Pi_s$ as \textit{blocks,} typically denoted using the letter $B$.
Partition $\Pi_t$ partitions the non-inert transitions of $\pijl{}$, where inertness of $\tau$-transitions is interpreted with respect to $\equiv_{\Pi_s}$.
We refer to the elements of $\Pi_t$ as \textit{bunches,} typically denoted using the letter $T$.

The partition of states $\Pi_s$ records the current knowledge about branching bisimilarity:
two states are in different blocks iff the algorithm has found a proof that they are not branching bisimilar
(see Invariant~\ref{inv:bbisim}).
The main idea behind the algorithm is to iteratively refine $\Pi_s$
until it induces a branching bisimulation relation.
In each iteration, a set of transitions with label $a$ and target block $B'$ is selected as the \emph{splitter,}
and for each block $B$ it is determined whether a strict subset of its states can silently reach a transition in the splitter.
If this is the case, $B$ is split based on this criterion.

The partition of transitions $\Pi_t$ records the current knowledge about splitters.
For each bunch of transitions, we initially assume that they can pairwise simulate each other (i.e.,\@ they can serve as transitions $s \pijl{a} s'$ and $t' \pijl{a} t''$ in Definition~\ref{def:branching}), until proven otherwise.
The algorithm maintains the invariant (formalised in Invariant~\ref{inv:bunches}) that
whenever a state in a block has a transition in a bunch,
then every state in that block can inertly reach a transition in the same bunch,
such that the condition of Definition~\ref{def:branching} is satisfied.
Whenever a particular block and bunch satisfy this invariant,
we say that the block is \emph{stable} with respect to that bunch.

However, Definition~\ref{def:branching} comes with some constraints on transition $t' \pijl{a} t''$:
first, it needs to have the same action label as $s \pijl{a} s'$,
and second, $s'$ and $t''$ need to be in the same target block.
If these constraints are violated,
it appears that our initial assumption about the bunch was incorrect,
and we try to correct this by splitting the bunch into parts
that do satisfy the constraints.
We do this by splitting off an \emph{action-block-slice,} a subset of transitions in a bunch with the same action label and the same target block and placing it in a new bunch.
The new bunch is called the \emph{primary} splitter,
the remainder of the bunch from which it was split off is called the \emph{secondary} splitter.
After such a split, the blocks in $\Pi_s$ need to be split with respect to both splitters to re-establish the invariant.

The algorithm uses a number of notions derived from the two partitions $\Pi_s$ and $\Pi_t$.
For bunches $T\in \Pi_t$, actions $a \in \mathit{Act}$ and blocks $B, B' \in \Pi_s$, we have:
\begin{itemize}
\item	The \textit{block-bunch-slices,} i.e., the transitions in $T$ that start in $B$:
$
\TB = \{ s \pijl{b} s' \in T \mid s \in B \}
$.
\item	The \textit{action-block-slices,} i.e., the transitions in $T$ that have label $a$ and end in $B'$:
$
\TaB = \{ s \pijl{a} s' \in T \mid s' \in B' \}
$.
\item	A block-bunch-slice intersected with an action-block-slice:\\
\hspace*{1cm}$\TBaB = \TB \cap \TaB = \{ s \pijl{a} s' \in T \mid s \in B \wedge s' \in B' \}$.
\item	The \textit{bottom states} of a block, i.e., the states without an outgoing inert transition:\\
\hspace*{1cm}$
\Bottom(B) = \{ s \in B \mid \neg\exists s' \in B . s \pijl{\tau} s' \}
$.
\item	The states in a block with a transition in a bunch:
$
\BT = \{ s \mid s \pijl{a} s' \in \TB \}
$.
\item	The blocks splittable by an action-block-slice:\\
\hspace*{1cm}
$\splittableBlocks(\TaB) = \{ B \in \Pi_s \mid  \emptyset \subset \TBaB \subset \TB \}$.
\item	The number of action-block-slices contained in a bunch:\\
\hspace*{1cm}$\nraB(T) = \left| \{ \TaB \mid a \in \mathit{Act}, B' \in
\Pi_s \} \setminus \{ \emptyset \} \right|$.
\item	If $B, B' \in \Pi_s$ and $B \neq B'$, then $B \pijl{\tau} B'$ are the non-inert $\tau$-transitions between $B$ and $B'$.
\item	The outgoing transitions of a block:
$B_{\pijl{}}=\{ s\pijl{a}s'\mid s\in B, a\in\mathit{Act} \text{ and } s'\in S\}$.
\item	The incoming transitions of a block:
$B_{\lijp{}}=\{ s\pijl{a}s'\mid s\in S, a\in\mathit{Act} \text{ and } s'\in B\}$.
\end{itemize}

The first two of these sets (block-bunch-slices and action-block-slices) are maintained as auxiliary data structures in the algorithm in order to meet the required performance bounds.
If the partitions $\Pi_s$ or $\Pi_t$ are adapted,
these derived sets also change. For instance if a block in $\Pi_s$ is replaced by two other blocks
(this happens at Lines~\ref{alg-split1} and \ref{alg-split2} of Algorithm~\ref{algo:main-algorithm-abstract}),
the corresponding block-bunch-slices and action-block-slices are split as well.
For the sake of brevity, we omit the updating of these derived sets in the high-level description of the algorithm. We describe how these sets are maintained in Section~\ref{sect:implementation}.

When $\Pi_t$ is changed, the invariant needs to be re-established by splitting blocks.
To keep track of the blocks that still need to be split, we partition the block-bunch-slices into \textit{stable} and \textit{unstable} block-bunch-slices.
A block-bunch-slice $\TB$ is stable if we have ensured that it is not a splitter for any block in $\Pi_s$.
Otherwise it is deemed unstable, and it needs to be checked whether it is stable, or whether the block $B$ must be split.
If a block-bunch-slice is unstable,
it is stored in the \emph{splitter list,} either as a \textit{primary} or as a \textit{secondary} splitter.
Moreover, if a block-bunch-slice is unstable, we divide the transitions in this block-bunch-slice in
\textit{marked} and \textit{non-marked} transitions.
These markings are used to determine
which bottom states in a block have a transition in a particular bunch, and are essential for efficient splitting of blocks.
For an unstable block-bunch-slice $\TB$ we write its marked transitions as $\Marked(\TB)$.
Note that, even when a block-bunch-slice $\TB$ resides on the splitter list, the block $B$ may be split due to another splitter. In such a case, the block-bunch-slice is split accordingly, and the splitter list is implicitly adapted.

\subsection{Overview of the algorithm}

Before performing the partition refinement,
the LTS is preprocessed to contract $\tau$-strongly connected components (SCCs) into a single state without $\tau$-loop.
This step is valid since all states in a $\tau$-SCC are branching bisimilar,
and it ensures that all $\tau$-paths in the LTS are finite.

The algorithm itself is a partition refinement algorithm. It iteratively refines the partitions $\Pi_s$ and $\Pi_t$.
The main objective for the algorithm is to guarantee the following:
If a block $B$ has a transition in an action-block-slice $\TaB$, then every bottom state in $B$ has a transition in the same action-block-slice $\TaB$.
Since all infinite $\tau$-paths in the LTS have been removed by the above preprocessing step,
every state in $B$ is guaranteed to inertly reach a bottom state,
hence every state in $B$ can inertly reach a transition in $\TaB$.
Therefore, whenever this objective has been reached, every block is a branching bisimulation equivalence class.

To achieve this objective, the algorithm maintains the following weaker invariant,
which keeps track of bunches instead of action-block-slices.
If a block $B$ has a transition to some block $B'$ in a bunch $T$,
then every bottom state in $B$ has a transition in the same bunch $T$.
Observe that, once every action-block-slice is in its own bunch, and this is reflected in the blocks, the main objective of the algorithm has been reached.

Hence, the main invariant of our algorithm is the following.
\begin{invariant}[Bunches]
	\label{inv:bunches}
	$\Pi_s$ is stable under $\Pi_t$, i.e.,
	if a bunch $T \in \Pi_t$ contains a transition with its source state in a block $B$ of $\Pi_s$, then every bottom state in block $B$ has a transition in bunch
	$T$ (in fact, in block-bunch-slice $\TB$).
\end{invariant}
Now, if a bunch contains multiple action-block-slices---we call that a \emph{non-trivial} bunch---, to get closer to our main objective, we have to refine $\Pi_t$ by splitting off an action-block-slice.
At the same time, to preserve the main invariant, when we split a bunch, we need to reflect this change in the blocks. Therefore, the blocks that had an inertly reachable transition in the original bunch are split into subblocks that can either inertly reach the new bunch, the remainder of the original bunch, or both.
This idea is fleshed out in Algorithm~\ref{algo:main-algorithm-abstract}.
We first illustrate one step of the algorithm using the example in Figure~\ref{fig:example-highlevel}.

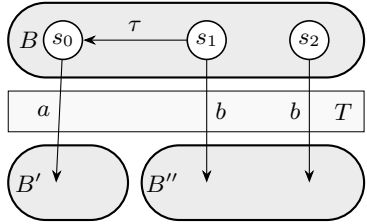
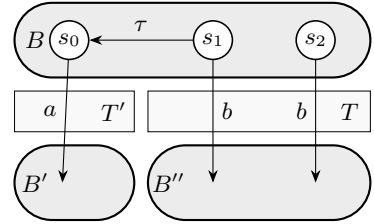
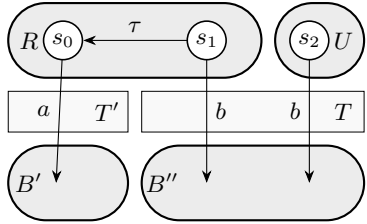
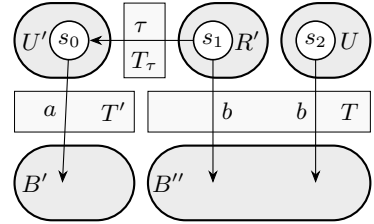
\begin{figure}[ht]
	\small
	\centering
	\begin{subfigure}[t]{.47\linewidth}
	\centering
	\begin{tikzpicture}[>=Stealth,scale=0.9]
		\draw[block]
			(0,0) rectangle (5.2,1.1)
			(0,-1) rectangle (1.75,-2.1)
			(1.95,-1) rectangle (5.2,-2.1)
			;
		\node at(0.3,0.55){$B$};
		\node at(0.3,-1.55){$B'$};
		\node at(2.25,-1.55){$B''$};

		\draw[bunch]
			(0,-0.8) rectangle (5.2,-0.2)
			;

		\node at(4.9,-0.5) {$T$};

		\node     [state](sr0) at (0.8,0.55){\makebox[\statetextwidth]{$s_0$}};
		\node     [state](sr1) at (2.9,0.55){\makebox[\statetextwidth]{$s_1$}};
		\node     [state](sr2) at (4.4,0.55){\makebox[\statetextwidth]{$s_2$}};

		\begin{scope}[->,line cap=rect]
			\draw (sr0) edge node[left,yshift=4pt]{$a$} (0.7,-1.55);
			\draw (sr1) edge node[above]{$\tau$} (sr0);
			\draw (sr1) edge node[right,yshift=4pt]{$b$} (2.9,-1.55);
			\draw (sr2) edge node[left,yshift=4pt]{$b$} (4.4,-1.55);
		\end{scope}
	\end{tikzpicture}
	\subcaption{Part of a labelled transition system with three states, $s_0$, $s_1$, $s_2$, divided into three blocks $B$, $B'$ and $B''$, and a single non-trivial bunch $T$.\label{fig:example-highlevel-start}}
	\end{subfigure}
	\hfill
	\begin{subfigure}[t]{.47\linewidth}
	\centering
	\begin{tikzpicture}[>=Stealth,scale=0.9]
		\draw[block]
			(0,0) rectangle (5.2,1.1)
			(0,-1) rectangle (1.75,-2.1)
			(1.95,-1) rectangle (5.2,-2.1)
			;
		\node at(0.3,0.55){$B$};
		\node at(0.3,-1.55){$B'$};
		\node at(2.25,-1.55){$B''$};

		\draw[bunch]
			(0,-0.8) rectangle (1.75,-0.2)
			(1.95,-0.8) rectangle (5.2,-0.2)
			;

		\node at(4.9,-0.5) {$T$};
		\node at(1.45,-0.5) {$T'$};

		\node     [state](sr0) at (0.8,0.55){\makebox[\statetextwidth]{$s_0$}};
		\node     [state](sr1) at (2.9,0.55){\makebox[\statetextwidth]{$s_1$}};
		\node     [state](sr2) at (4.4,0.55){\makebox[\statetextwidth]{$s_2$}};

		\begin{scope}[->,line cap=rect]
			\draw (sr0) edge node[left,yshift=4pt]{$a$} (0.7,-1.55);
			\draw (sr1) edge node[above]{$\tau$} (sr0);
			\draw (sr1) edge node[right,yshift=4pt]{$b$} (2.9,-1.55);
			\draw (sr2) edge node[left,yshift=4pt]{$b$} (4.4,-1.55);
		\end{scope}
	\end{tikzpicture}
	\subcaption{Situation after moving a small action-block-slice $T_{\pijl{a}B'}$ to its own bunch $T'$.\label{fig:example-highlevel-split-bunch}}
	\end{subfigure}

	\medskip

	\begin{subfigure}[t]{.47\linewidth}
	\centering
	\begin{tikzpicture}[>=Stealth,scale=0.9]
		\draw[block]
			(0,0) rectangle (3.7,1.1)
			(3.9,0) rectangle (5.2,1.1)
			(0,-1) rectangle (1.75,-2.1)
			(1.95,-1) rectangle (5.2,-2.1)
			;
		\node at(0.3,0.55){$R$};
		\node at(4.9,0.55){$U$};
		\node at(0.3,-1.55){$B'$};
		\node at(2.25,-1.55){$B''$};

		\draw[bunch]
			(0,-0.8) rectangle (1.75,-0.2)
			(1.95,-0.8) rectangle (5.2,-0.2)
			;

		\node at(4.9,-0.5) {$T$};
		\node at(1.45,-0.5) {$T'$};

		\node     [state](sr0) at (0.8,0.55){\makebox[\statetextwidth]{$s_0$}};
		\node     [state](sr1) at (2.9,0.55){\makebox[\statetextwidth]{$s_1$}};
		\node     [state](sr2) at (4.4,0.55){\makebox[\statetextwidth]{$s_2$}};

		\begin{scope}[->,line cap=rect]
			\draw (sr0) edge node[left,yshift=4pt]{$a$} (0.7,-1.55);
			\draw (sr1) edge node[above]{$\tau$} (sr0);
			\draw (sr1) edge node[right,yshift=4pt]{$b$} (2.9,-1.55);
			\draw (sr2) edge node[left,yshift=4pt]{$b$} (4.4,-1.55);
		\end{scope}
	\end{tikzpicture}
	\subcaption{
		Situation after splitting $B$ with respect to~$T'$.\label{fig:example-highlevel-split-block-vs-splitter}
	}
	\end{subfigure}
	\hfill
	\begin{subfigure}[t]{.47\linewidth}
	\centering
	\begin{tikzpicture}[>=Stealth,scale=0.9]
		\draw[block]
			(0,0) rectangle (1.4,1.1)
			(2.4,0) rectangle (3.7,1.1)
			(3.9,0) rectangle (5.2,1.1)
			(0,-1) rectangle (1.75,-2.1)
			(1.95,-1) rectangle (5.2,-2.1)
			;
		\node at(0.3,0.55){$U'$};
		\node at(3.4,0.55){$R'$};
		\node at(4.9,0.55){$U$};
		\node at(0.3,-1.55){$B'$};
		\node at(2.25,-1.55){$B''$};

		\draw[bunch]
			(0,-0.8) rectangle (1.75,-0.2)
			(1.95,-0.8) rectangle (5.2,-0.2)
			(1.6,0) rectangle (2.2,1.1)
			;

		\node at(4.9,-0.5) {$T$};
		\node at(1.45,-0.5) {$T'$};
		\node at(1.9,0.25) {$T_{\tau}$};

		\node     [state](sr0) at (0.8,0.55){\makebox[\statetextwidth]{$s_0$}};
		\node     [state](sr1) at (2.9,0.55){\makebox[\statetextwidth]{$s_1$}};
		\node     [state](sr2) at (4.4,0.55){\makebox[\statetextwidth]{$s_2$}};

		\begin{scope}[->,line cap=rect]
			\draw (sr0) edge node[left,yshift=4pt]{$a$} (0.7,-1.55);
			\draw (sr1) edge node[above]{$\tau$} (sr0);
			\draw (sr1) edge node[right,yshift=4pt]{$b$} (2.9,-1.55);
			\draw (sr2) edge node[left,yshift=4pt]{$b$} (4.4,-1.55);
		\end{scope}
	\end{tikzpicture}
	\subcaption{
		Situation after splitting $R$ with respect to~$T$.\label{fig:example-highlevel-split-block-vs-cosplitter}
	}
	\end{subfigure}

\caption{One step in the branching bisimulation algorithm}\label{fig:example-highlevel}
\end{figure}

\begin{example}
We start with a part of a labelled transition system in Figure~\ref{fig:example-highlevel-start}. We have three states, $s_0$, $s_1$, and $s_2$, with an inert transition $s_1 \pijl{\tau} s_0$, and a number of transitions in the bunch $T$. Note that $T$ is non-trivial: it contains two action-block-slices, $T_{\pijl{a}B'}$ and $T_{\pijl{b}B''}$.

To get closer to the situation where every bunch consists of exactly one action-block slice, the algorithm splits off a small action-block-slice from bunch $T$, and puts it into its own bunch $T'$. This situation is shown in Figure~\ref{fig:example-highlevel-split-bunch}.

To ensure that the invariant is maintained, we now need to split block $B$ with respect to $T'$ and $T$. We first split with respect to $T'$. This splits block $B$ into two blocks: $R$, which contains the states that can inertly reach a transition in $T'$, and $U$, containing those states from which $T'$ is unreachable inertly. This puts $s_0$ and $s_1$ in $R$, and $s_2$ in $U$. This is shown in Figure~\ref{fig:example-highlevel-split-block-vs-splitter}.

Block $U$ is stable with respect to $T$ since, according to the invariant that holds before splitting the bunch, all states that cannot inertly reach $T'$ must be able to inertly reach $T$.
Block $R$, however, is not yet stable with respect to $T$.
Therefore, $R$ must be split into the set of states from which a transition in $T$ is unreachable, denoted $U'$, and from which a transition in $T$ is reachable, denoted $R'$. This is shown in Figure~\ref{fig:example-highlevel-split-block-vs-cosplitter}.

Note that in this last split, the $\tau$-transition becomes non-inert, and a new bunch $T_\tau$ containing this transition is introduced. The algorithm needs to stabilise with respect to this new bunch, and it needs to ensure that $R'$, the block containing a new bottom state, is stabilised with respect to all other bunches it can inertly reach. In this particular example, $R'$ happens to be stable with respect to all bunches, so in this case, the required stabilisation has no effect.
\end{example}


\subsection{In-depth description of the algorithm}

\begin{algorithm}[!t]
	\caption{Abstract algorithm for branching bisimulation partitioning\label{algo:main-algorithm-abstract}}
	\algsetup{indent=0.85cm}
	\parbox{0.79\textwidth}{
	\small\begin{algorithmic}[1]
	\STATE	Find $\tau$-SCCs and contract each of them to a single state
		\label{alg-line-contract-SCCs}
		\rightbrace{4}{$\bigo{m}$}%
	\STATE $B_\text{vis} := \{ s \in S \mid s \text{ can inertly reach some } s' \pijl{a} s'' \}$;
	\label{alg-line-initialize-partitions-begin} \quad
		$B_\text{invis} := S \setminus B_\text{vis}$
	\STATE	$\Pi_s := \{ B_\text{vis}, B_\text{invis} \} \setminus \{ \emptyset \}$
		\label{alg-line-initialize-Pi_s-end}
	\STATE	$\Pi_t := \{ \{ s \pijl{a} s' \mid a \in \mathit{Act} \setminus \{\tau\}, s,s' \in S \} \cup B_\text{vis} \pijl{\tau} B_\text{invis} \}$
		\label{alg-line-initialize-partitions-end}
	\WHILE{a $T \in \Pi_t$ exists with $\nraB(T) > 1$} \label{algo:main:outerloop}
		\STATE	Select some $a \in \mathit{Act}$ and $B' \in \Pi_s$ such that  $\size{\TaB} \leq \frac{1}{2}\size{T}$
			\label{alg-line-select-small-bunch}
			\rightbrace[-1]{1}{$\leq m$ iterations}
		\STATE	$\Pi_t := (\Pi_t \setminus \{ T \}) \cup \{ \TaB, T \setminus \TaB \}$\label{alg:split_transitions}
		\rightbrace[-1]{9}{$\bigo{\size{\TaB}}$}%
		\FORALL {$B \in \splittableBlocks(\TaB)$\label{alg-line-markloop-begin}}
			\STATE	Add first $\TBaB$ and then $\TB \setminus \TBaB$ to the splitter list%
				\label{alg-line-find-splittable-block}. Label $\TBaB$ primary and $\TB \setminus \TBaB$ secondary%
					\STATE  Mark all transitions in $\TBaB$\label{alg-line-mark-all-transitions-in-primary}
			\STATE	For every state with both marked outgoing transitions and an outgoing transition in $\TB \setminus \TBaB$, mark one such transition
				\label{alg-line-mark-one-transition-in-secondary-splitter}
		\ENDFOR\label{alg-line-markloop-end}
		\FORALL	{$\TB'$ in the splitter list (in order)\label{alg-line-stabilize-begin}}
			\STATE	\label{alg:line:callsplit1}%
				$\langle R, U \rangle := \mathsf{split}(B, \TB')$
				\COMMENT{$R \supseteq B_{\pijl{T'}}$ can reach $\TB'$ and}%
			\item[]	\mbox{}\COMMENT{\makebox[0cm][l]{$U = B \setminus R$ cannot reach it}\phantom{$R \supseteq B_{\pijl{T'}}$ can reach $\TB'$ and}}%
			\STATE  Remove $\TB' = T'_{R \pijl{}}$ from the splitter list\label{alg:removeTBpijl}
				\rightbrace[-3]{1}{$\leq \size{\TaB}$ iterations}
			\STATE	$\Pi_s := (\Pi_s \setminus \{ B \}) \cup (\{ R, U \} \setminus \{ \emptyset \})$\label{alg-split1}
				\rightbrace[-3]{10}{\parbox[l]{0.3\textwidth}{%
					$\bigO(\size{\Marked(\TB')} + \mbox{}$ \\
					$\phantom{\bigO(\!}\size{U_{\pijl{}}} + \size{U_{\lijp{}}} + \mbox{}$ \phantom{)}\\    
					$\phantom{\bigO(\!}\size{\Bottom(N)_{\pijl{}}})$ \phantom{)}\\   
					or $\bigo{\size{R_{\pijl{}}} + \size{R_{\lijp{}}}}$}}%
			\IF	{$\TB'$ is primary (note: $\TB'= \TBaB$)}
					\STATE	Remove $T_{U\pijl{}} \setminus T_{U\pijl{a} B'}$ from the splitter list\label{alg:line-remove-TU-after-primary-split}
			\ENDIF
			\IF	{$R \pijl{\tau} U \neq \emptyset$}\label{alg:inert_condition}
				\STATE	Create a new bunch containing exactly $R \mathbin{\pijl{\tau}} U$,
						add $R \mathbin{\pijl{\tau}} U = (R \pijl{\tau} U)_{R \pijl{}}$ to the splitter list,
						and mark all its transitions\label{alg:new_bunch}
				\STATE	$\langle N, R' \rangle :=\mathsf{split}(R, R \pijl{\tau} U)$
							\label{alg:split-call2}
					\COMMENT{$N \supseteq R_{\pijl{R \pijl{\tau} U}}$ contains}%
				\item[]	\mbox{}\COMMENT{\makebox[0cm][l]{the new bottom states}\phantom{$N \supseteq R_{\pijl{R \pijl{\tau} U}}$ contains}}%
				\STATE	Remove $R \pijl{\tau} U = (R \pijl{\tau} U)_{N \pijl{}}$ from the splitter list
					\rightbrace[-2]{5}{\parbox[l]{0.3\textwidth}{%
						$\bigO(\size{R \pijl{\tau} U} + \mbox{}$ \\
						$\phantom{\bigO(\!}\size{R'_{\pijl{}}} + \size{R'_{\lijp{}}} + \mbox{}$ \phantom{)}\\   
						$\phantom{\bigO(\!}\size{\Bottom(N)_{\pijl{}}})$ \phantom{)} \\    
						or $\bigo{\size{N_{\pijl{}}} + \size{N_{\lijp{}}}}$}}%
				\STATE	
									$\Pi_s := (\Pi_s \setminus \{ R \}) \cup (\{ N, R' \} \setminus \{ \emptyset \})$\label{alg-split2}
				\STATE	Add $N \pijl{\tau} R'$ to the bunch containing $R \pijl{\tau} U$\label{alg:extend_bunch}
				\STATE	Add all $T_{N \pijl{}}$ to the splitter list and label them secondary
						\label{alg:line-make-all-slices-unstable}
						\rightbrace{1}{$\bigo{\size{\Bottom^{*}(N)_{\pijl{}}}}$}%
				\STATE	For each bottom state, mark one of its outgoing transitions in every $T_{N \pijl{}}$ where it has one
						\label{alg:line-mark-bb-of-new-bottom}
						\rightbrace[-1]{2}{$\bigo{\size{\Bottom(N)_{\pijl{}}}}$}%
			\ENDIF
		\ENDFOR\label{alg-line-stabilize-end}
	\ENDWHILE\label{algo:main:outerloop-end}
	\end{algorithmic}
	}
	\end{algorithm}

The pseudocode of the algorithm is given in Algorithm~\ref{algo:main-algorithm-abstract}.
The algorithm works as follows, where we start with the initialisation.
First, at Line~\ref{alg-line-contract-SCCs},
all $\tau$-SCCs are contracted into a single state each (without $\tau$-loop).
All states in a $\tau$-SCC are branching bisimilar (as they can all reach the same states, possibly after a few inert transitions).
This preprocessing ensures that there are no $\tau$-cycles in the LTS (Invariant~\ref{inv:noloops} below)
and from every non-bottom state a bottom state can be reached via inert transitions (Lemma~\ref{lem::noloop} below).

Second, at Lines~\ref{alg-line-initialize-partitions-begin}--\ref{alg-line-initialize-partitions-end},
we create the initial partitions $\Pi_s$ and $\Pi_t$ as follows.
Let $B_\text{vis}$ be the set of states from which a visible transition is inertly reachable,
and let $B_\text{invis}$ be the other states (i.e., the states from which only a deadlock state can be inertly reached).
Then $\Pi_s = \{ B_\text{vis}, B_\text{invis} \} \setminus \{ \emptyset \}$.
Initially, $\Pi_t$ contains one bunch consisting of all non-inert transitions.
All the block-bunch-slices induced by $\Pi_s$ and $\Pi_t$ are initially stable,
i.e., for all block-bunch-slices $\TB$, every bottom state of $B$, which must be $B_\text{vis}$, has a transition in $\TB$
(this is Invariant~\ref{inv:bunches} above).

If there are non-trivial bunches,
these bunches need to be split such that they ultimately become trivial.
The outer loop of the algorithm (Lines~\ref{algo:main:outerloop}--\ref{algo:main:outerloop-end}) takes a non-trivial bunch $T$ from $\Pi_t$,
and from this it moves a \emph{small} (containing at most half the number of transitions in the bunch) action-block-slice $\TaB$
into its own bunch in $\Pi_t$ (Line~\ref{alg:split_transitions}).
Hence, bunch $T$ is reduced to $T\setminus \TaB$.

The two new bunches $\TaB$ and $T\setminus \TaB$ can cause instability, violating Invariant~\ref{inv:bunches}.
This means there can be blocks with transitions in one new bunch,
but not all bottom states have transitions in that bunch
because some bottom states only have transitions in the other new bunch.
For such blocks, stability needs to be restored by splitting them.
The set $\splittableBlocks(\TaB)$ contains all blocks that have transitions in both new bunches $\TaB$ and $T\setminus \TaB$;
these blocks must be investigated.
All other blocks are stable with respect to the new bunches.

Earlier algorithms would now investigate all blocks to re-establish stability.
Instead, we investigate all block-bunch-slices in the smaller of the two new bunches. All blocks that do not have transitions in these block-bunch-slices are stable with respect to both bunches.
The first inner loop (Lines~\ref{alg-line-markloop-begin}--\ref{alg-line-markloop-end})
serves to insert all these block-bunch-slices into the splitter list.
Block-bunch-slices of the shape $\TBaB$ in the splitter list are labelled primary,
and all other block-bunch-slices in this list are labelled secondary.

This loop also marks some transitions.
The function of this marking is similar to that of the counters in \cite{PT87}:
it serves to determine quickly whether a bottom state has a transition in a secondary splitter $\TB \setminus \TBaB$
(or slices that are the result of splitting this slice).
Invariant~\ref{inv:marked_transitions} below indicates
that a bottom state has transitions in some splitter block-bunch-slice
if and only if it has marked transitions in this slice.
The marked transitions are stored separately in a block-bunch-slice
and therefore can be visited without spending time on unmarked transitions of the block.
This is essential to obtain the time complexity of the algorithm,
as we are allowed to perform one unit of work per transition in $\TaB$
(the smaller of the two new bunches),
and since $\size{\Marked(\TB \setminus \TBaB)} \leq \size{\TBaB}$, we do not mark too many transitions.

In the second loop (Lines~\ref{alg-line-stabilize-begin}--\ref{alg-line-stabilize-end}), one splitter $\TB'$ from the splitter list is
taken at a time
and its source block is split into $R$ (the part that can inertly reach some transition in $\TB'$) and $U$ (the part that cannot inertly reach $\TB'$) to
re-establish stability.
Formally, the routine $\mathsf{split}(B,T)$ delivers the pair $\langle R,U\rangle$ defined by:
\begin{equation}\begin{array}{l}
\label{eqn:splitproperty}
R =\{ s\in B\mid s\pijl{\tau}s_1\pijl{\tau}\cdots\pijl{\tau}s_n\pijl{a}s'\textrm{ where } s_1,\ldots, s_{n} \in B, s_n\pijl{a}s'\in T\},\\
U = B\setminus R.
\end{array}\end{equation}
We detail its algorithm and argue for its correctness in Section~\ref{sec:splitting_blocks}.

If $\TB'$ was a primary splitter of the form $\TBaB$ added to the splitter list at Line~\ref{alg-line-find-splittable-block},
then we know that $U$ must be stable under $T_{U\pijl{}} \setminus T_{U\pijl{a} B'}$, as
every bottom state in $B$ has a transition in the former block-bunch-slice $\TB$,
and as the states in $U$ have no transition in $\TBaB$,
every bottom state in $U$ must have a transition in $\TB \setminus \TBaB$.
Therefore, at Line~\ref{alg:line-remove-TU-after-primary-split}, block-bunch-slice $T_{U\pijl{}} \setminus T_{U\pijl{a} B'}$ can be removed from the splitter list
without investigating it. This is crucial for the complexity, and essentially it is the translation of the three-way split from~\cite{PT87}.

Some invisible transitions may have become non-inert, namely the $\tau$-transitions that go from $R$ to $U$.
There cannot be $\tau$-transitions from $U$ to $R$ as
otherwise, from the source state of such a transition, $\TB'$ could be inertly reached,
so it should have been added to $R$ instead of $U$.
The new non-inert transitions were not yet part of a bunch in $\Pi_t$.
So, a new bunch $R\pijl{\tau}U$ is formed containing these transitions. All transitions in this new bunch leave block
$R$ and therefore $R$ is the only block that may not be stable under this new bunch. To re-establish stability
we split $R$ into blocks $\langle N, R'\rangle$ (Line~\ref{alg:split-call2}). Observe that there can be transitions $N\pijl{\tau}R'$ that also become
non-inert, and we add those to the new bunch $R\pijl{\tau}U$.

In $N$, i.e., the set of states that can inertly reach some transition in $R \pijl{\tau} U$,
some states are new bottom states, while $R'$ contains all original bottom states of $R$.
In accordance with the observations in \cite{GV90} blocks containing new bottom states can become unstable under any
block-bunch-slice. Therefore, stability under all those block-bunch-slices must be re-established and therefore all the block-bunch-slices
leaving block $N$ are put on the splitter list.
We mark one of the transitions in every new bottom state
such that we can find the bottom states with a transition in $T_{N\pijl{}}$ in time proportional to the number of such new bottom states.

This algorithm is repeated until all action-block-slices coincide with the bunches. In the next section we prove that the resulting
partition $\Pi_s$ exactly coincides with branching bisimilarity.

We illustrate the algorithm in the following example. Note that the example also illustrates some of the details of the $\mathsf{split}$ subroutine, which is discussed in detail in Section~\ref{sec:splitting_blocks}.

\begin{figure}[!htp]
	\centering
	\begin{subfigure}[t]{.3\linewidth}
	\centering
	\begin{tikzpicture}[>=Stealth,scale=0.9]
		\draw[block]
			(0,0) rectangle (4.5,2.8);
		\node at(2.9,2.55){$B$};

		\draw[bunch]
			(0,-1.2) rectangle (1.6,-0.2)
			(1.8,-1.2) rectangle (4.5,-0.2)
			(0,3) rectangle (4.5,4);

		\node at(1.1,-0.5) {$\TaB$};
		\node[align=left] at(2.3,-0.7) {$T \setminus \mbox{}$ \\ $\TaB$};
		\node at(0.3,3.7) {$T'$};


		\node     [state](sr0) at (3.95,0.55){\makebox[\statetextwidth]{}};
		\node     [state](sr1) at (2.65,0.55){\makebox[\statetextwidth]{}};
		\node     [state](sr2) at (3.55,2.25){\makebox[\statetextwidth]{}};
		\node     [state](sr3) at (0.55,0.55){\makebox[\statetextwidth]{}};
		\node     [state](sr4) at (2.25,2.25){\makebox[\statetextwidth]{}};
		\node     [state](sr5) at (0.55,2.25){\makebox[\statetextwidth]{}};

		\begin{scope}[->,line cap=rect]
			\draw (sr0) to (4.5,4.4);
			\draw (sr0) to (4.4,-1.8);
			\draw (sr0) to (4.2,-1.7);
			\draw (sr0) to (4,-1.6);
			\draw (sr1) to (3.1,-1.8);
			\draw (sr1) to (3.3,-1.6);
			\draw[out=200,in=90] (sr1) to (1.4,-1.6);
			\draw (sr1) to (sr3);
			\draw (sr2) to (3.5,4.4);
			\draw (sr2) to (sr1);
			\draw (sr2) to (sr3);
			\draw (sr3) to (0.75,-1.6);
			\draw (sr3) to (sr5);
			\draw (sr4) to (sr2);
			\draw (sr4) to (sr3);
			\draw (sr4) to (sr5);
			\draw[out=250,in=110] (sr5) to (0.6,-1.8);
			\draw (sr5) to (1,4.4);

		\end{scope}
	\end{tikzpicture}
	\subcaption{\label{fig:exa:split-a}}
	\end{subfigure}
	\hfill\hfill
	\begin{subfigure}[t]{.3\linewidth}
		\centering
	\begin{tikzpicture}[>=Stealth,scale=0.9]
		\draw[block]
			(0,0) rectangle (4.5,2.8);
		\node at(2.9,2.55){$B$};

		\draw[bunch]
			(0,-1.2) rectangle (1.6,-0.2)
			(1.8,-1.2) rectangle (4.5,-0.2)
			(0,3) rectangle (4.5,4);

		\node at(1.1,-0.5) {$\TaB$};
		\node[align=left] at(2.3,-0.7) {$T \setminus \mbox{}$ \\ $\TaB$};
		\node at(0.3,3.7) {$T'$};


		\node     [blue state](sr0) at (3.95,0.55){\makebox[\statetextwidth]{}};
		\node     [red state](sr1) at (2.65,0.55){\makebox[\statetextwidth]{}};
		\node     [state](sr2) at (3.55,2.25){\makebox[\statetextwidth]{}};
		\node     [red state](sr3) at (0.55,0.55){\makebox[\statetextwidth]{}};
		\node     [state](sr4) at (2.25,2.25){\makebox[\statetextwidth]{}};
		\node     [red state](sr5) at (0.55,2.25){\makebox[\statetextwidth]{}};

		\begin{scope}[->,line cap=rect]
			\draw (sr0) to (4.5,4.4);
			\draw (sr0) to (4.4,-1.8);
			\draw (sr0) to (4.2,-1.7);
			\draw (sr0) to (4,-1.6);
			\draw (sr1) to node[gray,pos=0.615]{\small $m$} (3.1,-1.8);
			\draw (sr1) to (3.3,-1.6);
			\draw[out=200,in=90] (sr1) to node[red,pos=0.7723]{\small $m$} (1.4,-1.6);
			\draw (sr1) to (sr3);
			\draw (sr2) to (3.5,4.4);
			\draw (sr2) to (sr1);
			\draw (sr2) to (sr3);
			\draw (sr3) to node[red,pos=0.68]{\small $m$} (0.75,-1.6);
			\draw (sr3) to (sr5);
			\draw (sr4) to (sr2);
			\draw (sr4) to (sr3);
			\draw (sr4) to (sr5);
			\draw[out=250,in=110] (sr5) to node[red,pos=0.8]{\small $m$} (0.6,-1.8);
			\draw (sr5) to (1,4.4);

		\end{scope}

	\end{tikzpicture}
	\subcaption{\label{fig:exa:split-b}}
	\end{subfigure}
	\hfill\hfill
	\begin{subfigure}[t]{.3\linewidth}
		\centering
	\begin{tikzpicture}[>=Stealth,scale=0.9]
		\draw[block]
			(3.4,0) rectangle (4.5,1.5);
		\node at(3.75,1.1){$U$};

		\draw[block]
			(0,0) -- (0,2.8) [sharp corners] -- (3.55,2.8) arc[radius=0.55cm,start angle=90,end angle=-55.79454]
			-- ++(214.20546:0.4) arc[radius=0.75cm,start angle=124.20546,end angle=180] [rounded corners=4.95mm] -- (3.2,0) -- cycle;
		\node at(1.45,2.55){$R$};

		\draw[bunch]
			(0,-1.2) rectangle (1.6,-0.2)
			(1.8,-1.2) rectangle (4.5,-0.2)
			(0,3) rectangle (4.5,4);

		\node at(1.1,-0.5) {$\TaB$};
		\node[align=left] at(2.3,-0.7) {$T \setminus \mbox{}$ \\ $\TaB$};
		\node at(0.3,3.7) {$T'$};


		\node     [state](sr0) at (3.95,0.55){\makebox[\statetextwidth]{}};
		\node     [red state](sr1) at (2.65,0.55){\makebox[\statetextwidth]{}};
		\node     [state](sr2) at (3.55,2.25){\makebox[\statetextwidth]{}};
		\node     [state](sr3) at (0.55,0.55){\makebox[\statetextwidth]{}};
		\node     [state](sr4) at (2.25,2.25){\makebox[\statetextwidth]{}};
		\node     [blue state](sr5) at (0.55,2.25){\makebox[\statetextwidth]{}};

		\begin{scope}[->,line cap=rect]
			\draw (sr0) to (4.5,4.4);
			\draw (sr0) to (4.4,-1.8);
			\draw (sr0) to (4.2,-1.7);
			\draw (sr0) to (4,-1.6);
			\draw (sr1) to node[red,pos=0.615]{\small $m$} (3.1,-1.8);
			\draw (sr1) to (3.3,-1.6);
			\draw[out=200,in=90] (sr1) to (1.4,-1.6);
			\draw (sr1) to (sr3);
			\draw (sr2) to (3.5,4.4);
			\draw (sr2) to (sr1);
			\draw (sr2) to (sr3);
			\draw (sr3) to (0.75,-1.6);
			\draw (sr3) to (sr5);
			\draw (sr4) to (sr2);
			\draw (sr4) to (sr3);
			\draw (sr4) to (sr5);
			\draw[out=250,in=110] (sr5) to (0.6,-1.8);
			\draw (sr5) to (1,4.4);

		\end{scope}
	\end{tikzpicture}
	\subcaption{\label{fig:exa:split-c}}
	\end{subfigure}

	\bigskip

	\bigskip

	\bigskip

	\begin{subfigure}[t]{.3\linewidth}
		\centering
	\begin{tikzpicture}[>=Stealth,scale=0.9]
		\draw[block]
			(3.4,0) rectangle (4.5,1.5);
		\node at(3.75,1.1){$U$};

		\draw[block]
			(0,0) -- (0,2.8) [sharp corners] -- (3.55,2.8) arc[radius=0.55cm,start angle=90,end angle=-55.79454]
			-- ++(214.20546:0.4) arc[radius=0.75cm,start angle=124.20546,end angle=180] [rounded corners=4.95mm] -- (3.2,0) -- cycle;
		\node at(1.45,2.55){$R$};

		\draw[bunch]
			(0,-1.2) rectangle (1.6,-0.2)
			(1.8,-1.2) rectangle (4.5,-0.2)
			(0,3) rectangle (4.5,4);

		\node at(1.1,-0.5) {$\TaB$};
		\node[align=left] at(2.3,-0.7) {$T \setminus \mbox{}$ \\ $\TaB$};
		\node at(0.3,3.7) {$T'$};


		\node     [state](sr0) at (3.95,0.55){\makebox[\statetextwidth]{}};
		\node     [red state](sr1) at (2.65,0.55){\makebox[\statetextwidth]{}};
		\node     [state](sr2) at (3.55,2.25){\makebox[\statetextwidth]{}};
		\node     [state](sr3) at (0.55,0.55){\makebox[\statetextwidth]{}};
		\node     [state](sr4) at (2.25,2.25){\makebox[\statetextwidth]{\raisebox{-0.3ex}{2}}};
		\node     [blue state](sr5) at (0.55,2.25){\makebox[\statetextwidth]{}};

		\begin{scope}[->,line cap=rect]
			\draw (sr0) to (4.5,4.4);
			\draw (sr0) to (4.4,-1.8);
			\draw (sr0) to (4.2,-1.7);
			\draw (sr0) to (4,-1.6);
			\draw (sr1) to (3.1,-1.8);
			\draw (sr1) to node[red,pos=0.635]{\small\,\checkmark} (3.3,-1.6);
			\draw[out=200,in=90] (sr1) to (1.4,-1.6);
			\draw (sr1) to (sr3);
			\draw (sr2) to (3.5,4.4);
			\draw (sr2) to (sr1);
			\draw (sr2) to (sr3);
			\draw (sr3) to (0.75,-1.6);
			\draw (sr3) to (sr5);
			\draw (sr4) to (sr2);
			\draw (sr4) to (sr3);
			\draw (sr4) to node[blue,pos=0.45]{\small\checkmark} (sr5);
			\draw[out=250,in=110] (sr5) to (0.6,-1.8);
			\draw (sr5) to (1,4.4);

		\end{scope}
	\end{tikzpicture}
	\subcaption{\label{fig:exa:split-d}}
	\end{subfigure}
	\hfill\hfill
	\begin{subfigure}[t]{.3\linewidth}
		\centering
	\begin{tikzpicture}[>=Stealth,scale=0.9]
		\draw[block]
			(3.4,0) rectangle (4.5,1.5);
		\node at(3.75,1.1){$U$};

		\draw[block]
			(0,0) -- (0,2.8) [sharp corners] -- (3.55,2.8) arc[radius=0.55cm,start angle=90,end angle=-55.79454]
			-- ++(214.20546:0.4) arc[radius=0.75cm,start angle=124.20546,end angle=180] [rounded corners=4.95mm] -- (3.2,0) -- cycle;
		\node at(1.45,2.55){$R$};

		\draw[bunch]
			(0,-1.2) rectangle (1.6,-0.2)
			(1.8,-1.2) rectangle (4.5,-0.2)
			(0,3) rectangle (4.5,4);

		\node at(1.1,-0.5) {$\TaB$};
		\node[align=left] at(2.3,-0.7) {$T \setminus \mbox{}$ \\ $\TaB$};
		\node at(0.3,3.7) {$T'$};


		\node     [state](sr0) at (3.95,0.55){\makebox[\statetextwidth]{}};
		\node     [red state](sr1) at (2.65,0.55){\makebox[\statetextwidth]{}};
		\node     [red state](sr2) at (3.55,2.25){\makebox[\statetextwidth]{}};
		\node     [blue state](sr3) at (0.55,0.55){\makebox[\statetextwidth]{\raisebox{-0.3ex}{0}}};
		\node     [red state](sr4) at (2.25,2.25){\makebox[\statetextwidth]{\raisebox{-0.3ex}{2}}};
		\node     [blue state](sr5) at (0.55,2.25){\makebox[\statetextwidth]{}};

		\begin{scope}[->,line cap=rect]
			\draw (sr0) to (4.5,4.4);
			\draw (sr0) to (4.4,-1.8);
			\draw (sr0) to (4.2,-1.7);
			\draw (sr0) to (4,-1.6);
			\draw (sr1) to (3.1,-1.8);
			\draw (sr1) to node[red,pos=0.635]{\small\,\checkmark} (3.3,-1.6);
			\draw[out=200,in=90] (sr1) to (1.4,-1.6);
			\draw (sr1) to node[blue,pos=0.45]{\small\checkmark} (sr3);
			\draw (sr2) to (3.5,4.4);
			\draw (sr2) to node[red,pos=0.47]{\small\checkmark} (sr1);
			\draw (sr2) to (sr3);
			\draw (sr3) to (0.75,-1.6);
			\draw (sr3) to node[blue,pos=0.45]{\small\checkmark} (sr5);
			\draw (sr4) to node[red,pos=0.45]{\small\checkmark} (sr2);
			\draw (sr4) to (sr3);
			\draw (sr4) to node[blue,pos=0.45]{\small\checkmark} (sr5);
			\draw[out=250,in=110] (sr5) to (0.6,-1.8);
			\draw (sr5) to (1,4.4);

		\end{scope}
	\end{tikzpicture}
	\subcaption{\label{fig:exa:split-e}}
	\end{subfigure}
	\hfill\hfill
	\begin{subfigure}[t]{.3\linewidth}
		\centering
	\begin{tikzpicture}[>=Stealth,scale=0.9]
		\draw[block]
			(0,0) rectangle (1.1,2.8)
			(3.4,0) rectangle (4.5,1.5);
		\node at(0.85,1.4){$U_1$};
		\node at(3.75,1.1){$U$};

		\draw[block]
			(1.7,0) -- (1.7,2.8) [sharp corners] -- (3.55,2.8) arc[radius=0.55cm,start angle=90,end angle=-55.79454]
			-- ++(214.20546:0.4) arc[radius=0.75cm,start angle=124.20546,end angle=180] [rounded corners=4.95mm] -- (3.2,0) -- cycle;
		\node at(2.9,2.55){$R_1$};

		\draw[bunch]
			(1.2,0) rectangle (1.6,2.8)
			(0,-1.2) rectangle (1.6,-0.2)
			(1.8,-1.2) rectangle (4.5,-0.2)
			(0,3) rectangle (4.5,4);

		\node at(1.1,-0.5) {$\TaB$};
		\node[align=left] at(2.3,-0.7) {$T \setminus \mbox{}$ \\ $\TaB$};
		\node at(0.3,3.7) {$T'$};


		\node     [state](sr0) at (3.95,0.55){\makebox[\statetextwidth]{}};
		\node     [state](sr1) at (2.65,0.55){\makebox[\statetextwidth]{\raisebox{-0.3ex}{nb}}};
		\node     [state](sr2) at (3.55,2.25){\makebox[\statetextwidth]{}};
		\node     [state](sr3) at (0.55,0.55){\makebox[\statetextwidth]{}};
		\node     [state](sr4) at (2.25,2.25){\makebox[\statetextwidth]{}};
		\node     [state](sr5) at (0.55,2.25){\makebox[\statetextwidth]{}};

		\begin{scope}[->,line cap=rect]
			\draw (sr0) to (4.5,4.4);
			\draw (sr0) to (4.4,-1.8);
			\draw (sr0) to (4.2,-1.7);
			\draw (sr0) to (4,-1.6);
			\draw (sr1) to (3.1,-1.8);
			\draw (sr1) to (3.3,-1.6);
			\draw[out=200,in=90] (sr1) to (1.4,-1.6);
			\draw (sr1) to node[red,pos=0.632]{\small $m$} (sr3);
			\draw (sr2) to (3.5,4.4);
			\draw (sr2) to (sr1);
			\draw (sr2) to node[red,pos=0.762]{\small $m$} (sr3);
			\draw (sr3) to (0.75,-1.6);
			\draw (sr3) to (sr5);
			\draw (sr4) to (sr2);
			\draw (sr4) to node[red,pos=0.5]{\small $m$} (sr3);
			\draw (sr4) to node[red,pos=0.5]{\small $m$} (sr5);
			\draw[out=250,in=110] (sr5) to (0.6,-1.8);
			\draw (sr5) to (1,4.4);

		\end{scope}
	\end{tikzpicture}
	\subcaption{\label{fig:exa:split-f}}
	\end{subfigure}

	\bigskip

	\bigskip

	\bigskip

	\begin{subfigure}[t]{.3\linewidth}
		\centering
	\begin{tikzpicture}[>=Stealth,scale=0.9]
		\draw[block]
			(0,0) rectangle (1.1,2.8)
			(3.4,0) rectangle (4.5,1.5);
		\node at(0.85,1.4){$U_1$};
		\node at(3.75,1.1){$U$};

		\draw[block]
			(1.7,0) -- (1.7,2.8) [sharp corners] -- (3.55,2.8) arc[radius=0.55cm,start angle=90,end angle=-55.79454]
			-- ++(214.20546:0.4) arc[radius=0.75cm,start angle=124.20546,end angle=180] [rounded corners=4.95mm] -- (3.2,0) -- cycle;
		\node at(2.9,2.55){$R_1$};

		\draw[bunch]
			(1.2,0) rectangle (1.6,2.8)
			(0,-1.2) rectangle (1.6,-0.2)
			(1.8,-1.2) rectangle (4.5,-0.2)
			(0,3) rectangle (4.5,4);

		\node at(1.1,-0.5) {$\TaB$};
		\node[align=left] at(2.3,-0.7) {$T \setminus \mbox{}$ \\ $\TaB$};
		\node at(0.3,3.7) {$T'$};


		\node     [state](sr0) at (3.95,0.55){\makebox[\statetextwidth]{}};
		\node     [state](sr1) at (2.65,0.55){\makebox[\statetextwidth]{\raisebox{-0.3ex}{nb}}};
		\node     [state](sr2) at (3.55,2.25){\makebox[\statetextwidth]{}};
		\node     [state](sr3) at (0.55,0.55){\makebox[\statetextwidth]{}};
		\node     [state](sr4) at (2.25,2.25){\makebox[\statetextwidth]{}};
		\node     [state](sr5) at (0.55,2.25){\makebox[\statetextwidth]{}};

		\begin{scope}[->,line cap=rect]
			\draw (sr0) to (4.5,4.4);
			\draw (sr0) to (4.4,-1.8);
			\draw (sr0) to (4.2,-1.7);
			\draw (sr0) to (4,-1.6);
			\draw (sr1) to node[gray,pos=0.615]{\small $m$} (3.1,-1.8);
			\draw (sr1) to (3.3,-1.6);
			\draw[out=200,in=90] (sr1) to node[gray,pos=0.7723]{\small $m$} (1.4,-1.6);
			\draw (sr1) to (sr3);
			\draw (sr2) to (3.5,4.4);
			\draw (sr2) to (sr1);
			\draw (sr2) to (sr3);
			\draw (sr3) to (0.75,-1.6);
			\draw (sr3) to (sr5);
			\draw (sr4) to (sr2);
			\draw (sr4) to (sr3);
			\draw (sr4) to (sr5);
			\draw[out=250,in=110] (sr5) to (0.6,-1.8);
			\draw (sr5) to (1,4.4);

		\end{scope}
	\end{tikzpicture}
	\subcaption{\label{fig:exa:split-g}}
	\end{subfigure}
	\hfill\hfill
	\begin{subfigure}[t]{.3\linewidth}
		\centering
	\begin{tikzpicture}[>=Stealth,scale=0.9]
		\draw[block]
			(0,0) rectangle (1.1,2.8)
			(3.4,0) rectangle (4.5,1.5);
		\node at(0.85,1.4){$U_1$};
		\node at(3.75,1.1){$U$};

		\draw[block]
			(1.7,0) -- (1.7,2.8) [sharp corners] -- (3.55,2.8) arc[radius=0.55cm,start angle=90,end angle=-55.79454]
			-- ++(214.20546:0.4) arc[radius=0.75cm,start angle=124.20546,end angle=180] [rounded corners=4.95mm] -- (3.2,0) -- cycle;
		\node at(2.9,2.55){$R_1$};

		\draw[bunch]
			(1.2,0) rectangle (1.6,2.8)
			(0,-1.2) rectangle (1.6,-0.2)
			(1.8,-1.2) rectangle (4.5,-0.2)
			(0,3) rectangle (4.5,4);

		\node at(1.1,-0.5) {$\TaB$};
		\node[align=left] at(2.3,-0.7) {$T \setminus \mbox{}$ \\ $\TaB$};
		\node at(0.3,3.7) {$T'$};


		\node     [state](sr0) at (3.95,0.55){\makebox[\statetextwidth]{}};
		\node     [blue state](sr1) at (2.65,0.55){\makebox[\statetextwidth]{}};
		\node     [state](sr2) at (3.55,2.25){\makebox[\statetextwidth]{}};
		\node     [state](sr3) at (0.55,0.55){\makebox[\statetextwidth]{}};
		\node     [state](sr4) at (2.25,2.25){\makebox[\statetextwidth]{}};
		\node     [state](sr5) at (0.55,2.25){\makebox[\statetextwidth]{}};

		\begin{scope}[->,line cap=rect]
			\draw (sr0) to (4.5,4.4);
			\draw (sr0) to (4.4,-1.8);
			\draw (sr0) to (4.2,-1.7);
			\draw (sr0) to (4,-1.6);
			\draw (sr1) to (3.1,-1.8);
			\draw (sr1) to (3.3,-1.6);
			\draw[out=200,in=90] (sr1) to (1.4,-1.6);
			\draw (sr1) to (sr3);
			\draw (sr2) to (3.5,4.4);
			\draw (sr2) to (sr1);
			\draw (sr2) to (sr3);
			\draw (sr3) to (0.75,-1.6);
			\draw (sr3) to (sr5);
			\draw (sr4) to (sr2);
			\draw (sr4) to (sr3);
			\draw (sr4) to (sr5);
			\draw[out=250,in=110] (sr5) to (0.6,-1.8);
			\draw (sr5) to (1,4.4);

		\end{scope}
	\end{tikzpicture}
	\subcaption{\label{fig:exa:split-h}}
	\end{subfigure}
	\hfill\hfill
	\begin{subfigure}[t]{.3\linewidth}
		\centering
	\begin{tikzpicture}[>=Stealth,scale=0.9]

		\draw[block]
			(0,0) rectangle (1.1,2.8)
			(1.7,1.7) rectangle (4.1,2.8)
			(1.7,0) rectangle (3.2,1.1)
			(3.4,0) rectangle (4.5,1.5);
		\node at(0.85,1.4){$U_1$};
		\node at(2.15,0.8){$U_2$};
		\node at(2.9,2.55){$R_2$};
		\node at(3.75,1.1){$U$};

		\draw[bunch]
			(1.2,0) rectangle (1.6,2.8)
			(2.5,1.2) rectangle (3.3,1.6)
			(0,-1.2) rectangle (1.6,-0.2)
			(1.8,-1.2) rectangle (4.5,-0.2)
			(0,3) rectangle (4.5,4);

		\node at(1.1,-0.5) {$\TaB$};
		\node[align=left] at(2.3,-0.7) {$T \setminus \mbox{}$ \\ $\TaB$};
		\node at(0.3,3.7) {$T'$};


		\node     [state](sr0) at (3.95,0.55){\makebox[\statetextwidth]{}};
		\node     [state](sr1) at (2.65,0.55){\makebox[\statetextwidth]{}};
		\node     [state](sr2) at (3.55,2.25){\makebox[\statetextwidth]{\raisebox{-0.3ex}{nb}}};
		\node     [state](sr3) at (0.55,0.55){\makebox[\statetextwidth]{}};
		\node     [state](sr4) at (2.25,2.25){\makebox[\statetextwidth]{}};
		\node     [state](sr5) at (0.55,2.25){\makebox[\statetextwidth]{}};

		\begin{scope}[->,line cap=rect]
			\draw (sr0) to (4.5,4.4);
			\draw (sr0) to (4.4,-1.8);
			\draw (sr0) to (4.2,-1.7);
			\draw (sr0) to (4,-1.6);
			\draw (sr1) to (3.1,-1.8);
			\draw (sr1) to (3.3,-1.6);
			\draw[out=200,in=90] (sr1) to (1.4,-1.6);
			\draw (sr1) to (sr3);
			\draw (sr2) to (3.5,4.4);
			\draw (sr2) to node[red]{\small $m$} (sr1);
			\draw (sr2) to (sr3);
			\draw (sr3) to (0.75,-1.6);
			\draw (sr3) to (sr5);
			\draw (sr4) to (sr2);
			\draw (sr4) to (sr3);
			\draw (sr4) to (sr5);
			\draw[out=250,in=110] (sr5) to (0.6,-1.8);
			\draw (sr5) to (1,4.4);

		\end{scope}
	\end{tikzpicture}
	\subcaption{\label{fig:exa:split-i}}
	\end{subfigure}

	\caption{Illustration of splitting of a small block from $T$ and stabilising block $B$ with respect to the new bunches $\TaB$ and $T \setminus \TaB$, as explained in Example~\ref{exa:split}.\label{fig:exa:split}}
	\end{figure}
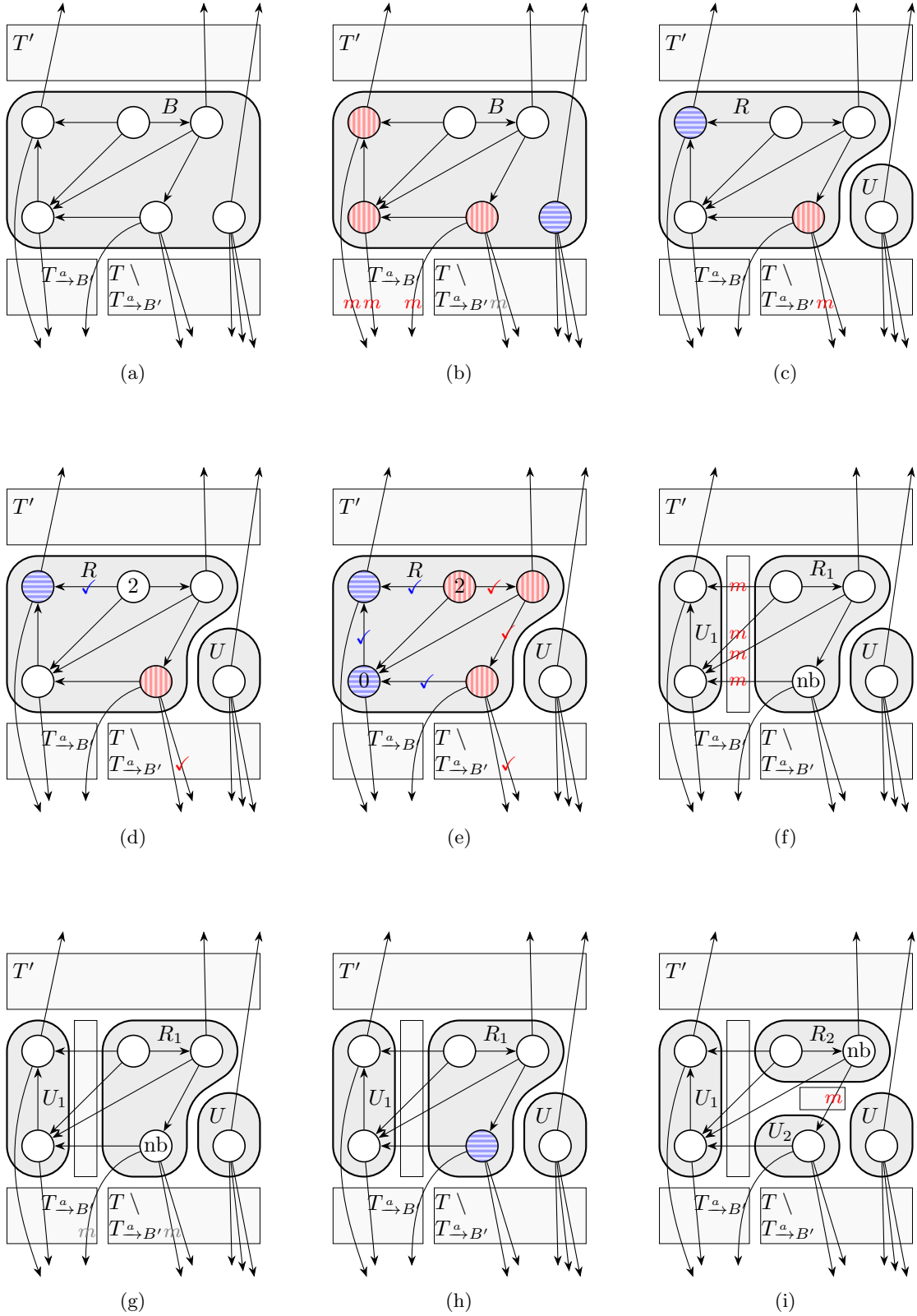

\begin{example}\label{exa:split}
Consider the situation in Figure~\ref{fig:exa:split-a}.
Observe that block $B$ is stable w.r.t.\@ the bunches $T$ and $T'$.
We have split off a small bunch $\TaB$ from $T$, and as a consequence, $B$ needs to be restabilised.
The bunches put on the splitter list initially are $\TaB$ and $T \setminus \TaB$.
When putting these bunches on the splitter list, all transitions in $\TBaB$ are marked, see the {\color{red}$m$}'s in Figure~\ref{fig:exa:split-b}.
Also, for states that have transitions both in $\TaB$ and in $T \setminus \TaB$,
one such transition in the latter bunch is marked, see the {\color{gray}$m$}'s in Figure~\ref{fig:exa:split-b}.

We now first split $B$ w.r.t.\@ the primary splitter $\TaB$ into $R$, the states that can inertly reach $\TaB$, and $U$, the states that cannot. In Figure~\ref{fig:exa:split-b}, the states known to be destined for $R$ are indicated \redstatetext, the states known to be destined for $U$ are indicated \bluestatetext.
Initially, all states with a marked outgoing transition are destined for $R$, the remaining bottom state of $B$ is destined for $U$.
The $\mathsf{split}$ subroutine proceeds to extend sets $R$ and $U$ in a backwards fashion using two coroutines, marking a state destined for $R$ if one of its successors is already in $R$, and marking a state destined for $U$ if all its successors are in $U$.
In this case, the state in $U$ does not have any incoming inert transitions, so its coroutine immediately terminates and all other states belong to $R$. Block $B$ is split into $R$ and $U$.
The resulting block $U$ is stable w.r.t.\@ both $\TaB$ and $T \setminus \TaB$. The resulting sets $R$ and $U$ are shown in Figure~\ref{fig:exa:split-c}.

We still need to split $R$ w.r.t.\@ $T \setminus \TaB$, into $R_1$ and $U_1$, say.
For this, we use the marked transitions in $T \setminus \TaB$ as a starting point to compute all bottom states that can reach a transition in $T \setminus \TaB$.
This guarantees that the time we use is proportional to $\size{\TaB}$.
Initially, there is one state destined for $R_1$, marked \redstatetext\ in Figure~\ref{fig:exa:split-c}, and one state destined for $U_1$, marked \bluestatetext\ in the same figure.
We now perform both coroutines in $\mathsf{split}$ simultaneously.
Figure~\ref{fig:exa:split-d} shows the situation after both coroutines have considered one transition:
The $U_1$-coroutine (the coroutine that calculates the states that cannot inertly reach a transition in $T \setminus \TaB$) has initialised the counter $\mathit{untested}$ of one state to $2$ at Line~\ref{alg:line-define-untested}$\ell$ (i.e., in the left column) of Algorithm~\ref{alg:split}
because two of its outgoing inert transitions have not yet been considered.
The $R_1$-coroutine (the coroutine that calculates the states that can inertly reach a transition in $T \setminus \TaB$) has checked the unmarked transition in the splitter $T_{R\pijl{}} \setminus T_{R \pijl{a} B'}$.
As the latter coroutine has finished visiting unmarked transitions in the splitter,
the $U_1$-coroutine no longer needs to run the slow test loop at Lines~\ref{alg:slow-test}$\ell$--\ref{alg:slow-test-end}$\ell$ of Algorithm~\ref{alg:split}.
In Figure~\ref{fig:exa:split-e} the situation is shown after two more steps in the coroutines.
Each has visited two extra transitions.
There are two extra states destined for $R_1$, marked \redstatetext,
and there is one state destined for $U_1$ with $0$ remaining inert transitions,
for which we know immediately that it has no transition in $T \setminus \TaB$, this is marked \bluestatetext.
Now, the $R_1$-coroutine is terminated, since it contains more that $\frac{1}{2}\size{R}$ states,
and the remaining incoming transitions of states in $U_1$ are visited.
This will not further extend $U_1$ since all states are already either destined for $R_1$ or for $U_1$.

The result of splitting is shown in Figure~\ref{fig:exa:split-f}.
Note that some inert transitions become non-inert,
so a new bunch with the transitions $R_1 \pijl{\tau} U_1$ is created,
and all those transitions are marked {\color{red}$m$}.

We next have to split $R_1$ with respect to this new bunch into the set of states $N_1$ that can inertly reach a transition in the new bunch, and the set $R'_1$ that cannot inertly reach this bunch.
In this case, all states in $R_1$ have a marked outgoing transition, hence $N_1 = R_1$, and $R'_1 = \emptyset$.
The coroutine that calculates the set of states that cannot inertly reach a transition in the bunch will immediately terminate because there are no transitions to be considered.

Observe that $R_1$ ($=N_1$) has a new bottom state.
This means that stability of $R_1$ with respect to any bunch is not guaranteed any more, and this stability needs to be re-established.
We therefore consider all bunches in which $R_1$ has an outgoing transition.
We add $T_{R_1\pijl{a}B'}$, $T_{R_1 \pijl{}} \setminus T_{R_1 \pijl{a} B'}$ and $T'_{R_1\pijl{}}$ to the splitter list as secondary splitters,
and mark one of the outgoing transitions from each bottom state in each of those bunches using {\color{gray}$m$}.
This situation is shown in Figure~\ref{fig:exa:split-g}.

In this case, $R_1$ is stable w.r.t.\@ $T_{R_1\pijl{a}B'}$ and $T_{R_1 \pijl{}} \setminus T_{R_1 \pijl{a} B'}$, i.e., all states in $R_1$ can inertly reach a transition in both bunches.
In both cases this is observed immediately after initialisation in split,
since the set of states that cannot inertly reach a transition in these bunches is initially empty, and the corresponding coroutine terminates immediately.

Therefore, consider splitting $R_1$ with respect to $T'_{R_1\pijl{}}$. This leads to $R_2$, the set of states that can inertly reach a transition in $T'$, and $U_2$, the set of states that cannot inertly reach a transition in $T'$.
Note there are no marked transitions in $T'_{R_1\pijl{}}$,
so initially all bottom states of $R_1$ are destined for $U_2$ (marked \bluestatetext\ in Figure~\ref{fig:exa:split-h}), and there are no states destined for $R_2$.
Then we start splitting $R_1$.
In the $R_2$-coroutine, we first add the states with an unmarked transition in $T'_{R_1\pijl{}}$ to $R_2$ at Line~\ref{alg:line-R-union-unmarked-T}$r$
(i.e., in the right column of Algorithm~\ref{alg:split})
and then all predecessors of the new bottom state need to be considered.
When $\mathsf{split}$ terminates, there will be no additional states in $U_2$, and the remaining states end up in $R_2$.

The situation after splitting $R_1$ into $R_2$ and $U_2$ is shown in Figure~\ref{fig:exa:split-i}.
One of the inert transitions becomes non-inert, this is marked {\color{red}$m$}.
Furthermore, $R_2$ contains a new bottom state. This is the state in $R_2$ with a transition to $T'$, and it is marked nb.
Note that, as each block necessarily has a bottom state (see Lemma~\ref{lem::noloop}), a non-bottom state had to become a bottom state in this case.

We need to continue stabilising $R_2$ w.r.t.\@ the bunch $R_2 \pijl{\tau} U_2$,
which does not lead to a new split,
and we need to restabilise $R_2$ w.r.t.\@ all bunches in which it has an outgoing transition.
This also does not lead to new splits, so the situation in Figure~\ref{fig:exa:split-i}
after removing the marking of transitions in $R_2 \pijl{\tau} U_2$
is the final result of splitting.
\end{example}

\subsection{Correctness}
The validity of the algorithm follows from a number of major invariants.

The main invariant is Invariant~\ref{inv:bunches}, which was stated before.
The following invariant
indicates that two non-inert transitions starting in the same block, having the same label and
ending in the same block, will always reside in the same bunch.
\begin{invariant}[Bunches are not unnecessarily split]
\label{inv:no_split_bunches}
For any pair of non-inert transitions $s\pijl{a}s'$ and $t\pijl{a}t'$,
if $s,t\in B$ and $s',t'\in B'$ then $s\pijl{a}s' \in T$ and $t\pijl{a}t'\in T$ for some bunch $T\in\Pi_t$.
\end{invariant}
\begin{proof}
Initially, $\Pi_t$ contains one bunch with all non-inert transitions.
Therefore, the invariant is valid.

There are two places in the algorithm where validity of this invariant can be jeopardised.
At Line~\ref{alg:split_transitions} $\Pi_t$ is refined,
and at Lines~\ref{alg:new_bunch} and \ref{alg:extend_bunch} a new bunch is created and extended.

We first look at Line~\ref{alg:split_transitions}.
We replace bunch $T$ with the two bunches $ \TaB$ and $ T \setminus \TaB$.
As all the transitions in $\TaB$ have label $a$ and go to block $B'$, the invariant is maintained for $\TaB$.
All transitions in $T \setminus \TaB$ have a label different from $a$ or go to a different block than $B'$,
so the invariant is maintained for these transitions as well.

At Lines~\ref{alg:new_bunch} and \ref{alg:extend_bunch} all new non-inert transitions are put in a new bunch.
As these are the only $\tau$-transitions between blocks $R$ and $U$,
the invariant remains valid under the creation of this bunch.
\end{proof}
The following invariant says that states that are branching bisimilar will never end up in separate blocks.
\begin{invariant}[Preservation of branching bisimilarity]
\label{inv:bbisim}
For all states $s,t \in S$, if $s \bbis t$, then there is some block $B \in \Pi_s$ such that $s,t \in B$.
\end{invariant}
\begin{proof}
Initially, this invariant is valid. We prove this by contraposition.
Consider two states $s,t$ in different blocks.
Then $s\in B_\text{vis}$ and $t\in B_\text{invis}$, or vice versa.
But then $s$ and $t$ cannot be branching bisimilar, because from $s$ a visible action is inertly reachable, whereas that is not the case from $t$ (or vice versa).

The preservation of this invariant can be seen as follows.
There are two places where blocks in $\Pi_s$ are split, namely at Lines~\ref{alg-split1} and \ref{alg-split2}.
We first concentrate on the splitting at Line~\ref{alg-split1}.
In this case a block $B$ is split into $R$ and $U$ where $\langle R, U\rangle$ is the result of an invocation of $\mathsf{split}(B, \TB')$.
Assume that the invariant would be invalidated.
This means that there are states $s\in R$ and $t\in U$ with $s\bbis t$.
As $s \in R$, by Equation~\eqref{eqn:splitproperty} on page~\pageref{eqn:splitproperty} we
know that $s\pijl{\tau}s_1\pijl{\tau}\cdots\pijl{\tau}s_n\pijl{a}s'$ with $\{ s_1,\ldots,s_n \} \subseteq B$, $s_n\pijl{a}s'\in \TB'$ and $s'\in B'$.
As $s_n\pijl{a}s'$ is non-inert and $s\bbis t$, we must have $t\pijl{\tau}t_1\pijl{\tau}\cdots\pijl{\tau}t_m\pijl{a}t'$, $t_j\in B$ for all $1\leq j\leq m$ and $t'\in B'$.
Hence, by Invariant~\ref{inv:no_split_bunches},
$t_m\pijl{a}t'\in T'$ where $T'$ is the bunch such that $\TB'\subseteq T'$.
But as $\TB'$ is a block-bunch-slice, $t_m\pijl{a}t'\in \TB'$.
Hence, by \eqref{eqn:splitproperty} on page~\pageref{eqn:splitproperty} it follows that $t\in R$,
contradicting the assumption that the invariant could be invalidated.

At Line~\ref{alg-split2} splitting takes place with regard to a splitter $R \pijl{\tau} U$.
As this is a bunch satisfying Invariant~\ref{inv:no_split_bunches},
the reasoning is completely analogous to the previous case.
\end{proof}
The following invariant says that there are no $\tau$-loops in any block. Actually, a stronger property holds, namely that there
are no $\tau$-loops at all, after they have been removed during the initialisation of the algorithm.
\begin{invariant}[No inert loops]
\label{inv:noloops}
There is no inert loop in a block, i.e., for every sequence $s_1\pijl{\tau}s_2\pijl{\tau}\cdots\pijl{\tau}s_n$ with $s_i\in B\in \Pi_s$ it holds
that $s_i\not=s_j$ for all $1\leq i<j\leq n$.
\end{invariant}
\begin{proof}
Initially, this invariant holds
because every strongly connected component consisting of states connected via $\tau$-transitions is contracted and merged into a single state.
The only operation that influences this invariant is splitting a block (Lines~\ref{alg-split1} and \ref{alg-split2}).
Splitting blocks cannot introduce new loops.
\end{proof}

As a consequence of the previous invariant, every block has at least one bottom state, and from every
state in a block a bottom state can be inertly reached.
\begin{lemma}
\label{lem::noloop}
Invariant~\ref{inv:noloops} implies that for all partitions $\Pi_s$ of $S$, and all blocks $B$ in $\Pi_s$, we have
\begin{enumerate}
	\item $\Bottom(B) \neq \emptyset$.
	\item For every state $s \in B$, there is a path of inert transitions leading to a bottom state in $B$.
\end{enumerate}
\end{lemma}
\begin{proof}
Let $\Pi_s$ be an arbitrary partition of $S$, and $B \in \Pi_s$ be a block in $S$ such that $B \neq \emptyset$.
\begin{enumerate}
	\item	Towards a contradiction, assume that $\Bottom(B) = \emptyset$.
		Then every state $s$ in $B$ has an outgoing transition $s \pijl{\tau} s'$ for some $s' \in B$.
		Since $B$ is finite, there must be a $\tau$-cycle in $B$.
		This contradicts Invariant~\ref{inv:noloops}.
	\item	Let $s \in B$ be arbitrary, and assume that $s$ does not have a path of inert transitions leading to a bottom state in $B$.
		Then, there must, again, be a $\tau$-cycle in $B$, which contradicts Invariant~\ref{inv:noloops}.
		\qedhere
\end{enumerate}
\end{proof}

%
%
%
%
%


Invariant~\ref{inv:second_inner_loop} is a technical invariant required to prove the main Invariant~\ref{inv:bunches}.
Invariant~\ref{inv:second_inner_loop} holds for the second inner for loop,
and says that Invariant~\ref{inv:bunches} holds,
except for the block-bunch-slices on the splitter list,
for which the invariant has to be re-established by splitting blocks.

\begin{invariant}[Inner loop at Lines~\ref{alg-line-stabilize-begin}--\ref{alg-line-stabilize-end}]
\label{inv:second_inner_loop}
	If a non-empty block-bunch-slice $\TB$ is not in the splitter list,
	then every bottom state in its source block $B$ has a transition in $\TB$.
\end{invariant}
\begin{proof}
First it is shown that the invariant holds when arriving at the second for loop at Line~\ref{alg-line-stabilize-begin}.
Consider a non-empty block-bunch-slice $\hat{T}_{\smash[t]{\hat{B}}\pijl{}}$ that is not in the splitter list.
If $\hat{T}_{\smash[t]{\hat{B}}\pijl{}}$ is not a subset of $T$ at Line~\ref{algo:main:outerloop},
it follows from Invariant~\ref{inv:bunches} that all $t\in \Bottom(\hat{B})$ have a transition in $\hat{T}$, so the invariant holds.

Now assume $\hat{T}_{\smash[t]{\hat{B}}\pijl{}}$ is a subset of $T$. Note that $\hat{T} = T$.
Then one of the following cases apply:
\begin{itemize}
\item
	$\hat{T}_{\smash[t]{\hat{B}}\pijl{}} = T_{\smash[t]{\hat{B}}\pijl{a}B'} \subseteq \TaB$.
	As $T_{\smash[t]{\hat{B}}\pijl{a}B'}$ is not empty,
	there is a transition $s\pijl{a}s'\in T_{\smash[t]{\hat{B}}\pijl{a}B'} \subseteq T$.
	Towards a contradiction, suppose there is some state $t\in\Bottom(\hat{B})$ that does not have an outgoing transition in $\TaB$.
	Then, $t$ must have a transition in $T \setminus \TaB$ by Invariant~\ref{inv:bunches},
	thus $\hat{B}\in \splittableBlocks(\TaB)$
	and hence $T_{\smash[t]{\hat{B}}\pijl{a}B'}$ is on the splitter list.
	This contradicts the assumption that $\hat{T}_{\smash[t]{\hat{B}}\pijl{}}$ is not on the splitter list.
\item
	$\hat{T}_{\smash[t]{\hat{B}}\pijl{}} = T_{\smash[t]{\hat{B}}\pijl{}}\subseteq T\setminus \TaB$.
	It cannot happen that $\hat{B} \in \splittableBlocks(\TaB)$, since
	then $T_{\smash[t]{\hat{B}}\pijl{}}$ is on the splitter list.
	Therefore, $\hat{B} \not \in \splittableBlocks(\TaB)$.
	Since $T_{\smash[t]{\hat{B}}\pijl{}} \neq \emptyset$, we have $T_{\smash[t]{\hat{B}}\pijl{}} = T$,
	and it immediately follows from Invariant~\ref{inv:bunches}
	that all $t\in \Bottom(\hat{B})$ have a transition in $T_{\smash[t]{\hat{B}}\pijl{}}$.
\end{itemize}

We now show that the loop starting at Line~\ref{alg-line-stabilize-begin} maintains this invariant.
Concretely, we consider some block-bunch-slice $\hat{T}_{\smash[t]{\hat{B}}\pijl{}}$
that is not empty and does not occur in the splitter list at Line~\ref{alg-line-stabilize-end}.
We make a case distinction on the block $\hat{B}$.

\begin{itemize}
\item
	Let us first assume that $\hat{B}$ is not the result of splitting a block at Lines~\ref{alg-split1} or \ref{alg-split2}.
	So, $\hat{B} \cap B = \emptyset$.
	This means that $\hat{T}_{\smash[t]{\hat{B}}\pijl{}}$ was not split during the last iteration of the for loop,
	and it is not a subset of the bunch with new non-inert transitions created at Lines~\ref{alg:new_bunch} and \ref{alg:extend_bunch}.
	Hence, the invariant remains valid for $\hat{T}_{\smash[t]{\hat{B}}\pijl{}}$ during this last iteration.
\item
	Assume $\hat{B}$ is a subset of $R$ and the condition at Line~\ref{alg:inert_condition} is not valid,
	i.e., $R\pijl{\tau} U=\emptyset$.
	This means $\hat{B}$ is not split further at Line~\ref{alg-split2}, i.e., $\hat{B} = R$.
	We treat the cases where $\hat{T}_{R\pijl{}}$ is or is not a subset of $T'_{B\pijl{}}$ separately.
	\begin{itemize}
	\item	If $\hat{T}_{R\pijl{}}\subseteq T'_{B\pijl{}}$ then we see that due to the splitting at Line~\ref{alg-split1}
		all states in $R$ can reach a transition in $T'_{B\pijl{}} = T'_{R\pijl{}}$.
		In particular, all bottom states of $\hat{B}$ have a transition in $T'_{B\pijl{}}$
		and this is also the case for $\hat{T}_{R\pijl{}}$ by construction.
		Hence, the invariant is valid for $\hat{T}_{R\pijl{}}$.
	\item	$\hat{T}_{R\pijl{}}\cap T'_{B\pijl{}}=\emptyset$.
		In this case, $\hat{T}_{R\pijl{}}$ is the result of splitting a block-bunch-slice $\hat{T}_{B\pijl{}}$
		that was not on the splitter list,
		as splitting $\hat{T}_{B\pijl{}}$ results in keeping the splitted block-bunch-slices on this list
		in case the original was already there.
		This means that all bottom states in $B$ have transitions in $\hat{T}_{B\pijl{}}$ by Invariant~\ref{inv:bunches},
		and by construction, those bottom states in $R$ have transitions in $\hat{T}_{R\pijl{}}$.
		As $R\pijl{\tau} U=\emptyset$, there are no additional bottom states in $\Bottom(R) \setminus \Bottom(B)$
		and the invariant is established for $\hat{T}_{R\pijl{}}$.
	\end{itemize}
\item
	Assume $\hat{B}$ is a subset of $R$ and the condition at Line~\ref{alg:inert_condition} is valid.
	This means that $\hat{B}$ is either equal to $N$ or $R'$.
	\begin{itemize}
	\item	Assume $\hat{B}$ equals $N$.
		This means that $\hat{T}_{\smash[t]{\hat{B}}\pijl{}}$ has the shape $T_{N\pijl{}}$ (Line~\ref{alg:line-make-all-slices-unstable})
		and is added to the splitter list,
		contradicting our assumption
		that $\hat{T}_{\smash[t]{\hat{B}}\pijl{}}$ is not on the splitter list at Line~\ref{alg-line-stabilize-end}.
	\item	Assume $\hat{B}$ is equal to $R'$.
		In this case the block-bunch-slice $\hat{T}_{R'\pijl{}}$ cannot be a subset of the new bunch
		created at Lines~\ref{alg:new_bunch} and \ref{alg:extend_bunch}
		as there are no transitions in $\mbox{$(R\pijl{\tau}U)$} \cup \mbox{$(N \pijl{\tau} R')$}$ from $R'$.

		If the original block-bunch-slice $\hat{T}_{B\pijl{}}$
		that existed at the beginning of the second for loop satisfying $\hat{T}_{R'\pijl{}}\subseteq \hat{T}_{B\pijl{}}$
		was in the splitter list,
		then also $\hat{T}_{R'\pijl{}}$ would be on the list.
		Contradiction.

		Therefore, $\hat{T}_{B\pijl{}}$ was not in the splitter list.
		As $\hat{T}_{B\pijl{}}$ was not empty, all bottom states in $B$ would have transitions in $\hat{T}_{B\pijl{}}$.
		As all new bottom states are moved to $N$, all bottom states in $R'$ are also bottom states in $B$,
		and it follows that all bottom states in $R'$ have transitions in $\hat{T}_{R'\pijl{}}$.
		Hence, the invariant holds.
	\end{itemize}
\item
	Here we consider the situation where $\hat{B}=U$.
	There are three cases to consider.
	\begin{itemize}
	\item	If $\hat{T}_{U\pijl{}}=T'_{U\pijl{}}$,
		we see that $T'_{U\pijl{}}$ is empty,
		because it contains all transitions in $T'_{B\pijl{}}$ reachable from states in $U$,
		from which, by construction, transitions in $\hat{T}_{B\pijl{}}$ cannot be reached.
		As $T'_{U\pijl{}}$ is empty, the invariant holds trivially.
	\item	Suppose $\hat{T}_{U\pijl{}}=T_{U\pijl{}}\setminus T_{U\pijl{a}B'}$ and $\TB' = \TBaB$ is primary.
		We know that $B$ is stable w.r.t.\ $\TB$ by Invariant~\ref{inv:bunches}.
		That means that every bottom state of $B$ has a transition in $\TB$.
		Bottom states in $U$ have no transitions in $\TaB$.
		Hence, they all have transitions in $T\setminus \TaB$.
		If we restrict this set of transitions to those starting in $U$,
		we see that all bottom states in $U$ also have a transition in $T_{U\pijl{}}\setminus T_{U\pijl{a}B'}$
		and the invariant holds.
	\item	In this last case we investigate the remaining situations.
		So either $\hat{T}_{U\pijl{}} \not= T'_{U\pijl{}}$,
		or $\hat{T}_{U\pijl{}} \not= T_{U\pijl{}}\setminus T_{U\pijl{a}B'}$
		or $\TB'$ is not primary.

		If the original block-bunch-slice $\hat{T}_{B\pijl{}}$
		that existed at the beginning of the second for loop satisfying $\hat{T}_{U\pijl{}} \subseteq \hat{T}_{B\pijl{}}$
		was in the splitter list,
		then also $\hat{T}_{U\pijl{}}$ would be on the list.
		Contradiction.

		Therefore, $\hat{T}_{B\pijl{}}$ was not in the splitter list.
		As $\hat{T}_{U\pijl{}}$ is not empty,
		$\hat{T}_{B\pijl{}}$ is also not empty
		and so, all bottom states in $B$ have a transition in $\hat{T}_{B\pijl{}}$.
		All bottom states of $U$ are bottom states from $B$.
		(Otherwise there must have been a transition $s\pijl{\tau}s'$ from a state $s\in U$ to a state $s'\in R$;
		but then $s$ would be part of $R$.)
		Hence, all bottom states in $U$ have a transition in $\hat{T}_{U\pijl{}}$, showing that the invariant holds.
		\qedhere
	\end{itemize}
\end{itemize}
\end{proof}
Invariant~\ref{inv:bunches} is the main invariant of the algorithm. It is valid at Line~$\ref{algo:main:outerloop}$.
It is an adaptation to branching bisimulation of a similar invariant in \cite{Valmari2009}. It says that the
partition is always stable w.r.t.\ the bunches.
Stability refers to the presence of a transition in a bunch, and hence does not relate to actions or target blocks.
We finally prove its invariance.
The proof relies on Invariant~\ref{inv:second_inner_loop}.

\begin{proof}[Proof of Invariant~\ref{inv:bunches}]
We show that 	$\Pi_s$ is stable under $\Pi_t$, i.e.,
if a bunch $T \in \Pi_t$ contains a transition with its source state in a block $B$ of $\Pi_s$, then every bottom state in block $B$ has a transition in bunch
$T$ (in fact, in block-bunch-slice $\TB$).

Initially, this invariant is valid:
The initial blocks in $\Pi_s$ are $B_\text{vis}$ and $B_\text{invis}$.
Furthermore, there is only one bunch.
From the states in $B_\text{invis}$ no visible transition is reachable, and this means that all transitions from states in $B_\text{invis}$ are inert.
Therefore, they do not occur in the initial bunch and Invariant~\ref{inv:bunches} holds trivially.
Each bottom state in $B_\text{vis}$ has an outgoing visible transition, which must occur in the initial bunch.
This also makes the invariant valid in a trivial way, albeit for a different reason.

The invariant is invalidated when $\Pi_t$ is split at Line~\ref{alg:split_transitions}.
At the end of the second for loop (Line~\ref{alg-line-stabilize-end}),
emptiness of the splitter list and
Invariant~\ref{inv:second_inner_loop} imply
that all block-bunch-slices are stable.
This implies the invariant as follows.
Consider a bunch $T\in \Pi_t$.
Assume there is a transition $s\pijl{a}s'\in T$ with $s\in B$.
The transition $s\pijl{a}s'$ occurs in the block-bunch-slice $T_{B\pijl{}}$.
As $T_{B\pijl{}}$ is stable, it holds that every  state $t\in \Bottom(B)$ has an outgoing transition in $T_{B\pijl{}}$
and therefore the invariant holds at Line~\ref{alg-line-stabilize-end}.
\end{proof}

The invariants given above allow us to prove that the algorithm works correctly.
When the algorithm terminates (and this always happens, see Section~\ref{sec:complexity}),
branching bisimilar states are perfectly grouped in blocks.
\begin{theorem}
From the Invariants~\ref{inv:bbisim}, \ref{inv:noloops} and \ref{inv:bunches}, after the algorithm terminates, it holds that $\mathord{\equiv_{\Pi_s}}\! = \mathord{\bbis}$\,.
\end{theorem}
\begin{proof}
By Invariant~\ref{inv:bbisim} it follows that $\mathord{\bbis} \subseteq \mathord{\equiv_{\Pi_s}}$.
Hence, we only need to show that $\mathord{\equiv_{\Pi_s}}\subseteq \mathord{\bbis}$\,.
We do this by showing that $\equiv_{\Pi_s}$ is a branching bisimulation relation.

Consider states $s,t\in S$ such that $s,t\in B$ for some block $B\in \Pi_s$. Assume $s\pijl{a}s'$.
\begin{itemize}
\item
	If $a = \tau$ and $s'\in B$, i.e., the
	transition is inert, then $s'\equiv_{\Pi_s} t$ and we have fulfilled the proof obligation for branching bisimulation in this case.
\item
	If the transition $s\pijl{a}s'$ is not inert, it is part of some bunch $T\in \Pi_t$.
	By Invariant~\ref{inv:noloops} and Lemma~\ref{lem::noloop} there is a path of inert transitions $t\pijl{\tau}\cdots\pijl{\tau} t'$ with $t'\in \Bottom(B)$,
	and by Invariant~\ref{inv:bunches} there is a transition $t'\pijl{b}t''\in T$.
	As the algorithm terminated, the condition at Line~\ref{algo:main:outerloop} is false,
	which means that $\nraB(T)=1$.
	In other words, $T$ is equal to some action-block-slice. This must be $\TaB$, as $s\pijl{a}s'\in T$ (where $s' \in B'$).
	Hence, $b=a$ and $t''\in B'$, and the proof obligation for branching bisimulation has been fulfilled.
\end{itemize}
Concluding, $\equiv_{\Pi_s}$ is a branching bisimulation, and therefore $\mathord{\equiv_{\Pi_s}} = \mathord{\bbis}$\,.
\end{proof}

We provide the following invariant as a precondition for splitting blocks in Section~\ref{sec:splitting_blocks}.
It says that whenever $\mathsf{split}(\hat{B}, \hat{T}_{\smash[t]{\hat{B}}\pijl{}})$ is called, the bottom states in block $\hat{B}$ with transitions in bunch $\hat{T}$ can be found by only looking at the marked
transitions. Checking whether a state has a marked outgoing transition can be done in constant time, which is essential
for the algorithm.

\begin{invariant}[Marked transitions]
\label{inv:marked_transitions}
Whenever $\mathsf{split}(\hat{B}, \hat{T}_{\smash[t]{\hat{B}}\pijl{}})$ is called,
it holds that \[\Bottom(\hat{B})_{\pijl{\hat{T}_{\smash[t]{\hat{B}}\pijl{}}}} = \Bottom(\hat{B})_{\pijl{\Marked(\hat{T}_{\smash[t]{\hat{B}}\pijl{}})}}.\]
\end{invariant}
\begin{proof}
Whenever $\mathsf{split}(\hat{B}, \hat{T}_{\smash[t]{\hat{B}}\pijl{}})$ is called,
the block-bunch-slice $\hat{T}_{\smash[t]{\hat{B}}\pijl{}}$ is in the splitter list.
There are four cases when a block-bunch-slice is inserted into the splitter list:
\begin{itemize}
\item	$\hat{T}_{\smash[t]{\hat{B}}\pijl{}} = T_{\smash[t]{\hat{B}}\pijl{a}B'}$ is a primary splitter of $\hat{B}$.
	Then, all transitions in $T_{\smash[t]{\hat{B}}\pijl{a}B'}$ are marked at Line~\ref{alg-line-mark-all-transitions-in-primary},
	so the invariant obviously holds.
\item	$\hat{T}_{\smash[t]{\hat{B}}\pijl{}}$ is a secondary splitter of $\hat{B}$ that has been added at Line~\ref{alg-line-find-splittable-block}.
	Note that for every block, we first split under its primary splitter $T_{\smash[t]{\hat{B}}\pijl{a}B'}$
	and then apply what remains of the secondary splitter only to $R$
	(or only to $R'$ in case $R \pijl{\tau} U$ is not empty,
	but for brevity, we only mention $R$ below).
	So, the call is $\mathsf{split}(R, T_{R\pijl{}} \setminus T_{R\pijl{a}B'})$.
	As every bottom state of $R$ has a transition in $T_{\smash[t]{\hat{B}}\pijl{}}$,
	Line~\ref{alg-line-mark-one-transition-in-secondary-splitter} ensures
	that every bottom state in $R$
	which already was a bottom state in $\hat{B}$
	and which has a transition in $T_{R\pijl{}} \setminus T_{R\pijl{a}B'}$,
	also has a marked transition in it.
	So, the invariant holds.
\item	$\hat{T}_{\smash[t]{\hat{B}}\pijl{}} = R \pijl{\tau} U$ at Line~\ref{alg:split-call2}.
	All transitions of $R \pijl{\tau} U$ are marked at Line~\ref{alg:new_bunch},
	so the invariant holds.
\item	$\hat{T}_{\smash[t]{\hat{B}}\pijl{}} = \hat{T}_{N\pijl{}}$ is a splitter of $N$ added at Line~\ref{alg:line-make-all-slices-unstable}.
	In that case, Line~\ref{alg:line-mark-bb-of-new-bottom} ensures
	that every bottom state with a transition in $\hat{T}_{N\pijl{}}$ has a marked transition in it.
	Hence, also in this last case, the invariant holds.
\end{itemize}
Additionally, it can happen that a splitter itself is split.
Then, the two new subslices keep their markings.
If new bottom states result from the split,
they are handled in a similar way to the last case above:
in all relevant slices, each new bottom state with an outgoing transition in that slice also has a marked outgoing transition in that slice.
\end{proof}

\subsection{Complexity}\label{sec:complexity}


To simplify the complexity notations
we assume that $n \leq m + 1$. This is not a significant restriction,
since it is satisfied by any labelled transition system in which every non-initial state has an incoming transition.
This can easily be achieved by preprocessing the graph to remove unreachable states.
Furthermore, we assume that we can access action labels fast enough to bucket sort the transitions
in time $\bigo{m}$, which is for instance the case if the action labels are consecutively numbered.

We show that the algorithm runs in time $\bigo{m \log n}$,
using the time budgets
printed in grey at the right-hand side of the pseudocode,
which indicate how much time each piece of code is allowed to spend.
In Section~\ref{sect:implementation} it is explained how the data structures meet these time budgets.

The initialisation (Lines~\ref{alg-line-contract-SCCs}--\ref{algo:main:outerloop}) can be performed in $\bigo{m}$ time, where for the calculation of the sets of states, the assumption that $n \leq m+1$ is used.
The calculation of the time complexity of the while loop
is split into three separate parts. The first part regards splitting bunches, putting block-bunch-slices on the splitter list
and marking transitions,
which is attributed to the transitions in a new small bunch.
The second part deals with splitting blocks, which is attributed to the transitions of the smaller subblock.
The third
part handles the calculations that are required when states become new bottom states.

When splitting bunches, we apply the ``Process the smaller half'' principle to new bunches.
Every transition is an element of a new bunch $\TaB$ at Line~\ref{alg-line-find-splittable-block} at most $\lfloor \log_2 n^2 \rfloor + 1$ times,
because the first time $\TaB$ is investigated, it has at most $n^2$ elements,
and every subsequent time a new bunch containing the same transition is investigated,
it has at most half the size of the previous bunch.
Whenever a transition is an element of a new bunch $\TaB$,
it is processed in constant time.
This is indicated with the time budget $\bigo{\size{\TaB}}$ for Lines~\ref{alg-line-select-small-bunch}--\ref{alg-line-stabilize-begin}.
Also, processing block-bunch-slices of the form $\TBaB$ and $\TB \setminus \TBaB$
(as added to the splitter list at Line~\ref{alg-line-find-splittable-block})
requires time in $\bigo{\size{\Marked(\TB')}}$ at Lines~\ref{alg:line:callsplit1}--\ref{alg:new_bunch},
corresponding to at most the number of transitions in $\TBaB$.
Summing over all transitions gives runtime $\bigo{m \log n}$
attributed to new bunches.

In the next section we explain in detail how long $\mathsf{split}$ can take.
In short, its runtime depends on the \emph{smaller} of the two resulting subblocks---%
so we apply ``Process the smaller half'' to that block.
Every state can be part of such a smaller subblock  at most $\lfloor \log_2 n \rfloor$ times,
because the first time it is part of a subblock of at most $n/2$ states,
and at every subsequent time the same state becomes part of another smaller subblock,
the latter has at most half the size of the previous subblock.
Whenever a state is in the smaller subblock of $\mathsf{split}$,
we are allowed to attribute time proportional to the number of incoming and outgoing transitions of that state.
More precisely, provided each source state of $T$ is in $B$, the complexity of calculating $\mathsf{split}(B,T)$ is the following.
If $\size{R} \leq \size{U}$, then the time spent is $\bigo{\size{R_{\pijl{}}} + \size{R_{\lijp{}}}}$,
and if $\size{U} \leq \size{R}$, the time spent is $\bigo{\size{\Marked(T)} +\size{U_{\pijl{}}} + \size{U_{\lijp{}}} + \size{(\Bottom(R)\setminus \Bottom(B))_{\pijl{}}}}$.

In Algorithm~\ref{algo:main-algorithm-abstract} this is indicated with the time budget ``$\bigo{\size{U_{\pijl{}}} + \size{U_{\lijp{}}}}$ or $\bigo{\size{R_{\pijl{}}} + \size{R_{\lijp{}}}}$'' for Lines~\ref{alg:line:callsplit1}--\ref{alg:new_bunch}
and with ``$\bigo{\size{R'_{\pijl{}}} + \size{R'_{\lijp{}}}}$ or $\bigo{\size{N_{\pijl{}}} + \size{N_{\lijp{}}}}$'' for Lines~\ref{alg:split-call2}--\ref{alg:extend_bunch}.
Also, as $R \pijl{\tau} U \subseteq U_{\lijp{}} \cap R_{\pijl{}}$,
we attribute $\bigo{\size{R \pijl{\tau} U}}$ at Lines~\ref{alg:split-call2}--\ref{alg:extend_bunch} to $U$ or $R$, whichever is smaller.
Summing over all states gives rise to a cumulative time complexity of $\bigo{m \log n}$.

Finally, some work is attributed to new bottom states.
Every non-bottom state can become a bottom state at most once during the whole execution.
When this happens, we attribute time proportional to its outgoing transitions to it.
This is indicated with several time budgets $\bigo{\size{\Bottom(N)_{\pijl{}}}}$ at Lines~\ref{alg:line:callsplit1}--\ref{alg:line-mark-bb-of-new-bottom}.
At Line~\ref{alg:line-make-all-slices-unstable} we need to include not only the current new bottom states but also the future ones
because there may be block-bunch-slices that only have transitions from non-bottom states.
When $N$ is split under such a block-bunch-slice, at least one of these non-bottom states will become a bottom state
but we cannot yet say which one(s) right now.
Also, for block-bunch-slices of the form $T_{N\pijl{}}$ at Line~\ref{alg:line-mark-bb-of-new-bottom},
we budget $\bigo{\size{\Marked(\TB')}}$ at Lines~\ref{alg:line:callsplit1}--\ref{alg:new_bunch} corresponding to the new bottom states in $N$.
Summing over all states gives runtime $\bigo{m}$
attributed to new bottom states.

Adding up these three time budgets shows that the grand total of all work is  $\bigo{m \log n}$.

\section{Splitting blocks}
\label{sec:splitting_blocks}

The function $\mathsf{split}(B,T)$, presented in Algorithm~\ref{alg:split}, refines the block $B$ into two subblocks, $R$ and $U$,
where $R$ contains those states in $B$ that can inertly reach a transition in $T$,
and $U$ contains the states that cannot,
as formally specified in Equation~\eqref{eqn:splitproperty}.

These two sets are computed by two coroutines executing in lockstep:
the two coroutines start the same number of loop iterations,
so that the overhead is at most proportional to the faster of the two,
and all work done in both coroutines can be attributed to the smaller of the two subblocks $R$ and $U$.

As a precondition, guaranteed by Invariant~\ref{inv:marked_transitions},
the function requires that the marking of transitions in $\TB$ is such
that bottom states of $B$ that have an outgoing transition in $\TB$ also have a marked outgoing transition in $\TB$.
Formally,
\begin{equation*}
\Bottom(B)_{\pijl{\Marked(\TB)}} = \Bottom(B)_{\pijl{\TB}}.
\end{equation*}
The initial sets are computed as follows.
Initially, all states in $B_{\pijl{\Marked(T)}}$,
i.e., all states that are the source of a marked transition in $T$, are put in $R$.
All bottom states that are not in $R$ initially are put in $U$.
Observe that these initial sets can be computed in $\bigo{\size{\Marked(T)}}$ time
by grouping bottom states with marked transitions.

The sets are extended as follows in the coroutines.
For the states in $R$, first the states in $B_{\pijl{T \setminus \Marked(T)}}$ are added that were not yet in $R$.
These are all the sources of an unmarked transitions in $T$.
Using backward reachability along inert transitions,
$R$ is extended until either $R$ is stable (no states can be added), or $R$ contains more than half the states in $B$.

To identify the states in $U$, observe that a state $t$ is in $U$
if all its inert successors are in $U$
and it does not have a transition in $\TB$.
The first condition is vacuously satisfied for bottom states.
To compute $U$, we let at counter $\mathit{untested}[t]$ for every non-bottom state $t$ record the number of outgoing inert transitions to states that are not yet known to be in $U$.
If $\mathit{untested}[t] = 0$, this means all inert successors of $t$ are guaranteed to be in $U$, so, provided $t$ does not
have a transition in $\TB$, $t$ is also added to $U$.
To take care of the possibility that all inert transitions of $t$ have been visited before all states that are the source of a transition in
$\TB$ are added to $R$, we check all non-inert transitions of $t$ to determine
whether they are not in $\TB$ at Lines~\ref{alg:slow-test}$\ell$--\ref{alg:slow-test-end}$\ell$, i.e., in the left column of Algorithm~\ref{alg:split}.
Note that checking all successors of such a state is balanced with marking the states in $R$.
In the next section, we explain how to initialize the array $\mathit{untested}$ in constant time.

The coroutine that finishes first,
provided that its number of states does not exceed $\frac{1}{2}\size{B}$,
has completely computed the smaller subblock resulting from the refinement,
and the other coroutine can be aborted.


\begin{algorithm}[tb]
	\caption{Refinement of a block under a splitter}
	\label{alg:split}
	\algsetup{indent=0.45cm}
	\small\begin{algorithmic}[1]
	\STATE	\textbf{function} $\mathsf{split}(\text{block } B, \text{block-bunch-slice } T$)
	\STATE \label{alg:line-initRMarked}  $R:=B_{\pijl{\Marked(T)}}$; \quad
	       \label{alg:line-initU} $U := \Bottom(B) \setminus R$
	\STATE	\textbf{begin coroutines} \label{alg:line-after-initU}
		\\[0.18\baselineskip]\setlength{\itemsep}{-0.7\baselineskip}%
	\begin{tabular}{@{}p{5.95cm}||p{4.92cm}@{}}
		\STATE	\hspace*{\algorithmicindent}Set $\mathit{untested}[t]$ to undefined for all $t \mathbin{\in} B$                           & $R:=R \cup B_{\pijl{T \setminus \Marked(T)}}$\label{alg:line-R-union-unmarked-T}\rightbrace[-2]{1}{$\bigo{\size{\Marked(T)}}$}\label{alg:line-initR}\label{alg:line-init-untested}\\
		\STATE	\hspace*{\algorithmicindent}\algorithmicforall\ $s\in U$ \algorithmicwhile\ $\size{U}\leq\frac{1}{2}\size{B}$ \algorithmicdo & \algorithmicforall\,$s\,{\in}\,R$ \algorithmicwhile\,$\size{R}\,{\leq}\,\frac{1}{2}\size{B}$\,\algorithmicdo\rightbrace[-1]{1}{$\bigo{1}$ or $\bigo{\size{R_{\pijl{}}}}$}\\
		\STATE	\hspace*{2\algorithmicindent}\algorithmicforall\ inert transitions $t \pijl{\tau} s$ \algorithmicdo                      & \hspace*{\algorithmicindent}\algorithmicforall\,inert\,transitions\,$t\,{\pijl{\tau}}\,s$ \algorithmicdo\label{alg:line-for-inert-predecessors}\rightbrace[-1]{8}{\parbox[l]{0.3\textwidth}{$\bigo{\size{U_{\lijp{}}}}$ or \\ $\bigo{\size{R_{\lijp{}}}}$}}\\
		\STATE	\hspace*{3\algorithmicindent}\algorithmicif\ $t\in R$ \algorithmicthen\ Skip state $t$\label{alg:line-skip-R-states}      & \\
		\STATE	\hspace*{3\algorithmicindent}\algorithmicif\ $\mathit{untested}[t]$ is undefined \algorithmicthen                         & \\
		\STATE	\hspace*{4\algorithmicindent}$\mathit{untested}[t]\,{:=}\,\size{\{ t\,{\pijl{\tau}}\,u\,{\mid}\,u\,{\in}\,B \}}$          & \\ \label{alg:line-define-untested}
		\STATE	\hspace*{3\algorithmicindent}\algorithmicendif                                                                            & \\
		\STATE	\hspace*{3\algorithmicindent}$\mathit{untested}[t] := \mathit{untested}[t] - 1$                                           & \\
		\STATE	\hspace*{3\algorithmicindent}\algorithmicif\,$\mathit{untested}[t]\,{>}\,0$\,\algorithmicthen\,Skip\,state\,$t$           & \\
		\STATE	\hspace*{3\algorithmicindent}\algorithmicif\ $\BT\not \subseteq R$\ \algorithmicthen \label{alg:slow-test}       & \rightbrace{5}{\parbox[l]{0.3\textwidth}{%
																			$\bigO(\size{U_{\pijl{}}} + \mbox{}$ \\
																			$\phantom{\bigO (}\lvert(\Bottom(R) \setminus \mbox{}$ \\
																			$\phantom{\bigO (\lvert(}\Bottom(B))_{\pijl{}}\rvert)$}}\\
		\STATE	\hspace*{4\algorithmicindent}\algorithmicforall\ non-inert $t \pijl{\alpha} u$ \algorithmicdo                            & \\
		\STATE	\hspace*{5\algorithmicindent}\algorithmicif\ $t \pijl{\alpha} u \in T$  \algorithmicthen\ Skip $t$                        & \\
		\STATE	\hspace*{4\algorithmicindent}\algorithmicendfor                                                                           & \\
		\STATE	\hspace*{3\algorithmicindent}\algorithmicendif \label{alg:slow-test-end}                                                  & \\
		\STATE	\hspace*{3\algorithmicindent}Add $t$ to $U$                                                                               & \hspace*{2\algorithmicindent}Add $t$ to $R$\rightbrace{3}{\parbox[l]{0.3\textwidth}{$\bigo{\size{U_{\lijp{}}}}$ or \\ $\bigo{\size{R_{\lijp{}}}}$}}\\
		\STATE	\hspace*{2\algorithmicindent}\algorithmicendfor                                                                           & \hspace*{\algorithmicindent}\algorithmicendfor\\
		\STATE	\hspace*{\algorithmicindent}\algorithmicendfor\label{alg:line-handle-state-loopend}                                       & \algorithmicendfor\\
		\STATE	\hspace*{\algorithmicindent}\algorithmicif\ $\size{U}>\frac{1}{2}\size{B}$ \algorithmicthen                               & \algorithmicif\ $\size{R}>\frac{1}{2}\size{B}$ \algorithmicthen\rightbrace{5}{$\bigo{1}$}\\
		\STATE	\hspace*{2\algorithmicindent}Abort this coroutine                                                                         & \hspace*{\algorithmicindent}Abort this coroutine\\
		\STATE	\hspace*{\algorithmicindent}\algorithmicendif                                                                             & \algorithmicendif\\
		\STATE	\hspace*{\algorithmicindent}Abort the other coroutine                                                                     & Abort the other coroutine \\
		\STATE	\hspace*{\algorithmicindent}\textbf{return} $(B\setminus U,U)$\label{alg:line-return}                                     & \textbf{return} $(R,B\setminus R)$\\
	\end{tabular}\setlength{\itemsep}{0cm}%
	\STATE	\textbf{end coroutines}\\
	\end{algorithmic}
	\end{algorithm}

In detail, the runtime complexity of $\langle R, U \rangle := \mathsf{split}(B, T)$ is:
\begin{itemize}
\item $\bigo{\size{R_{\pijl{}}} + \size{R_{\lijp{}}}}$, if $\size{R} \leq \size{U}$,
and
\item $\bigo{\size{\Marked(T)} + \size{U_{\pijl{}}} + \size{U_{\lijp{}}} + \size{(\Bottom(R) \setminus \Bottom(B))_{\pijl{}}}}$, if $\size{U} \leq \size{R}$.
\end{itemize}

This complexity can be inferred as follows.
First, note that we execute the coroutines in lockstep.
That means that running them will incur an overhead that is at most proportional to the faster of the two.
As soon as one coroutine has found more than $\frac{1}{2}\size{B}$ states,
its subblock is known to be too large,
so this coroutine is aborted.
This will only reduce the overhead.
Therefore, it is enough to show that the runtime bound for the smaller subblock is satisfied.
Note that if both subblocks have the same size, $\size{R} = \size{U} = \frac{1}{2} \size{B}$,
the faster of the two finishes first, so both runtime bounds apply.

Consider the case $\size{R} \leq \size{U}$.
Observe $\size{\Marked(T)} \leq \size{R_{\pijl{}}}$, and all work is attributed to the coroutine at the right of Algorithm~\ref{alg:split},
so we get $\bigo{\size{R_{\pijl{}}} + \size{R_{\lijp{}}}}$ directly from the $R$-coroutine.

Now consider the case $\size{U} \leq \size{R}$.
Then we use time in $\bigo{\size{\Marked(T)}}$ for Line~\ref{alg:line-initRMarked},
and we use time in $\bigo{\size{U_{\lijp{}}}}$ for everything else except Lines~\ref{alg:slow-test}$\ell$--\ref{alg:slow-test-end}$\ell$.
For these latter lines,
we distinguish two cases.
If it turns out that $t$ has no transition $t \pijl{\alpha} u \in T$,
it is a $U$-state, so we can attribute the time to $\bigo{\size{U_{\pijl{}}}}$.
Otherwise, i.e., if there is some $t \pijl{\alpha} u \in T$,
it is an $R$-state.
How to account for such an $R$-state in the coroutine that is supposed to find $U$?
The solution is that $t$ is a new bottom state.
It had some inert transitions in $B$,
but they all are now in $R \pijl{\tau} U$.
So we attribute the time to the outgoing transitions of new bottom states:
$\bigo{\size{(\Bottom(R) \setminus \Bottom(B))_{\pijl{}}}}$.
Note that eventually we reach the situation that all states in $\BT$ are in $R$ by Line~\ref{alg:line-R-union-unmarked-T}$r$,
and this expensive check is skipped.

It can also happen that $U$ is empty. In that case time in $\bigo{\size{\Marked(T)}}$ is spent in $\mathsf{split}$. The initialisation takes $\bigo{\size{\Marked(T)}}$ time, and since $U$ is empty after initialisation, the left coroutine terminates immediately, and the red coroutine is aborted. The latter takes $\bigo{1}$ time.

\section{Implementation details}
\label{sect:implementation}

This section shows how the operations of the abstract algorithm can be implemented
in a way that fits all time bounds.

\subsection{Implementation of data types}

\paragraph{States}
are stored as a \emph{refinable partition} data structure \cite{ValmariL08},
grouped per block,
i.e., in an array such that states in the same block are adjacent to each other.
Then, a block of states can be described as a slice in this array. An example of this is shown in Figure~\ref{fig:examplepartition}.
Within each such slice, we separate bottom from non-bottom states and states known to be destined for $R$ from other states. Note that initially, $R$ contains exactly the states with marked outgoing transitions. This is illustrated in the first two lines of Figure~\ref{fig:sort-then-swap}.
This makes it possible to visit all states in a block $B$ in time $\bigo{\size{B}}$,
to find its bottom states that are not in $R$ in constant time (cf.\@ Line~\ref{alg:line-initU} of Algorithm~\ref{alg:split}) and visit these bottom states in time $\bigo{\size{\Bottom(B) \setminus R}}$.

\begin{figure}[!htp]
	\centering
	\begin{subfigure}{.3\linewidth}
		\begin{tikzpicture}[>=Stealth,scale=0.9]

			\draw[block]
				(0,0) rectangle (1.1,2.8)
				(3.4,0) rectangle (4.5,1.5)
				(0,3.8) rectangle (4.5,4.9)
				(0,-1) rectangle (4.5,-2.1)
				;

			\node at(0.4,4.35){$B_1$};
			\node at(0.4,-1.55){$B_2$};
			\node at(3.75,1.1){$B_3$};
			\node at(0.85,1.4){$B_4$};

			\draw[block]
			(1.7,0) -- (1.7,2.8) [sharp corners] -- (3.55,2.8) arc[radius=0.55cm,start angle=90,end angle=-55.79454]
			-- ++(214.20546:0.4) arc[radius=0.75cm,start angle=124.20546,end angle=180] [rounded corners=4.95mm] -- (3.2,0) -- cycle;
		  \node at(2.9,2.55){$B_0$};

			\draw[bunch]
				(0,3) rectangle (4.5,3.6) 
				(0,-0.8) rectangle (2.15,-0.2) 
				(2.35,-0.8) rectangle (4.5,-0.2) 
				(1.2,0) rectangle (1.6,2.8) 
			;

			\node at(0.3,3.3) {$T'$};
			\node at(0.3,-0.5) {$T''$};
			\node[align=left] at(2.65,-0.5) {$T$};
			\node at(1.4,2.5) {$T_\tau$};

			\node     [state](sr0) at (3.95,0.55){\makebox[\statetextwidth]{$s_0$}};
			\node     [state](sr1) at (2.65,0.55){\makebox[\statetextwidth]{$s_1$}};
			\node     [state](sr2) at (3.55,2.25){\makebox[\statetextwidth]{$s_2$}};
			\node     [state](sr3) at (0.55,0.55){\makebox[\statetextwidth]{$s_3$}};
			\node     [state](sr4) at (2.25,2.25){\makebox[\statetextwidth]{$s_4$}};
			\node     [state](sr5) at (0.55,2.25){\makebox[\statetextwidth]{$s_5$}};
			\node			[state](sr6) at (3.55,4.35){\makebox[\statetextwidth]{$s_6$}};
			\node			[state](sr7) at (1.08, -1.55){\makebox[\statetextwidth]{$s_7$}};
			\node			[state](sr8) at (3.43,-1.55){\makebox[\statetextwidth]{$s_8$}};

			\begin{scope}[->,line cap=rect]
				\draw (sr0) edge[bend right] node[right] {$d$} (sr6);
				\draw (sr0) edge[bend left] node[right,yshift=2pt] {$c$} (sr8);
				\draw (sr1) to node[right,yshift=1pt] {$b$} (sr8);
				\draw (sr1) to node[left,xshift=1pt] {$a$} (sr7);
				\draw (sr1) to (sr3);
				\draw (sr2) to node[right] {$d$} (sr6);
				\draw (sr2) to (sr1);
				\draw (sr2) to (sr3);
				\draw (sr3) to node[right,xshift=-2pt] {$a$} (sr7);
				\draw (sr3) to (sr5);
				\draw (sr4) to (sr2);
				\draw (sr4) to (sr3);
				\draw (sr4) to (sr5);
				\draw (sr5) edge[out=235,in=110] node[yshift=20pt] {$a$} (sr7);
				\draw (sr5) edge[bend left] node[below,xshift=-5pt,yshift=-5pt] {$d$} (sr6);
				\draw (sr7) edge[out=50,in=17,looseness=40] node[below,yshift=-20pt,xshift=-20pt] {$b$} (sr7);
				\draw (sr8) edge[out=60,in=90,looseness=15] node[above,yshift=-2pt] {$b$} (sr8);

			\end{scope}
		\end{tikzpicture}
		\subcaption{An example LTS and partition. Unlabelled transitions are assumed to be $\tau$-transitions.\label{fig:examplepartition-lts}}
		\end{subfigure}
		\hfill
		\begin{subfigure}{.6\linewidth}
		\tikzset{ground/.pic={\draw[arrows=-] (-0.15,0)--(0.15,0) (-0.1125,-0.05) -- (0.1125,-0.05) (-0.075,-0.1) -- (0.075,-0.1);}}
		\begin{tikzpicture}[scale=0.9,>=Stealth]
			\begin{scope}[every node/.style={rectangle,draw,minimum width=9mm,minimum height=4.5mm}]
				\node[alias=s6] (pos6) at (0,0){$s_6$};

				\node[alias=s7] (pos7) at (1,0){$s_7$};
				\node[alias=s8] (pos8) at (2,0){$s_8$};

				\node[alias=s0] (pos5) at (3,0){$s_0$};

				\node[alias=s5] (pos0) at (4,0){$s_5$};
				\node[alias=s3] (pos1) at (5,0){$s_3$};

				\node[alias=s1] (pos2) at (6,0){$s_1$};
				\node[alias=s4] (pos3) at (7,0){$s_4$};
				\node[alias=s2] (pos4) at (8,0){$s_2$};

			\end{scope}
			\begin{scope}[color=gray,line cap=round,decoration={brace,amplitude=2.75pt},below,rectangle,text=black]
				\draw[decorate] (s0) ++(0.485,-0.3) -- node{\footnotesize bot.} ++(-0.97,0);
				\draw[decorate] (s1) ++(0.485,-0.3) -- node{\footnotesize bot.} ++(-0.97,0);
				\draw[decorate] (s5) ++(0.485,-0.3) -- node{\footnotesize bot.} ++(-0.97,0);
				\draw[decorate] (s6) ++(0.485,-0.3) -- node{\footnotesize bot.} ++(-0.97,0);
				\draw[decorate] (s8) ++(0.485,-0.3) -- node{\footnotesize bot.} ++(-1.97,0);
			\end{scope}
			\begin{scope}[every node/.style={rounded rectangle,minimum height=6mm,draw,inner ysep=0pt,inner xsep=1pt,minimum width=6mm}]
				\draw (pos0) ++(0.5,1.2) node (B4){$B_4$};
				\draw (pos2) ++(1, 1.2) node (B0){$B_0$};
				\draw (pos5) ++(0  ,1.2) node (B3){$B_3$};
				\draw (pos6) ++(0  ,1.2) node (B1){$B_1$};
				\draw (pos7) ++(0.5,1.2) node (B2){$B_2$};
			\end{scope}

			\begin{scope}[line cap=rect]
			\draw[->] (pos0.north) to (B4);
			\draw[->] (pos1.north) to (B4);
			\draw[->] (pos2.north) to (B0);
			\draw[->] (pos3.north) to (B0);
			\draw[->] (pos4.north) to (B0);
			\draw[->] (pos5.north) to (B3);
			\draw[->] (pos6.north) to (B1);
			\draw[->] (pos7.north) to (B2);
			\draw[->] (pos8.north) to (B2);
			\end{scope}
		\end{tikzpicture}
		\subcaption{Refinable partition data structure.\label{fig:examplepartition-partition}}
		\end{subfigure}
	\caption{Snapshot of an LTS with its partitions, and the corresponding refinable partition data structure.\label{fig:examplepartition}}
	\end{figure}
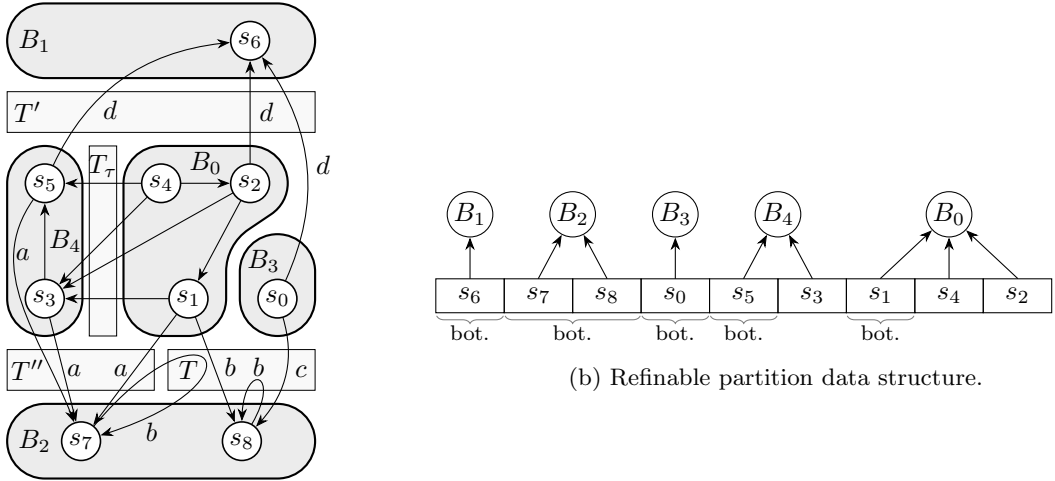

\begin{example}
	Figure~\ref{fig:examplepartition-lts} shows an LTS and its current partition.
	The corresponding refinable partition data structure is shown in Figure~\ref{fig:examplepartition-partition}. The history of the splitting is as follows. We start with a block consisting of all states. From this, the states that cannot inertly reach a visible transition are split off as $B_1$. Subsequently, $B_2$, $B_3$ and $B_4$ are split off from the remaining states, resulting in the partition shown in the figure.
	Note that in Figure~\ref{fig:examplepartition-partition} it is also illustrated that in blocks $B_0$ and $B_4$ the non-bottom states are grouped together, and appear after the bottom states.
	\end{example}

\begin{figure}[!hbt]
\begin{center}
\begin{tikzpicture}
	\draw	( 0  ,5) node[draw,right,rectangle,minimum width=5cm,fill=black!10](bottom){\vphantom{Gg}\smash{\makebox[0pt]{bottom}}}
		( 5  ,5) node[draw,right,rectangle,minimum width=8cm](non-bottom){\vphantom{Gg}\smash{\makebox[0pt]{non-bottom}}}

		( 0  ,4) node[draw,right,rectangle,minimum width=3.4cm,fill=black!10](bottomunmarked){\vphantom{Gg}\smash{\makebox[0pt]{bottom}}}
		( 3.4,4) node[draw,right,rectangle,minimum width=1.6cm,fill=black!10](bottommarked){\vphantom{Gg}\smash{\makebox[0pt]{$R$}}}
		( 5  ,4) node[draw,right,rectangle,minimum width=6.7cm](non-bottomunmarked){\vphantom{Gg}\smash{\makebox[0pt]{non-bottom}}}
		(11.7,4) node[draw,right,rectangle,minimum width=1.3cm](non-bottommarked){\vphantom{Gg}\smash{\makebox[0pt]{$R$}}}

		( 0  ,3) node[draw,right,rectangle,minimum width=2.7cm,fill=blue!20](bluebottomtemp){\vphantom{Gg}\smash{\makebox[0pt]{$U$ bottom}}}
		( 2.7,3) node[draw,right,rectangle,minimum width=2.3cm,fill=red!20](redbottomtemp){\vphantom{Gg}\smash{\makebox[0pt]{$R$ bottom}}}
		( 5  ,3) node[draw,right,rectangle,minimum width=1.2cm,fill=blue!5](bluenon-bottomtemp){\vphantom{Gg}\smash{\makebox[0pt]{$U$}}}
		( 6.2,3) node[draw,right,rectangle,minimum width=2.8cm](notbluenonzero){\vphantom{Gg}\smash{\makebox[0pt]{$\mathit{untested} > 0$}}}
		( 9  ,3) node[draw,right,rectangle,minimum width=2.5cm](notblueundefined){\vphantom{Gg}\smash{\makebox[0pt]{$\mathit{untested} = \text{?}$}}}
		(11.5,3) node[draw,right,rectangle,minimum width=1.5cm,fill=red!5](rednon-bottomtemp){\vphantom{Gg}\smash{\makebox[0pt]{$R$}}}

		( 0  ,2) node[draw,right,rectangle,minimum width=2.7cm,fill=blue!20](bluebottom){\vphantom{Gg}\smash{\makebox[0pt]{$U$ bottom}}}
		( 2.7,2) node[draw,right,rectangle,minimum width=2.3cm,fill=red!20](redbottom){\vphantom{Gg}\smash{\makebox[0pt]{$R$ bottom}}}
		( 5  ,2) node[draw,right,rectangle,minimum width=3.1cm,fill=blue!5](bluenon-bottom){\vphantom{Gg}\smash{\makebox[0pt]{$U$ non-bottom}}}
		( 8.1,2) node[draw,right,rectangle,minimum width=4.9cm,fill=red!5](rednon-bottom){\vphantom{Gg}\smash{\makebox[0pt]{$R$ non-bottom}}}

		( 0  ,1) node[draw,right,rectangle,minimum width=2.7cm,fill=blue!20](bluebottom2){\vphantom{Gg}\smash{\makebox[0pt]{$U$ bottom}}}
		( 2.7,1) node[draw,right,rectangle,minimum width=3.05cm,fill=blue!5](bluenon-bottom2){\vphantom{Gg}\smash{\makebox[0pt]{$U$ non-bottom}}}
		( 5.85,1) node[draw,right,rectangle,minimum width=2.25cm,fill=red!20](redbottom2){\vphantom{Gg}\smash{\makebox[0pt]{$R$ bottom}}}
		( 8.1,1) node[draw,right,rectangle,minimum width=4.9cm,fill=red!5](rednon-bottom2){\vphantom{Gg}\smash{\makebox[0pt]{$R$ non-bottom}}}

		( 0  ,0) node[draw,right,rectangle,minimum width=2.7cm,fill=black!10](bluebottom3){\vphantom{Gg}\smash{\makebox[0pt]{bottom}}}
		( 2.7,0) node[draw,right,rectangle,minimum width=3.05cm](bluenon-bottom3){\vphantom{Gg}\smash{\makebox[0pt]{non-bottom}}}
		( 5.85,0) node[draw,right,rectangle,minimum width=2.25cm,fill=black!10](redbottom3){\vphantom{Gg}\smash{\makebox[0pt]{bottom}}}
		( 8.1,0) node[draw,right,rectangle,minimum width=1.8cm,fill=black!10](newbottom3){\vphantom{Gg}\smash{\makebox[0pt]{new b.}}}
		( 9.9,0) node[draw,right,rectangle,minimum width=3.1cm](rednon-bottom3){\vphantom{Gg}\smash{\makebox[0pt]{non-bottom}}}
	;
	\draw[-Stealth]  (non-bottommarked.west) -- ++(-0.25,0);
	\draw[-Stealth]      (bottommarked.west) -- ++(-0.25,0);
	\draw[-Stealth](bluenon-bottomtemp.east) -- ++( 0.25,0);
	\draw[-Stealth]  (notblueundefined.west) -- ++(-0.25,0);
	\draw[-Stealth]    (notbluenonzero.east) -- ++( 0.25,0);
	\draw[-Stealth] (rednon-bottomtemp.west) -- ++(-0.25,0);
	\draw[-Stealth]        (newbottom3.east) -- ++( 0.25,0);
	\begin{scope}[color=gray,line cap=round,decoration={brace,amplitude=2.75pt},below,rectangle,text=black]
		\draw[decorate] (bottom) ++(-2.485,0.3) -- node[above]{$B$ before refinement\vphantom{Gg}} ++(12.97,0);
		\draw[decorate] (bluenon-bottom3) ++(1.51,-0.3) -- node{new block} ++(-5.72,0);
		\draw[decorate] (rednon-bottom3) ++(1.535,-0.3) -- node{$B$ after refinement} ++(-7.12,0);
	\end{scope}
	\draw[-Stealth] (bluenon-bottom) to (bluenon-bottom2);
	\draw[-Stealth] (redbottom) to (redbottom2);

	\draw	(-1.7,4) node[anchor=west]{\ref{alg:line-initRMarked}}
		(-1.7,3) node[anchor=west]{\ref{alg:line-after-initU}--\ref{alg:line-handle-state-loopend}}
		(-1.7,2) node[anchor=west]{\ref{alg:line-return}}
		(-1.7,1) node[anchor=west]{\ref{alg-split1}/\ref{alg-split2}}
	;
\end{tikzpicture}
\end{center}
\caption{Internal structure of a block during and shortly after $\mathsf{split}$.
	In this example, the $U$-subblock is smaller, so it will become the new block.
	\label{fig:sort-then-swap}}
\end{figure}
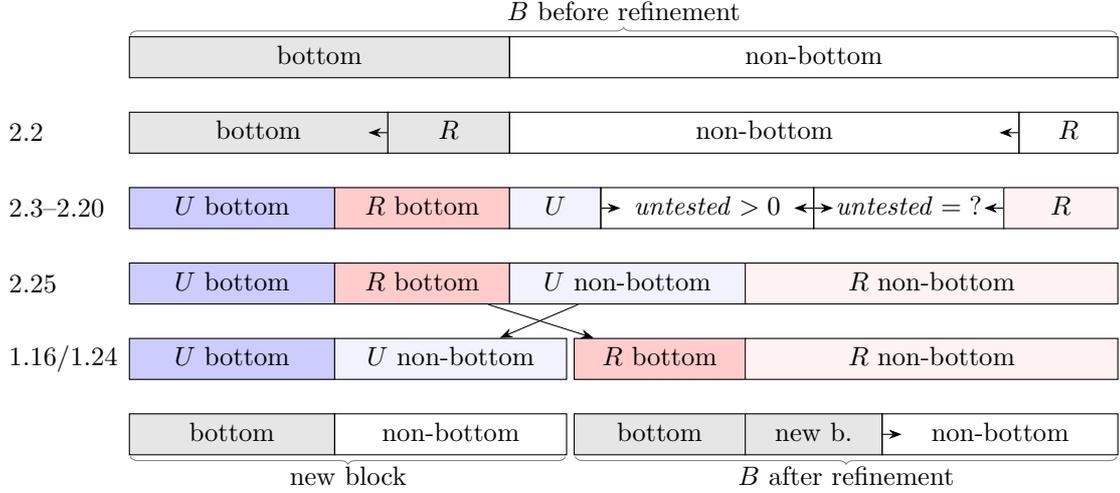

When a block is split, we need to update the data structure.
This can be done in time proportional to the smaller subblock at Lines~\ref{alg-split1} and~\ref{alg-split2}.
Figure~\ref{fig:sort-then-swap} illustrates how $U$ and $R$ are located in the slice of $B$.
In the third line in the figure, each state with $\mathit{untested}[t] \neq 0$ is stored in a specific slice of non-bottom states. To initialise $\mathit{untested}[t]$ to ``undefined'' at Line~\ref{alg:line-init-untested}$\ell$ of Algorithm~\ref{alg:split}, it suffices to set this slice to the empty slice.
After $\mathsf{split}$ has finished, the bottom states of $R$ and the non-bottom states of $U$ exchange places;
note that this can be done in time proportional to the smaller of the two subblocks.
In the example, a part of the non-bottom states of $U$ can even stay where they are.
At the end, new bottom states are searched and added to the bottom states of $R$.

\paragraph{Transitions}
are stored in four linked \emph{refinable partitions} \cite{ValmariL08}.
While not all four are essential for the concepts of the algorithm, we need them to ensure the time complexity bounds.
In Figure~\ref{fig:exampletransitions} we illustrate each of the refinable partitions for the transitions corresponding to the example in Figure~\ref{fig:examplepartition}.

\begin{itemize}
\item	Transitions are stored in an array \emph{grouped per bunch,}
	i.e., transitions in the same bunch are adjacent to each other, see Figure~\ref{fig:exampletransitions-action-block-slices}.
	Then, each bunch can be described as a slice in the array.

	Within a bunch, transitions are grouped further per action-block-slice.
	As a consequence, when a small action-block-slice needs to be split off a bunch,
	one can easily select either the first or the last action-block-slice in the bunch
	and split it off in constant time.

	When a block is split, we need to split its action-block-slices,
	which can be done in time proportional to the incoming transitions of the smaller subblock at Lines~\ref{alg-split1} and \ref{alg-split2}.
	This operation fits into the time budget.
\item	Transitions are stored in an array \emph{grouped per block-bunch-slice,} see Figure~\ref{fig:exampletransitions-block-bunch-slices}.
	Within each slice, the marked transitions are separated from the unmarked ones
	when the block-bunch-slice is a splitter.

	When a bunch is split, we need to split its block-bunch-slices,
	which can be done in time proportional to the smaller new bunch at Line~\ref{alg:split_transitions}.
	When a block is split, we need to split its block-bunch-slices,
	which can be done in time proportional to the outgoing transitions of the smaller subblock at Lines~\ref{alg-split1} and \ref{alg-split2}.
	Both operations fit into the allowed time budget.

	We need this partition to mark transitions quickly, namely in constant time per transition,
	to visit all marked transitions at Line~\ref{alg:line-initRMarked},
	and to visit all other transitions at Line~\ref{alg:line-initR}$r$.
\item	Transitions are stored \emph{grouped per source state,} see Figure~\ref{fig:exampletransitions-source-state}.
	Within each slice, transitions are further grouped into non-inert and inert transitions,
	and the non-inert transitions are grouped per bunch.

	When a bunch is split, we need to regroup the transitions in that bunch as well.
	This is done in time proportional to the smaller new bunch at Line~\ref{alg:split_transitions}.

	We need this partition to visit all outgoing transitions of a state at Line~\ref{alg:slow-test}$\ell$.
	We also use this partition to decide whether a state with a transition in $\TaB$ also has a transition in $\TB \setminus \TaB$:
	While regrouping the transitions in $\TB$ leaving from the same source state,
	we can recognize whether all transitions move to $\TaB$ or some remain in $\TB \setminus \TaB$.
\item	Transitions are stored \emph{grouped per target state,} see Figure~\ref{fig:exampletransitions-target-state}.
	Within each slice, transitions are further grouped into non-inert and inert transitions.
	This partition hardly ever changes.
	We need this partition to visit all incoming (inert) transitions of a state at Line~\ref{alg:line-for-inert-predecessors}.
\end{itemize}
When a transition becomes non-inert, we have to change all four partitions:
Create a new bunch, create a new block-bunch-slice, and move the transition from the inert to the non-inert ones in the last two partitions.
We do this by running over all outgoing (formerly) inert transitions of $R$ or all incoming (formerly) inert transitions of $U$,
depending on which subblock is smaller.
This requires either time $\bigo{\size{R_{\pijl{}}}}$ or $\bigo{\size{U_{\lijp{}}}}$, respectively,
which fits into the time budget at Lines~\ref{alg-split1} and~\ref{alg-split2}.

In our implementation, the four transition partitions are linked together via pointers.
When source and target state and (pointers to) the relevant slices mentioned above are only stored once,
we need nine pointers or \texttt{size\_t} integers per transition.

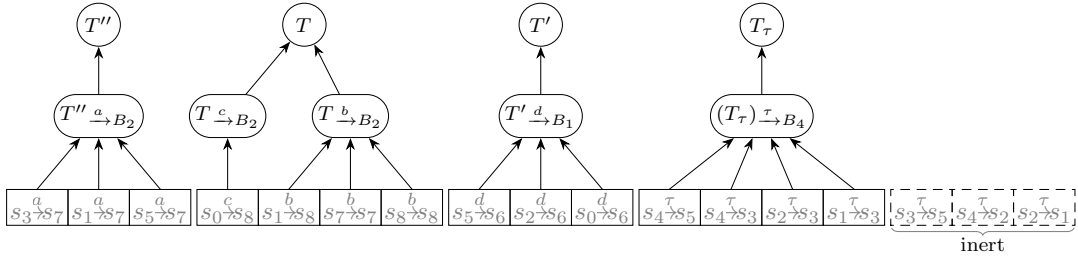
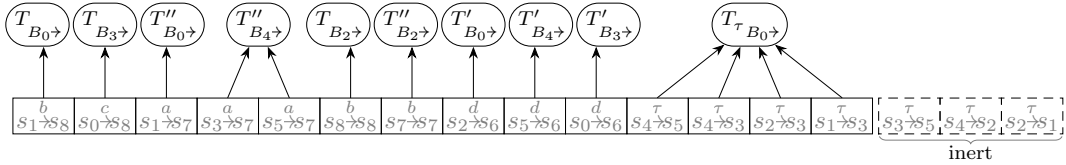
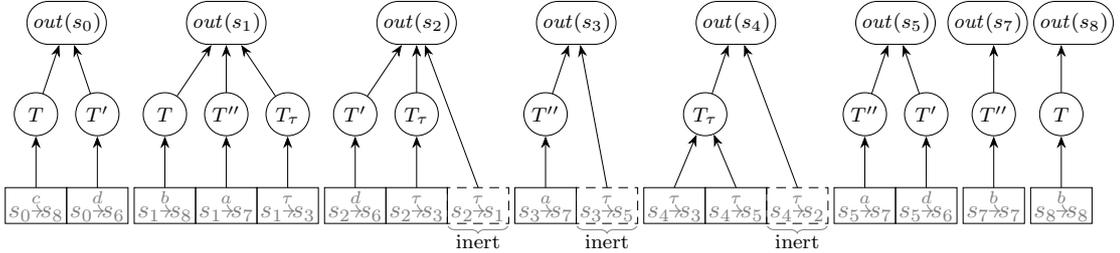
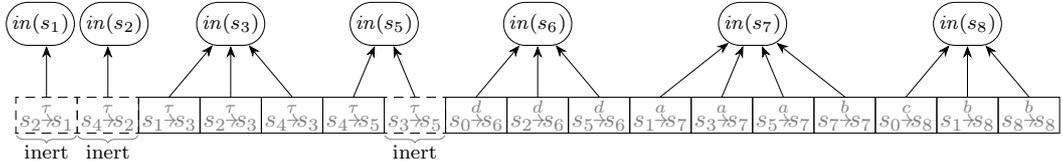
\begin{figure}[htp]
	\begin{subfigure}{\linewidth}
	\centering
	\scalebox{.95}{
	\begin{tikzpicture}[>=Stealth,scale=0.85,line cap=rect]
		\begin{scope}[every node/.style={rectangle,draw,minimum width=8.5mm,minimum height=5mm,anchor=mid,inner sep=0pt},
				inert/.style={line cap=rect,dash pattern=on 1.2mm off 1.05mm,dash phase=0.85mm}]
			\draw (0,5)
				++(0  ,0) node[text=gray,alias=T''start](s3tos7){$s_3\xshortto{a} s_7$}
				++(1  ,0) node[text=gray](s1tos7){$s_1\xshortto{a} s_7$}
				++(1  ,0) node[text=gray](s5tos7){$s_5\xshortto{a} s_7$}

				++(1.1,0) node[text=gray,alias=T'start](s0tos8){$s_0\xshortto{c} s_8$}
				++(1  ,0) node[text=gray](s1tos8){$s_1\xshortto{b} s_8$}
				++(1  ,0) node[text=gray](s7tos7){$s_7\xshortto{b} s_7$}
				++(1  ,0) node[text=gray](s8tos8){$s_8\xshortto{b} s_8$}

				++(1.1,0) node[text=gray,alias=Tstart](s5tos6){$s_5\xshortto{d} s_6$}
				++(1  ,0) node[text=gray](s2tos6){$s_2\xshortto{d} s_6$}
				++(1  ,0) node[text=gray](s0tos6){$s_0\xshortto{d} s_6$}

				++(1.1,0) node[text=gray,alias=Ttaustart](s4tos5){$s_4\xshortto{\tau} s_5$}
				++(1  ,0) node[text=gray](s4tos3){$s_4\xshortto{\tau} s_3$}
				++(1  ,0) node[text=gray](s2tos3){$s_2\xshortto{\tau} s_3$}
				++(1  ,0) node[text=gray](s1tos3){$s_1\xshortto{\tau} s_3$}

				++(1.1,0) node[text=gray,inert](s3tos5){$s_3\xshortto{\tau} s_5$}
				++(1  ,0) node[text=gray,inert](s4tos2){$s_4\xshortto{\tau} s_2$}
				++(1  ,0) node[text=gray,inert](s2tos1){$s_2\xshortto{\tau} s_1$}
			;
		\end{scope}

		\begin{scope}[every node/.style={rounded rectangle,minimum height=6mm,draw,inner ysep=0pt,inner xsep=1pt,minimum width=6mm}]
			\footnotesize
			 \draw (T''start)  ++( 1,1.5) node(T''aB2){${T''}_{\pijl{a}{B_2}}$};
			 \draw (Tstart)    ++( 1,1.5)   node(T'dB1){${T'}_{\pijl{d}{B_1}}$};
			 \draw (T'start)   ++( 0,1.5)   node(TcB2){${T^{\vphantom{\prime}}}_{\pijl{c}{B_2}}$};
			 \draw (T'start)   ++( 2,1.5) node(TbB2){${T^{\vphantom{\prime}}}_{\pijl{b}{B_2}}$};
			 \draw (Ttaustart) ++( 1.5,1.5) node(TtautauB4){$(T_\tau){\vphantom{T'}}_{\pijl{\tau}{B_4}}$};
		\end{scope}

		\begin{scope}[color=gray,line cap=round,decoration={brace,amplitude=2.75pt},below,rectangle,text=black]
		 	\draw[decorate] (s2tos1)       ++(0.485,-0.34) -- node{\footnotesize inert} ++(-2.97,0);
		\end{scope}
		\begin{scope}[every node/.style={rounded rectangle,minimum height=6mm,draw,inner ysep=0pt,inner xsep=1pt,minimum width=6mm}]
		\draw (T''aB2)    ++(0   ,1.5) node(T''){\footnotesize $T''$};
		\draw (TcB2)      ++(1.25,1.5) node(T){\footnotesize $T^{\vphantom{\prime}}$};
		\draw (T'dB1)     ++(0   ,1.5) node(T'){\footnotesize $T'$};
		\draw (TtautauB4) ++(0   ,1.5) node(Ttau){\footnotesize $\smash[b]{T_\tau}\vphantom{T'}$};
		\end{scope}
		\draw[->] (s3tos7.north) to (T''aB2);
		\draw[->] (s1tos7.north) to (T''aB2);
		\draw[->] (s5tos7.north) to (T''aB2);

		\draw[->] (s0tos8.north) to (TcB2);
		\draw[->] (s1tos8.north) to (TbB2);
		\draw[->] (s7tos7.north) to (TbB2);
		\draw[->] (s8tos8.north) to (TbB2);

		\draw[->] (s5tos6.north) to (T'dB1);
		\draw[->] (s2tos6.north) to (T'dB1);
		\draw[->] (s0tos6.north) to (T'dB1);

		\draw[->] (s4tos5.north) to (TtautauB4);
		\draw[->] (s4tos3.north) to (TtautauB4);
		\draw[->] (s2tos3.north) to (TtautauB4);
		\draw[->] (s1tos3.north) to (TtautauB4);

		\draw[->] (T''aB2) 		to (T'');
		\draw[->] (TbB2) 			to (T);
		\draw[->] (TcB2) 			to (T);
		\draw[->] (T'dB1) 		to (T');
		\draw[->] (TtautauB4) to (Ttau);
	\end{tikzpicture}
	} 

	\subcaption{
		The forest of transitions per bunch, grouped by action-block-slice, with the bunches.\label{fig:exampletransitions-action-block-slices}
	}
	\end{subfigure}

	\medskip
	\begin{subfigure}{\linewidth}
	\centering
	\scalebox{.95}{
	\begin{tikzpicture}[>=Stealth,scale=0.85,line cap=rect]
		\begin{scope}[every node/.style={rectangle,draw,minimum width=8.5mm,minimum height=5mm,anchor=mid,inner sep=0pt},
				inert/.style={line cap=rect,dash pattern=on 1.2mm off 1.05mm,dash phase=0.85mm}]
			\draw (-0.1,5)
				++(0  ,0) node[text=gray,alias=TB0start](s1tos8){$s_1\xshortto{b} s_8$}
				++(1  ,0) node[text=gray,alias=TB3start](s0tos8){$s_0\xshortto{c} s_8$}
				++(1  ,0) node[text=gray,alias=T''B0start](s1tos7){$s_1\xshortto{a} s_7$}
				++(1  ,0) node[text=gray,alias=T''B4start](s3tos7){$s_3\xshortto{a} s_7$}
				++(1  ,0) node[text=gray](s5tos7){$s_5\xshortto{a} s_7$}

				++(1  ,0) node[text=gray,alias=TB2start](s8tos8){$s_8\xshortto{b} s_8$}

				++(1  ,0) node[text=gray,alias=T''B2start](s7tos7){$s_7\xshortto{b} s_7$}

				++(1  ,0) node[text=gray,alias=T'B0start](s2tos6){$s_2\xshortto{d} s_6$}
				++(1  ,0) node[text=gray,alias=T'B4start](s5tos6){$s_5\xshortto{d} s_6$}
				++(1  ,0) node[text=gray,alias=T'B3start](s0tos6){$s_0\xshortto{d} s_6$}

				++(1  ,0) node[text=gray,alias=TtauB0start](s4tos5){$s_4\xshortto{\tau} s_5$}
				++(1  ,0) node[text=gray](s4tos3){$s_4\xshortto{\tau} s_3$}
				++(1  ,0) node[text=gray](s2tos3){$s_2\xshortto{\tau} s_3$}
				++(1  ,0) node[text=gray](s1tos3){$s_1\xshortto{\tau} s_3$}

				++(1.1,0) node[text=gray,inert](s3tos5){$s_3\xshortto{\tau} s_5$}
				++(1  ,0) node[text=gray,inert](s4tos2){$s_4\xshortto{\tau} s_2$}
				++(1  ,0) node[text=gray,inert](s2tos1){$s_2\xshortto{\tau} s_1$}
			;
		\end{scope}

		\begin{scope}[every node/.style={rounded rectangle,minimum height=6mm,draw,inner ysep=0pt,inner xsep=1pt,minimum width=6mm}]
			\footnotesize
				\draw[overlay] (TB0start) 		++(-0.1 ,1.5) node(TB0){$T_{B_0\shortto}^{\phantom{\prime}}$};
				\draw (TB3start) 		++( 0   ,1.5) node(TB3){$T_{B_3\shortto}^{\phantom{\prime}}$};
			 	\draw (T''B0start) 	++( 0.1 ,1.5) node(T''B0){$T''_{B_0\shortto}$};
			 	\draw (T''B4start) 	++( 0.5 ,1.5) node(T''B4){$T''_{B_4\shortto}$};
			 	\draw (TB2start) 		++(-0.2 ,1.5) node(TB2){$T_{B_2\shortto}^{\phantom{\prime}}$};
			 	\draw (T''B2start) 	++(-0.1 ,1.5) node(T''B2){$T''_{B_2\shortto}$};
			 	\draw (T'B0start) 	++( 0   ,1.5) node(T'B0){$T'_{B_0\shortto}$};
			 	\draw (T'B4start) 	++( 0.1 ,1.5) node(T'B4){$T'_{B_4\shortto}$};
			 	\draw (T'B3start) 	++( 0.2 ,1.5) node(T'B3){$T'_{B_3\shortto}$};
			 	\draw (TtauB0start) ++( 1.5 ,1.5) node(TtauB0){$T_\tau^{\phantom{\prime}}{}_{B_0\shortto}^{\phantom{\prime}}$};
		\end{scope}

		\begin{scope}[color=gray,line cap=round,decoration={brace,amplitude=2.75pt},below,rectangle,text=black]
		 	\draw[decorate] (s2tos1)       ++(0.485,-0.34) -- node{\footnotesize inert} ++(-2.97,0);
		\end{scope}

		\draw[->] (s1tos8.north) to (s1tos8 |- TB0.south);
		\draw[->] (s0tos8.north) to (s0tos8 |- TB3.south);
		\draw[->] (s1tos7.north) to (s1tos7 |- T''B0.south);
		\draw[->] (s3tos7.north) to (T''B4);
		\draw[->] (s5tos7.north) to (T''B4);
		\draw[->] (s8tos8.north) to (s8tos8 |- TB2.south);
		\draw[->] (s7tos7.north) to (s7tos7 |- T''B2.south);
		\draw[->] (s2tos6.north) to (s2tos6 |- T'B0.south);
		\draw[->] (s5tos6.north) to (s5tos6 |- T'B4.south);
		\draw[->] (s0tos6.north) to (s0tos6 |- T'B3.south);
		\draw[->] (s4tos5.north) to (TtauB0);
		\draw[->] (s4tos3.north) to (TtauB0);
		\draw[->] (s2tos3.north) to (TtauB0);
		\draw[->] (s1tos3.north) to (TtauB0);

	\end{tikzpicture}
	}
	\subcaption{The forest of transitions per block-bunch-slice.\label{fig:exampletransitions-block-bunch-slices}}
	\end{subfigure}

	\medskip
	\begin{subfigure}{\linewidth}
	\centering
	\scalebox{.95}{
	\begin{tikzpicture}[>=Stealth,scale=0.85,line cap=rect]
		\begin{scope}[every node/.style={rectangle,draw,minimum width=8.5mm,minimum height=5mm,anchor=mid,inner sep=0pt},
				inert/.style={line cap=rect,dash pattern=on 1.2mm off 1.05mm,dash phase=0.85mm}]
			\draw (-0.1,5)
				++(0   ,0) node[text=gray,alias=s0start](s0tos8){$s_0\xshortto{c} s_8$}
				++(1   ,0) node[text=gray](s0tos6){$s_0\xshortto{d} s_6$}

				++(1.1 ,0) node[text=gray,alias=s1start](s1tos8){$s_1\xshortto{b} s_8$}
				++(1   ,0) node[text=gray](s1tos7){$s_1\xshortto{a} s_7$}
				++(1   ,0) node[text=gray](s1tos3){$s_1\xshortto{\tau} s_3$}

				++(1.1 ,0) node[text=gray,alias=s2start](s2tos6){$s_2\xshortto{d} s_6$}
				++(1   ,0) node[text=gray](s2tos3){$s_2\xshortto{\tau} s_3$}
				++(1   ,0) node[text=gray,inert](s2tos1){$s_2\xshortto{\tau} s_1$}

				++(1.1 ,0) node[text=gray,alias=s3start](s3tos7){$s_3\xshortto{a} s_7$}
				++(1   ,0) node[text=gray,inert](s3tos5){$s_3\xshortto{\tau} s_5$}

				++(1.1 ,0) node[text=gray,alias=s4start](s4tos3){$s_4\xshortto{\tau} s_3$}
				++(1   ,0) node[text=gray](s4tos5){$s_4\xshortto{\tau} s_5$}
				++(1   ,0) node[text=gray,inert](s4tos2){$s_4\xshortto{\tau} s_2$}

				++(1.1 ,0) node[text=gray,alias=s5start](s5tos7){$s_5\xshortto{a} s_7$}
				++(1   ,0) node[text=gray](s5tos6){$s_5\xshortto{d} s_6$}

				++(1.1 ,0) node[text=gray,alias=s7start](s7tos7){$s_7\xshortto{b} s_7$}

				++(1.1 ,0) node[text=gray,alias=s8start](s8tos8){$s_8\xshortto{b} s_8$}
			;
		\end{scope}

		\begin{scope}[every node/.style={rounded rectangle,minimum height=6mm,draw,inner ysep=0pt,inner xsep=1pt,minimum width=6mm}]
			\footnotesize
			 \draw (s0start) ++( 0  ,1.5) node(s0T){$T$};
			 \draw (s0start) ++( 1  ,1.5) node(s0T'){$T\smash{'}$};
			 \draw (s1start) ++( 0  ,1.5) node(s1T){$T$};
			 \draw (s1start) ++( 1  ,1.5) node(s1T''){$T\smash{''}$};
			 \draw (s1start) ++( 2  ,1.5) node(s1Ttau){$\smash[b]{T_\tau}$};
			 \draw (s2start) ++( 0  ,1.5) node(s2T'){$T\smash{'}$};
			 \draw (s2start) ++( 1  ,1.5) node(s2Ttau){$\smash[b]{T_\tau}$};
			 \draw (s3start) ++( 0  ,1.5) node(s3T''){$T\smash{''}$};
			 \draw (s4start) ++( 0.5,1.5) node(s4Ttau){$\smash[b]{T_\tau}$};
			 \draw (s5start) ++( 0  ,1.5) node(s5T''){$T\smash{''}$};
			 \draw (s5start) ++( 1  ,1.5) node(s5T'){$T\smash{'}$};
			 \draw (s7start) ++( 0  ,1.5) node(s7T''){$T\smash{''}$};
			 \draw (s8start) ++( 0  ,1.5) node(s8T){$T$};
		\end{scope}

		\begin{scope}[color=gray,line cap=round,decoration={brace,amplitude=2.75pt},below,rectangle,text=black]
			 \draw[decorate] (s2tos1) ++(0.485,-0.34) -- node{\footnotesize inert} ++(-0.97,0);
			 \draw[decorate] (s3tos5) ++(0.485,-0.34) -- node{\footnotesize inert} ++(-0.97,0);
			 \draw[decorate] (s4tos2) ++(0.485,-0.34) -- node{\footnotesize inert} ++(-0.97,0);
		\end{scope}
		\begin{scope}[every node/.style={rounded rectangle,minimum height=6mm,draw,inner ysep=0pt,inner xsep=1pt,minimum width=6mm}]
		 \draw (s0T) 		++( 0.5,1.5) node(outs0){\footnotesize $\mathit{out}(s_0)$};
		 \draw (s1T) 		++( 1  ,1.5) node(outs1){\footnotesize $\mathit{out}(s_1)$};
		 \draw (s2T') 	++( 1  ,1.5) node(outs2){\footnotesize $\mathit{out}(s_2)$};
		 \draw (s3T'') 	++( 0.5,1.5) node(outs3){\footnotesize $\mathit{out}(s_3)$};
		 \draw (s4Ttau) ++( 0.5,1.5) node(outs4){\footnotesize $\mathit{out}(s_4)$};
		 \draw (s5T'') 	++( 0.5,1.5) node(outs5){\footnotesize $\mathit{out}(s_5)$};
		 \draw (s7T'') 	++(-0.1,1.5) node(outs7){\footnotesize $\mathit{out}(s_7)$};
		 \draw[overlay] (s8T) 		++( 0.2,1.5) node(outs8){\footnotesize $\mathit{out}(s_8)$};
		\end{scope}

		\draw[->] (s0tos8.north) to (s0T);
		\draw[->] (s0tos6.north) to (s0T');

		\draw[->] (s1tos8.north) to (s1T);
		\draw[->] (s1tos7.north) to (s1T'');
		\draw[->] (s1tos3.north) to (s1Ttau);

		\draw[->] (s2tos6.north) to (s2T');
		\draw[->] (s2tos3.north) to (s2Ttau);
		\draw[->] (s2tos1.north) to (outs2);

		\draw[->] (s3tos7.north) to (s3T'');
		\draw[->] (s3tos5.north) to (outs3);

		\draw[->] (s4tos3.north) to (s4Ttau);
		\draw[->] (s4tos5.north) to (s4Ttau);
		\draw[->] (s4tos2.north) to (outs4);

		\draw[->] (s5tos7.north) to (s5T'');
		\draw[->] (s5tos6.north) to (s5T');

		\draw[->] (s7tos7.north) to (s7T'');

		\draw[->] (s8tos8.north) to (s8T);

		\draw[->] (s0T) 		to (outs0);
		\draw[->] (s0T') 		to (outs0);
		\draw[->] (s1T) 		to (outs1);
		\draw[->] (s1T'') 	to (outs1);
		\draw[->] (s1Ttau) 	to (outs1);
		\draw[->] (s2T') 		to (outs2);
		\draw[->] (s2Ttau) 	to (outs2);
		\draw[->] (s3T'') 	to (outs3);
		\draw[->] (s4Ttau) 	to (outs4);
		\draw[->] (s5T'') 	to (outs5);
		\draw[->] (s5T') 		to (outs5);
		\draw[->] (s7T'') 	to (s7T'' |- outs7.south);
		\draw[->] (s8T) 		to (s8T |- outs8.south);
	\end{tikzpicture}
	}
	\subcaption{The forest of transitions per source state, grouped into non-inert and inert transitions, and non-inert transitions grouped per bunch.\label{fig:exampletransitions-source-state}}
	\end{subfigure}

	\medskip
	\begin{subfigure}{\linewidth}
	\centering
	\scalebox{.95}{
	\begin{tikzpicture}[>=Stealth,scale=0.85,line cap=rect]
		\begin{scope}[every node/.style={rectangle,draw,minimum width=8.5mm,minimum height=5mm,anchor=mid,inner sep=0pt},
				inert/.style={line cap=rect,dash pattern=on 1.2mm off 1.05mm,dash phase=0.85mm}]
			\draw (-0.1,5)
				++(0   ,0) node[text=gray,inert,alias=s1start](s2tos1){$s_2\xshortto{\tau} s_1$}

				++(1   ,0) node[text=gray,inert,alias=s2start](s4tos2){$s_4\xshortto{\tau} s_2$}

				++(1   ,0) node[text=gray,alias=s3start](s1tos3){$s_1\xshortto{\tau} s_3$}
				++(1   ,0) node[text=gray](s2tos3){$s_2\xshortto{\tau} s_3$}
				++(1   ,0) node[text=gray](s4tos3){$s_4\xshortto{\tau} s_3$}

				++(1   ,0) node[text=gray,alias=s5start](s4tos5){$s_4\xshortto{\tau} s_5$}
				++(1   ,0) node[text=gray,inert](s3tos5){$s_3\xshortto{\tau} s_5$}

				++(1   ,0) node[text=gray,alias=s6start](s0tos6){$s_0\xshortto{d} s_6$}
				++(1   ,0) node[text=gray](s2tos6){$s_2\xshortto{d} s_6$}
				++(1   ,0) node[text=gray](s5tos6){$s_5\xshortto{d} s_6$}

				++(1   ,0) node[text=gray,alias=s7start](s1tos7){$s_1\xshortto{a} s_7$}
				++(1   ,0) node[text=gray](s3tos7){$s_3\xshortto{a} s_7$}
				++(1   ,0) node[text=gray](s5tos7){$s_5\xshortto{a} s_7$}
				++(1   ,0) node[text=gray](s7tos7){$s_7\xshortto{b} s_7$}

				++(1   ,0) node[text=gray,alias=s8start](s0tos8){$s_0\xshortto{c} s_8$}
  			++(1   ,0) node[text=gray](s1tos8){$s_1\xshortto{b} s_8$}
				++(1   ,0) node[text=gray](s8tos8){$s_8\xshortto{b} s_8$}
			;
		\end{scope}

		\begin{scope}[color=gray,line cap=round,decoration={brace,amplitude=2.75pt},below,rectangle,text=black]
			 \draw[decorate] (s2tos1)       ++(0.485,-0.34) -- node{\footnotesize inert} ++(-0.97,0);
			 \draw[decorate] (s3tos5)       ++(0.485,-0.34) -- node{\footnotesize inert} ++(-0.97,0);
			 \draw[decorate] (s4tos2)       ++(0.485,-0.34) -- node{\footnotesize inert} ++(-0.97,0);
		\end{scope}
		\begin{scope}[every node/.style={rounded rectangle,minimum height=6mm,draw,inner ysep=0pt,inner xsep=1pt,minimum width=6mm}]
		 \draw[overlay] (s1start) ++( -0.1  ,1.5) node(ins1){\footnotesize $\mathit{in}(s_1)$};
		 \draw (s2start) ++( 0.1  ,1.5) node(ins2){\footnotesize $\mathit{in}(s_2)$};
		 \draw (s3start) ++( 1  ,1.5) node(ins3){\footnotesize $\mathit{in}(s_3)$};
		 \draw (s5start) ++( 0.5,1.5) node(ins5){\footnotesize $\mathit{in}(s_5)$};
		 \draw (s6start) ++( 1,1.5) node(ins6){\footnotesize $\mathit{in}(s_6)$};
		 \draw (s7start) ++( 1.5  ,1.5) node(ins7){\footnotesize $\mathit{in}(s_7)$};
		 \draw (s8start) ++( 1,1.5) node(ins8){\footnotesize $\mathit{in}(s_8)$};
		\end{scope}
		\draw[->] (s3tos7.north) to (ins7);
		\draw[->] (s1tos7.north) to (ins7);
		\draw[->] (s5tos7.north) to (ins7);
		\draw[->] (s7tos7.north) to (ins7);

		\draw[->] (s1tos8.north) to (ins8);
		\draw[->] (s0tos8.north) to (ins8);
		\draw[->] (s8tos8.north) to (ins8);

		\draw[->] (s5tos6.north) to (ins6);
		\draw[->] (s2tos6.north) to (ins6);
		\draw[->] (s0tos6.north) to (ins6);

		\draw[->] (s4tos5.north) to (ins5);
		\draw[->] (s4tos3.north) to (ins3);
		\draw[->] (s2tos3.north) to (ins3);
		\draw[->] (s1tos3.north) to (ins3);

		\draw[->] (s2tos1.north) to (s2tos1 |- ins1.south);
		\draw[->] (s3tos5.north) to (ins5);
		\draw[->] (s4tos2.north) to (s4tos2 |- ins2.south);
	\end{tikzpicture}
	}

	\subcaption{The forest of transitions per goal state, grouped into non-inert and inert transitions.\label{fig:exampletransitions-target-state}}
	\end{subfigure}

	\caption{Four refinable partition instances for the transitions of the example in Figure~\ref{fig:examplepartition}.\label{fig:exampletransitions}}
\end{figure}

\paragraph{Block-bunch-slices}
are also stored in lists.
We store, per block, a list of its stable block-bunch-slices,
and additionally one global list containing all unstable block-bunch-slices, called the splitter list.

When a block is split, its list of stable block-bunch-slices needs to be distributed over the two blocks.
This does not require additional time complexity over splitting the block-bunch-slices themselves.
New stable block-bunch-slices are inserted into the list of the new block.
New unstable block-bunch-slices are inserted into the splitter list.

We obviously need the unstable block-bunch-slices at Line~\ref{alg-line-stabilize-begin},
and we need the stable block-bunch-slices of block $N$ at Line~\ref{alg:line-make-all-slices-unstable}.
We store a stability flag with each block-bunch-slice
to decide whether a split-off new block-bunch-slice should go into the stable or the unstable list.
When executing Line~\ref{alg:line-make-all-slices-unstable}, we now need to clear the stability flag of every stable block-bunch-slice of $N$.
As every block-bunch-slice of $N$ either already contains a transition from a new bottom state,
or will contain a transition from a new bottom state after it has been used as a splitter,
we assign the runtime needed to clear this flag to the present and future new bottom states of~$N$.

Care needs to be taken that $T_{U\pijl{}} \setminus T_{U\pijl{a} B'}$ can be found at Line~\ref{alg:line-remove-TU-after-primary-split}.
We ensure this as follows.
At Line~\ref{alg-line-find-splittable-block}, the primary splitter $\TBaB$ and the secondary splitter $\TB \setminus \TBaB$ are added to the splitter list in this order.
After $\TBaB$ has been removed from the list (Line~\ref{alg:removeTBpijl}),
$\TB \setminus \TBaB$ is the first element of the remaining list.
At Line~\ref{alg-split1}, this is split into $T_{R\pijl{}} \setminus T_{R\pijl{a} B'}$ and $T_{U\pijl{}} \setminus T_{U\pijl{a} B'}$,
in an order that depends on which subblock is the smaller.
So either the first or the second element of the splitter list is the required slice at Line~\ref{alg:line-remove-TU-after-primary-split}.

\subsection{Several small optimisations}

We mention a few additional optimisations that our implementation uses, which are not essential for the complexity,
but speed up the implementation.

In cases when we mark all transitions in a block-bunch-slice (Lines~\ref{alg-line-mark-all-transitions-in-primary} and~\ref{alg:new_bunch}),
we instead add their source states to $R$ immediately.

In Line~\ref{alg:line-make-all-slices-unstable},
we actually know that $N$ is stable under $R \pijl{\tau} U$ because that was the splitter applied last,
so we do not make $R \pijl{\tau} U$ unstable.
Also, if $\Bottom(N) \setminus \Bottom(R) = \emptyset$, $N$ is stable under $T'_{N \pijl{}} \subseteq \TB'$
(because $R$ was stable under $\TB'$,
and stability is preserved
if no more new bottom states are found at Line~\ref{alg:split-call2}),
so we do not make $T'_{N \pijl{}}$ unstable.


\section{Benchmarks}
\label{sec:benchmarks}

The new algorithm (JGKW19) has been implemented in the mCRL2 toolset~\cite{DBLP:conf/tacas/BunteGKLNVWWW19}, and is available in its 201908.0 release.
This toolset also contains implementations of various other algorithms, such as the algorithm by Groote and Vaandrager (GV)~\cite{GV90}
and the GJKW algorithm of~\cite{GJKW2017}, which we refer to as GJKW17.
In addition, it offers an implementation of the partition-refinement algorithm using state signatures by Blom and Orzan (BO)~\cite{BO2003}.
For each state, a signature is maintained describing which blocks the state can reach directly via its outgoing transitions.
Although its time complexity is $O(mn^2)$, in some cases, it is known to outperform GV.

In this section, we report on the experiments we have conducted to compare GV, BO, GJKW17 and JGKW19 when applied to practical examples.
All experiments involve the branching bisimulation minimisation of a given LTS, which GJKW17 first transforms into a Kripke structure. Note that for an LTS of $n$ states and $m$ transitions, this transformation results in a Kripke structure consisting of $n + m$ states and $2m$ transitions in the worst case.

The set of benchmarks consists of all LTSs offered by the VLTS benchmark set\footnote{\url{http://cadp.inria.fr/resources/vlts}.} with at least 60,000 transitions, plus
three cases that have been derived from models distributed with the mCRL2 toolset. These models are:
\begin{enumerate}
\item \textbf{lift6-final}: this model is based on an elevator model, extended to six elevators;
\item \textbf{dining\_14}: this is the dining philosophers model with 14 philosophers;
\item \textbf{1394-fin3}: this model is an altered version of the 1394-fin model, extended to three processes and two data elements.
\end{enumerate}

Table~\ref{tab:chars} presents the structural characteristics for each benchmark: the number of states ($n$), the number of transitions ($m$), the number of $\tau$-transitions ($m_\tau$), the number of actions ($\mid$\textit{Act}$\mid$), and the number of states and transitions after branching bisimulation reduction (min.\ $n$ and min. $m$, respectively).

\jk{Opmerking Tim: waarom geen min. $m_\tau$ in de tabel?}
\begin{table}[t]
\scriptsize
\centering
\caption{Structural characteristics of the benchmark LTSs.}
\begin{tabular}{lrrrrrr}
\hline
\textbf{model} & \textbf{\textit{n}} & \textbf{\textit{m}} & \textbf{\textit{m}}$_\tau$ &  \textbf{\textit{$\mid$Act$\mid$}} & \textbf{min.}\ \textbf{\textit{n}} & \textbf{min.}\ \textbf{\textit{m}} \\
\hline
\hline
vasy\_40\_60 & 40,006 & 60,007 & 20,003 & 4 & 20,003 & 40,004 \\ \hline
vasy\_18\_73 & 18,746 & 73,043 & 39,217 & 18 & 2,326 & 9,751 \\ \hline
vasy\_157\_297 & 157,604 & 297,000 & 31,798 & 236 & 3,038 & 12,095 \\ \hline
vasy\_52\_318 & 52,268 & 318,126 & 130,752 & 18 & 66 & 333 \\ \hline
vasy\_83\_325 & 83,436 & 325,584 & 45,696 & 212 & 42,195 & 197,200 \\ \hline
vasy\_116\_368 & 116,456 & 368,569 & 263,296 & 22 & 22,398 & 87,674 \\ \hline
vasy\_720\_390 & 720,247 & 390,999 & 1 & 50 & 3,278 & 116,537 \\ \hline
vasy\_69\_520 & 69,754 & 520,633 & 1 & 136 & 69,753 & 520,632 \\ \hline
cwi\_371\_641 & 371,804 & 641,565 & 445,600 & 62 & 2,134 & 5,634 \\ \hline
vasy\_166\_651 & 166,464 & 651,168 & 91,392 & 212 & 42,195 & 197,200 \\ \hline
cwi\_214\_684 & 214,202 & 684,419 & 550,611 & 6 & 478 & 1,612 \\ \hline
cwi\_142\_925 & 142,472 & 925,429 & 862,298 & 8 & 23 & 49 \\ \hline
vasy\_386\_1171 & 386,496 & 1,171,872 & 122,976 & 74 & 71 & 108 \\ \hline
vasy\_66\_1302 & 66,929 & 1,302,664 & 117,866 & 82 & 51,128 & 1,018,692 \\ \hline
vasy\_164\_1619 & 164,865 & 1,619,204 & 109,910 & 38 & 992 & 3,456 \\ \hline
vasy\_65\_2621 & 65,537 & 2,621,480 & 0 & 72 & 65,536 & 2,621,440 \\ \hline
cwi\_566\_3984 & 566,640 & 3,984,157 & 3,666,614 & 12 & 198 & 791 \\ \hline
vasy\_1112\_5290 & 1,112,490 & 5,290,860 & 0 & 23 & 265 & 1,300 \\ \hline
cwi\_2165\_8723 & 2,165,446 & 8,723,465 & 3,830,225 & 27 & 4,256 & 20,880 \\ \hline
vasy\_6120\_11031 & 6,120,718 & 11,031,292 & 3,152,976 & 126 & 2,505 & 5,358 \\ \hline
vasy\_2581\_11442 & 2,581,374 & 11,442,382 & 2,508,518 & 224 & 704,737 & 3,972,600 \\ \hline
vasy\_574\_13561 & 574,057 & 13,561,040 & 0 & 141 & 3,577 & 16,168 \\ \hline
vasy\_4220\_13944 & 4,220,790 & 13,944,372 & 2,546,649 & 224 & 1,186,266 & 6,863,329 \\ \hline
vasy\_4338\_15666 & 4,338,672 & 15,666,588 & 3,127,116 & 224 & 704,737 & 3,972,600 \\ \hline
cwi\_2416\_17605 & 2,416,632 & 17,605,592 & 17,490,904 & 16 & 730 & 2,899 \\ \hline
vasy\_6020\_19353 & 6,020,550 & 19,353,474 & 17,526,144 & 512 & 256 & 510 \\ \hline
vasy\_11026\_24660 & 11,026,932 & 24,660,513 & 2,748,559 & 120 & 775,618 & 2,454,834 \\ \hline
lift6-final & 6,047,527 & 26,539,368 & 12,668,580 & 31 & 1,699 & 9,870 \\ \hline
vasy\_12323\_27667 & 12,323,703 & 27,667,803 & 3,153,502 & 120 & 876,944 & 2,780,022 \\ \hline
vasy\_8082\_42933 & 8,082,905 & 42,933,110 & 2,535,944 & 212 & 290 & 680 \\ \hline
cwi\_7838\_59101 & 7,838,608 & 59,101,007 & 22,842,122 & 21 & 62,031 & 470,230 \\ \hline
dining\_14 & 18,378,370 & 164,329,284 & 142,722,790 & 15 & 228,486 & 2,067,856 \\ \hline
cwi\_33949\_165318 & 33,949,609 & 165,318,222 & 74,133,306 & 32 & 12,463 & 71,466 \\ \hline
1394-fin3 & 126,713,623 & 276,426,688 & 172,900,987 & 104 & 160,258 & 538,936 \\ \hline
\end{tabular}
\label{tab:chars}
\end{table}

All experiments have been conducted on individual nodes of the DAS-5 cluster~\cite{DAS5}. Each of these nodes was running \textsc{CentOS Linux}
7.4, had an \textsc{Intel Xeon} E5-2698-v3 2.3GHz CPU, and was equipped with 256 GB RAM.
The experiments were performed using development version 201808.0.c59cfd413f of mCRL2.\footnote{\url{https://github.com/mCRL2org/mCRL2/commit/c59cfd413f}}


\begin{table}[tbp]
\scriptsize
\centering
\rotatebox{90}{\begin{varwidth}{\textheight}\centering
\parbox{16cm}{\caption{Running time and memory use results for GV, BO, GJKW17 and JGKW19. \og and \ob: significantly better (resp.\ worse) than all three other algorithms.\label{tab:results}
}}
\begin{tabular}{lrrrrrrrr}
\hline
\multicolumn{1}{c}{\multirow{2}{*}{\textbf{model}}} & \multicolumn{4}{c}{\textbf{time}} & \multicolumn{4}{c}{\textbf{space}} \\
\multicolumn{1}{c}{} & \textbf{GV} & \textbf{BO} & \textbf{GJKW17} & \textbf{JGKW19} & \textbf{GV} & \textbf{BO} & \textbf{GJKW17} & \textbf{JGKW19} \\
\hline
\hline
vasy\_40\_60 & 24.\phantom{0}\phantom{0} s \ns & 138.\phantom{0}\phantom{0} s \ob & .1\phantom{0} s \ns & .05 s \og &          65.5\phantom{G}MB \ns &          60.6\phantom{G}MB \ns &          70\phantom{.}\phantom{0}\phantom{G}MB \ns &          60\phantom{.}\phantom{0}\phantom{G}MB \ns \\ \hline
vasy\_18\_73 & .21 s \ns & .37 s \ob & .11 s \ns & .07 s \og &          55.6\phantom{G}MB \ns &          56.7\phantom{G}MB \ns &          50\phantom{.}\phantom{0}\phantom{G}MB \ns &          50\phantom{.}\phantom{0}\phantom{G}MB \ns \\ \hline
vasy\_157\_297 & 1.7\phantom{0} s \ns & 2.\phantom{0}\phantom{0} s \ns & .4\phantom{0} s \ns & .2\phantom{0} s \og &          97.3\phantom{G}MB \ns &          94.3\phantom{G}MB \ns &        127.2\phantom{G}MB \ob &          90\phantom{.}\phantom{0}\phantom{G}MB \ns \\ \hline
vasy\_52\_318 & .31 s \ns & .9\phantom{0} s \ob & .2\phantom{0} s \ns & .2\phantom{0} s \ns &          73.4\phantom{G}MB \ns &          90.4\phantom{G}MB \ns &          90.6\phantom{G}MB \ob &          73.4\phantom{G}MB \ns \\ \hline
vasy\_83\_325 & 2.6\phantom{0} s \ob & 1.0\phantom{0} s \ns & .9\phantom{0} s \ns & .3\phantom{0} s \og &        116.2\phantom{G}MB \ns &       .11\phantom{.}\phantom{0}\phantom{0}\phantom{M}GB \ns &        230.5\phantom{G}MB \ob &       .10\phantom{.}\phantom{0}\phantom{0}\phantom{M}GB \ns \\ \hline
vasy\_116\_368 & .9\phantom{0} s \ns & 5.\phantom{0}\phantom{0} s \ob & .6\phantom{0} s \ns & .4\phantom{0} s \og &          92.8\phantom{G}MB \ns &        110.6\phantom{G}MB \ns &       .13\phantom{.}\phantom{0}\phantom{0}\phantom{M}GB \ob &          90\phantom{.}\phantom{0}\phantom{G}MB \ns \\ \hline
vasy\_720\_390 & .4\phantom{0} s \ns & .9\phantom{0} s \ns & .6\phantom{0} s \ns & .4\phantom{0} s \ns &        105.2\phantom{G}MB \ns &        103.2\phantom{G}MB \ns &       .19\phantom{.}\phantom{0}\phantom{0}\phantom{M}GB \ob &          95.9\phantom{G}MB \og \\ \hline
vasy\_69\_520 & 1.5\phantom{0} s \ns & 4.\phantom{0}\phantom{0} s \ob & 2.4\phantom{0} s \ns & .8\phantom{0} s \og &       .15\phantom{.}\phantom{0}\phantom{0}\phantom{M}GB \ns &       .15\phantom{.}\phantom{0}\phantom{0}\phantom{M}GB \ns &        358.1\phantom{G}MB \ob &        162.0\phantom{G}MB \ns \\ \hline
cwi\_371\_641 & 7.4\phantom{0} s \ob & 5.9\phantom{0} s \ns & 1.\phantom{0}\phantom{0} s \ns & .7\phantom{0} s \ns &       .17\phantom{.}\phantom{0}\phantom{0}\phantom{M}GB \ns &        229.0\phantom{G}MB \ob &        185.4\phantom{G}MB \ns &       .14\phantom{.}\phantom{0}\phantom{0}\phantom{M}GB \og \\ \hline
vasy\_166\_651 & 4.9\phantom{0} s \ob & 1.9\phantom{0} s \ns & 2.\phantom{0}\phantom{0} s \ns & .7\phantom{0} s \og &        157.5\phantom{G}MB \ns &        141.8\phantom{G}MB \ns &        342.9\phantom{G}MB \ob &        139.5\phantom{G}MB \og \\ \hline
cwi\_214\_684 & 1.4\phantom{0} s \ns & 9.\phantom{0}\phantom{0} s \ob & .5\phantom{0} s \ns & .5\phantom{0} s \ns &        140.7\phantom{G}MB \ns &        162.1\phantom{G}MB \ob &        152.0\phantom{G}MB \ns &       .13\phantom{.}\phantom{0}\phantom{0}\phantom{M}GB \ns \\ \hline
cwi\_142\_925 & 1.4\phantom{0} s \ob & .8\phantom{0} s \ns & 1.0\phantom{0} s \ns & .9\phantom{0} s \ns &        152.5\phantom{G}MB \ns &        117.9\phantom{G}MB \og &        156.6\phantom{G}MB \ob &        152.5\phantom{G}MB \ns \\ \hline
vasy\_386\_1171 & 1.4\phantom{0} s \ns & 2.\phantom{0}\phantom{0} s \ob & 1.3\phantom{0} s \ns & .9\phantom{0} s \og &        229.2\phantom{G}MB \ns &        210.1\phantom{G}MB \og &        273.4\phantom{G}MB \ob &        228.7\phantom{G}MB \ns \\ \hline
vasy\_66\_1302 & 3.\phantom{0}\phantom{0} s \ns & 4.7\phantom{0} s \ns & 5.\phantom{0}\phantom{0} s \ns & 2.2\phantom{0} s \og &       .23\phantom{.}\phantom{0}\phantom{0}\phantom{M}GB \og &        283.1\phantom{G}MB \ns &        618.1\phantom{G}MB \ob &        268.0\phantom{G}MB \ns \\ \hline
vasy\_164\_1619 & 2.0\phantom{0} s \ns & 5.\phantom{0}\phantom{0} s \ob & 3.\phantom{0}\phantom{0} s \ns & - \phantom{0}\phantom{0} \ns &       .25\phantom{.}\phantom{0}\phantom{0}\phantom{M}GB \ns &        235.4\phantom{G}MB \ns &        262.4\phantom{G}MB \ns &        245.0\phantom{G}MB \ns \\ \hline
vasy\_65\_2621 & 90\phantom{.}\phantom{0}\phantom{0} s \ob & 11.\phantom{0}\phantom{0} s \ns & 20\phantom{.}\phantom{0}\phantom{0} s \ns & 4.7\phantom{0} s \og &       .5\phantom{.}\phantom{0}\phantom{0}\phantom{0}\phantom{M}GB \ns &        534.7\phantom{G}MB \ns &     1.8\phantom{.}\phantom{0}\phantom{0}\phantom{0}\phantom{M}GB \ob &       .5\phantom{.}\phantom{0}\phantom{0}\phantom{0}\phantom{M}GB \ns \\ \hline
cwi\_566\_3984 & 8.\phantom{0}\phantom{0} s \ns & 7.\phantom{0}\phantom{0} s \ns & 8.\phantom{0}\phantom{0} s \ns & 6.\phantom{0}\phantom{0} s \ns &       .5\phantom{.}\phantom{0}\phantom{0}\phantom{0}\phantom{M}GB \ns &        351.5\phantom{G}MB \og &        514.0\phantom{G}MB \ob &       .5\phantom{.}\phantom{0}\phantom{0}\phantom{0}\phantom{M}GB \ns \\ \hline
vasy\_1112\_5290 & 10.\phantom{0}\phantom{0} s \ns & 17.\phantom{0}\phantom{0} s \ob & 10\phantom{.}\phantom{0}\phantom{0} s \ns & 10\phantom{.}\phantom{0}\phantom{0} s \ns &       .8\phantom{.}\phantom{0}\phantom{0}\phantom{0}\phantom{M}GB \ns &        720.9\phantom{G}MB \ns &        931.5\phantom{G}MB \ns &       .7\phantom{.}\phantom{0}\phantom{0}\phantom{0}\phantom{M}GB \ns \\ \hline
cwi\_2165\_8723 & .4 min\phantom{0} \ns & 3.\phantom{0} min\phantom{0} \ob & - \phantom{0}\phantom{0} & .3 min\phantom{0} \ns &     1.3\phantom{.}\phantom{0}\phantom{0}\phantom{0}\phantom{M}GB \ns &     1.8726\phantom{.}\phantom{M}GB \ns &     2.1321\phantom{.}\phantom{M}GB \ob &     1.2\phantom{.}\phantom{0}\phantom{0}\phantom{0}\phantom{M}GB \ns \\ \hline
vasy\_6120\_11031 & 2.\phantom{0} min\phantom{0} \ob & 1.7 min\phantom{0} \ns & - \phantom{0}\phantom{0} & .4 min\phantom{0} \ns &     1.8\phantom{.}\phantom{0}\phantom{0}\phantom{0}\phantom{M}GB \ns &     1.7379\phantom{.}\phantom{M}GB \ns &     3.6596\phantom{.}\phantom{M}GB \ob &     1.5960\phantom{.}\phantom{M}GB \og \\ \hline
vasy\_2581\_11442 & 10\phantom{.}\phantom{0} min\phantom{0} \ob & 3.\phantom{0} min\phantom{0} \ns & - \phantom{0}\phantom{0} & - \phantom{0}\phantom{0} \ns &     1.5999\phantom{.}\phantom{M}GB \ns &     1.7434\phantom{.}\phantom{M}GB \ns &     4.1299\phantom{.}\phantom{M}GB \ob &     1.4\phantom{.}\phantom{0}\phantom{0}\phantom{0}\phantom{M}GB \og \\ \hline
vasy\_574\_13561 & 50\phantom{.}\phantom{0}\phantom{0} s \ns & 56.\phantom{0}\phantom{0} s \ns & - \phantom{0}\phantom{0} & - \phantom{0}\phantom{0} \ns &     1.8835\phantom{.}\phantom{M}GB \ob &     1.5217\phantom{.}\phantom{M}GB \ns &     1.5\phantom{.}\phantom{0}\phantom{0}\phantom{0}\phantom{M}GB \ns &     1.5\phantom{.}\phantom{0}\phantom{0}\phantom{0}\phantom{M}GB \ns \\ \hline
vasy\_4220\_13944 & 30\phantom{.}\phantom{0} min\phantom{0} \ob & 5.\phantom{0} min\phantom{0} \ns & - \phantom{0}\phantom{0} & .6 min\phantom{0} \ns &     2.0965\phantom{.}\phantom{M}GB \ns &     2.3188\phantom{.}\phantom{M}GB \ns &     5.8661\phantom{.}\phantom{M}GB \ob &     2.0\phantom{.}\phantom{0}\phantom{0}\phantom{0}\phantom{M}GB \og \\ \hline
vasy\_4338\_15666 & 34.\phantom{0} min\phantom{0} \ob & 3.\phantom{0} min\phantom{0} \ns & 2.\phantom{0} min\phantom{0} \ns & .8 min\phantom{0} \og &     2.4043\phantom{.}\phantom{M}GB \ns &     2.3559\phantom{.}\phantom{M}GB \ns &     5.9888\phantom{.}\phantom{M}GB \ob &     1.8535\phantom{.}\phantom{M}GB \og \\ \hline
cwi\_2416\_17605 & 30\phantom{.}\phantom{0}\phantom{0} s \ns & 19.\phantom{0}\phantom{0} s \ns & 20\phantom{.}\phantom{0}\phantom{0} s \ns & 20\phantom{.}\phantom{0}\phantom{0} s \ns &     1.6\phantom{.}\phantom{0}\phantom{0}\phantom{0}\phantom{M}GB \ns &     1.5157\phantom{.}\phantom{M}GB \og &     1.6638\phantom{.}\phantom{M}GB \ns &     1.6748\phantom{.}\phantom{M}GB \ob \\ \hline
vasy\_6020\_19353 & 25.\phantom{0}\phantom{0} s \ns & 40\phantom{.}\phantom{0}\phantom{0} s \ob & 6.\phantom{0}\phantom{0} s \ns & 5.\phantom{0}\phantom{0} s \ns &        870.\phantom{0}\phantom{G}MB \ns &     2.3442\phantom{.}\phantom{M}GB \ob &        870.\phantom{0}\phantom{G}MB \ns &        870.\phantom{0}\phantom{G}MB \ns \\ \hline
vasy\_11026\_24660 & 50\phantom{.}\phantom{0} min\phantom{0} \ob & 20\phantom{.}\phantom{0} min\phantom{0} \ns & 3.\phantom{0} min\phantom{0} \ns & 1.\phantom{0} min\phantom{0} \og &     3.6475\phantom{.}\phantom{M}GB \ns &     4.0513\phantom{.}\phantom{M}GB \ns &     9.6425\phantom{.}\phantom{M}GB \ob &     3.4412\phantom{.}\phantom{M}GB \og \\ \hline
lift6-final & 1.0 min\phantom{0} \ns & 3.\phantom{0} min\phantom{0} \ob & - \phantom{0}\phantom{0} & .9 min\phantom{0} \ns &     3.3846\phantom{.}\phantom{M}GB \ns &     8.1984\phantom{.}\phantom{M}GB \ob &     6.2971\phantom{.}\phantom{M}GB \ns &     3.2125\phantom{.}\phantom{M}GB \og \\ \hline
vasy\_12323\_27667 & 40\phantom{.}\phantom{0} min\phantom{0} \ob & 10\phantom{.}\phantom{0} min\phantom{0} \ns & - \phantom{0}\phantom{0} & 1.\phantom{0} min\phantom{0} \ns &     4.0091\phantom{.}\phantom{M}GB \ns &     4.5371\phantom{.}\phantom{M}GB \ns &   10.6743\phantom{.}\phantom{M}GB \ob &     3.7298\phantom{.}\phantom{M}GB \og \\ \hline
vasy\_8082\_42933 & 2.\phantom{0} min\phantom{0} \ns & 5.\phantom{0} min\phantom{0} \ob & - \phantom{0}\phantom{0} & 2.\phantom{0} min\phantom{0} \ns &     6.1231\phantom{.}\phantom{M}GB \ns &     5.4358\phantom{.}\phantom{M}GB \og &     6.6896\phantom{.}\phantom{M}GB \ob &     5.4600\phantom{.}\phantom{M}GB \ns \\ \hline
cwi\_7838\_59101 & 5.\phantom{0} min\phantom{0} \ns & 100\phantom{.}\phantom{0} min\phantom{0} \ob & 6.\phantom{0} min\phantom{0} \ns & 3.\phantom{0} min\phantom{0} \ns &     6.5283\phantom{.}\phantom{M}GB \og &     8.3266\phantom{.}\phantom{M}GB \ns &   13.7899\phantom{.}\phantom{M}GB \ob &     6.7646\phantom{.}\phantom{M}GB \ns \\ \hline
dining\_14 & 17.\phantom{0} min\phantom{0} \ns & 20\phantom{.}\phantom{0} min\phantom{0} \ns & 20\phantom{.}\phantom{0} min\phantom{0} \ns & 10\phantom{.}\phantom{0} min\phantom{0} \ns &   20.4826\phantom{.}\phantom{M}GB \og &   21.7156\phantom{.}\phantom{M}GB \ns &   23.7810\phantom{.}\phantom{M}GB \ob &   20.9756\phantom{.}\phantom{M}GB \ns \\ \hline
cwi\_33949\_165318 & 11.\phantom{0} min\phantom{0} \ns & 80\phantom{.}\phantom{0} min\phantom{0} \ob & 20\phantom{.}\phantom{0} min\phantom{0} \ns & 8.\phantom{0} min\phantom{0} \ns &   22.7204\phantom{.}\phantom{M}GB \ns &   33.0351\phantom{.}\phantom{M}GB \ns &   37.8606\phantom{.}\phantom{M}GB \ob &   21.0611\phantom{.}\phantom{M}GB \og \\ \hline
1394-fin3 & 25. h\phantom{000000} \ob & 3. h\phantom{000000} \ns & .5 h\phantom{000000} \ns & .3 h\phantom{000000} \og &   37.4893\phantom{.}\phantom{M}GB \ns &   71.8698\phantom{.}\phantom{M}GB \ob &   53.2166\phantom{.}\phantom{M}GB \ns &   31.5132\phantom{.}\phantom{M}GB \og \\ \hline\hline
\textbf{Total} & \textbf{28. h\phantom{00000}} \ob & \textbf{8. h\phantom{00000}} \ns & \textbf{1.4 h\phantom{00000}} \ns & \textbf{.8 h\phantom{00000}} \og & \textbf{121.8\phantom{.}\phantom{.}\phantom{0}\phantom{M}GB} \ns & \textbf{176.33\phantom{.}\phantom{.}\phantom{M}GB} \ns & \textbf{194.2\phantom{.}\phantom{.}\phantom{.}\phantom{.}\phantom{M}GB} \ob & \textbf{112.0\phantom{.}\phantom{.}\phantom{.}\phantom{.}\phantom{M}GB} \og \\ \hline

\end{tabular}
\end{varwidth}}
\end{table}


Table~\ref{tab:results} presents the obtained results. On each benchmark, we have applied each of the four algorithms ten times, and report the mean
runtime (in seconds or minutes) and memory use (in MB or GB) of those ten runs. In the table, only the significant digits are listed, which are identified by first estimating the standard deviation,
given the ten results. Given results $x_0, \ldots, x_9$, the standard deviation \textit{std} is estimated to be $\textit{std} = \frac{(\Sigma_{0 \leq i \leq 9} x_i^2) - (\Sigma_{0 \leq i \leq 9} x_i)^2/10}{8.5}$~\cite{B69}.
For all presented data the estimated standard deviation is less than $20\%$ of the mean. Cases in which this is not true are indicated by `-' in Table~\ref{tab:results}.

Concerning the significant digits, a decimal dot indicates that the unit digit is significant.
If this dot is missing, there is one insignificant zero.
The estimated standard deviation is used to identify the significant digits.
For example, `3.6~s' has a standard deviation in $[0.01,0.1)$~s, `540.~MB' has a standard deviation in $[0.1,1)$~MB, and `100~s' has a standard deviation in $[1,10)$~s.

The \og-symbol after a table entry indicates that the measurement is significantly better than the corresponding measurements for the other three algorithms, and the \ob-symbol indicates that the measurement is significantly worse.
Here, the results are considered significant if, given a hundred tables such as Table~\ref{tab:results}, one table of running time (resp. memory) is expected to contain spuriously significant results.

Concerning the runtimes, clearly, GV and BO perform significantly worse than the other two algorithms,
and JGKW19 in many cases performs significantly better than the others.
In particular, it should be noted that, although GJKW17 has the same time complexity, JGKW19 often still outperforms GJKW17.
Concerning memory use, in the majority of cases GJKW17 uses more memory than the others,
while sometimes BO is the most memory-hungry.
JGKW19 is much more competitive, in many cases even outperforming every other algorithm.

Overall, the results show that when applied to practical cases, JGKW19 is generally the most efficient algorithm time-wise,
and when other algorithms have similar runtimes, it is almost always the most efficient memory-wise.
This combination makes JGKW19, the algorithm presented in this paper, currently the best option for branching bisimulation minimisation of LTSs.

\bibliographystyle{plain}
\bibliography{LTS-bisimulation}

\appendix

\section{Branching bisimilarity via translation to Kripke structures is $\Theta(m \log m)$}
\label{app:complexity_via_stuttering}

In the original paper \cite{GJKW2017},
it is stated that the complexity of determining branching bisimilarity is $\bigo{m (\log \size{\mathit{Act}} + \log n)}$
when using the following translation from LTS to Kripke structure,
and subsequently computing divergence-blind stuttering bisimilarity.

\begin{definition}[LTS Embedding~\cite{dNV1995}]
Let $A = (S, \mathit{Act}, \mathord{\pijl{}})$ be an LTS.
We construct the embedding of $A$ to be the Kripke structure $K_A = (S_A, \mathit{AP}, \to, L)$ as follows:
\begin{enumerate}
	\item $S_A = S \cup \{ \langle a,t \rangle \mid \exists s \in S . s \pijl{a} t \}$
	\item $\mathit{AP} = \mathit{Act} \cup \{ \bot \}$
	\item $\to$ is the least relation satisfying (for $s,t \in S$, $a \in \mathit{Act} \setminus \{ \tau \}$)
	\[\frac{s\pijl{a}{}t}{s\pijl{}{}\langle a,t\rangle} \qquad
	\frac{s\pijl{a}t}{\langle a,t\rangle \to t} \qquad
	\frac{s\pijl{\tau}t}{s \to t}\]
	\item
	$L(s)=\{\bot\}$ for $s\in S$ and $L(\langle a,t\rangle)=\{a\}$.
	\end{enumerate}
\end{definition}

The observations made were as follows:
an LTS with $n$ states and $m$ transitions, in the worst case,
has an embedding of $n + m$ states and $2m$ transitions,
so the algorithm requires $\bigo{m \log (n+m)}$,
or, since $m \leq \size{\mathit{Act}}n^2$,
$\bigo{m \log (n+\size{\mathit{Act}}n^2)} = \bigo{m (\log \size{\mathit{Act}} + \log n)}$ time.
The example LTS in Figure~\ref{fig:LTS} (left) illustrates that this bound is tight, and thus cannot be improved to $\bigo{m \log n}$.

To run the algorithm for branching bisimulation minimisation according to \cite{GJKW2017},
we first translate it to the Kripke structure in the same figure (right),
where for each pair $\langle \text{action}, \text{goal state} \rangle$ there is an extra state.
The initial partition is given in Figure~\ref{fig:KS-initial}:
we have a block $B$ for the original states and separate blocks per action label for the extra states.
We also have a single constellation $\mathbf{C}$ for all states, indicating that no splitters have been used until now.
Note that every bottom state in $B$ has a transition, so there is no need to separate states without transitions from those with
(a preprocessing step similar to Lines~\ref{alg-line-initialize-partitions-begin}--\ref{alg-line-initialize-Pi_s-end} in Algorithm~\ref{algo:main-algorithm-abstract} of this technical report).
Every state is a bottom state, except $s_1$, which has an inert transition to $s_0$.

We now run the algorithm of \cite{GJKW2017} and assume that it first handles the blocks for actions $a_1, \ldots, a_k$ as splitters.
Only $B$ can be split, as the other blocks only contain one state,
but every state in $B$ has a transition to the blocks $\{ \langle s_0, a_i \rangle \}$.
After $k$ splits of constellations, we get the situation on the left side of Figure~\ref{fig:KS-after-splits}.
Note that there are $k+1$ constellations now.

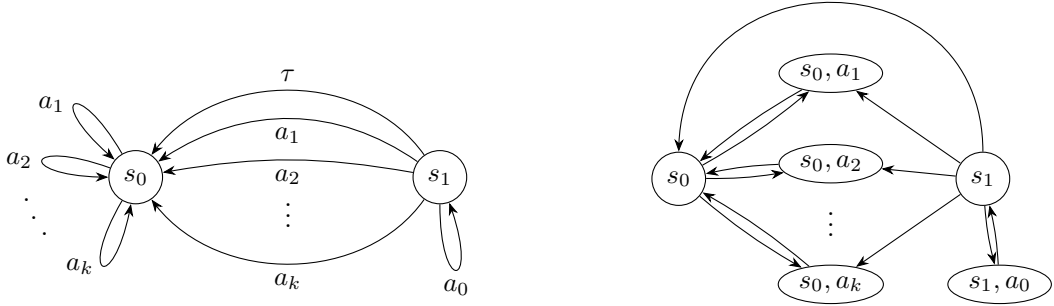
\begin{figure}[tbp]
	\begin{center}
		\begin{tikzpicture}[>=Stealth,scale=4]
			\node[circle,draw] (s0) at(0,0) {$s_0$};
			\node[circle,draw] (s1) at(1,0) {$s_1$};

			\draw[->] (s1) edge[bend right=55] node[above]{$\tau$}(s0)
			               edge[bend right=35] node[below]{$a_1$} (s0)
			               edge[bend right=10] node[below,align=center]{$a_2$ \\ $\vdots$} (s0)
			               edge[bend left= 55] node[below]{$a_k$} (s0);
			\draw[->] (s0) edge[out=120,in=140,looseness=25] node[left]{$a_1$} (s0)
			               edge[out=160,in=180,looseness=25] node(a2loop)[left]{$a_2$} (s0)
			               edge[out=240,in=260,looseness=25] node(akloop)[left]{$a_k$} (s0);

			\draw[->] (s1) edge[out=270,in=290,looseness=25] node[below]{$a_0$} (s1);

			\path (a2loop) edge[draw=none,bend right=20] node[near start]{.} node{.} node[near end]{.} (akloop);
		\end{tikzpicture}
		\hspace*{2cm}
		\begin{tikzpicture}[>=Stealth,scale=4]
			\node[circle,draw] (s0) at(0,0) {$s_0$};
			\node[circle,draw] (s1) at(1,0) {$s_1$};

			\node[ellipse,draw,inner sep=2pt] (s0a1) at(0.5  , 0.35) {$s_0,a_1$};
			\node[ellipse,draw,inner sep=2pt] (s0a2) at(0.5  , 0.05) {$s_0,a_2$};
			\node[ellipse,draw,inner sep=2pt] (s0ak) at(0.5  ,-0.35) {$s_0,a_k$};
			\node[ellipse,draw,inner sep=2pt] (s1a0) at(1.055,-0.35) {$s_1,a_0$};

			\draw[->] (s1) .. controls (1,0.75) and (0,0.75) .. (s0);
			\draw[->] (s1) edge                (s0a1)
			               edge                (s0a2)
			               edge                (s0ak)
			               edge[bend right=5]  (s1a0)
			          (s0) edge[bend right=5]  (s0a1)
			               edge[bend right=5]  (s0a2)
			               edge[bend right=5]  (s0ak)
			        (s0a1) edge[bend right=5]  (s0)
			        (s0a2) edge[bend right=5]  (s0)
			        (s0ak) edge[bend right=5]  (s0)
			        (s1a0) edge[bend right=5]  (s1);
			\path (s0a2) edge[draw=none] node[yshift=1mm]{$\vdots$} (s0ak);
		\end{tikzpicture}
	\end{center}
	\caption{A labelled transition system and its translation to a Kripke structure\label{fig:LTS}}
\end{figure}

\begin{figure}[tbp]
	\begin{center}
		\begin{tikzpicture}[>=Stealth,scale=4]
			\draw[very thin,fill=black!2]
				(-0.2,-0.475) rectangle (1.31,0.625);

			\draw[rounded corners=4mm,thick,fill=black!10]
				(0.27 , 0.25) rectangle (0.73 , 0.45)
				(0.27 ,-0.05) rectangle (0.73 , 0.15)
				(0.27 ,-0.45) rectangle (0.73 ,-0.25)
				(0.825,-0.45) rectangle (1.285,-0.25)
				(-0.15,0.6) |- (0.15,-0.15) |- (0.85,0.55) |- (1.15,-0.15) -- (1.15,0.6) -- cycle;

			\node[anchor=south west]               at (-0.2 ,-0.475) {$\mathbf{C}$};
			\node[anchor=north west,inner sep=5pt] at (-0.15, 0.55 ) {$B$};

			\node[circle,draw] (s0) at(0,0) {$s_0$};
			\node[circle,draw] (s1) at(1,0) {$s_1$};

			\node[ellipse,draw,inner sep=2pt] (s0a1) at(0.5  , 0.35) {$s_0,a_1$};
			\node[ellipse,draw,inner sep=2pt] (s0a2) at(0.5  , 0.05) {$s_0,a_2$};
			\node[ellipse,draw,inner sep=2pt] (s0ak) at(0.5  ,-0.35) {$s_0,a_k$};
			\node[ellipse,draw,inner sep=2pt] (s1a0) at(1.055,-0.35) {$s_1,a_0$};

			\draw[->] (s1) .. controls (1,0.75) and (0,0.75) .. (s0);
			\draw[->] (s1) edge                (s0a1)
			               edge                (s0a2)
			               edge                (s0ak)
			               edge[bend right=5]  (s1a0)
			          (s0) edge[bend right=5]  (s0a1)
			               edge[bend right=5]  (s0a2)
			               edge[bend right=5]  (s0ak)
			        (s0a1) edge[bend right=5]  (s0)
			        (s0a2) edge[bend right=5]  (s0)
			        (s0ak) edge[bend right=5]  (s0)
			        (s1a0) edge[bend right=5]  (s1);
			\path (s0a2) edge[draw=none] node[yshift=1mm]{$\vdots$} (s0ak);
		\end{tikzpicture}
	\end{center}
	\caption{The initial partition of the Kripke structure in Figure~\ref{fig:LTS}\label{fig:KS-initial} for the algorithm of \cite{GJKW2017}}
\end{figure}
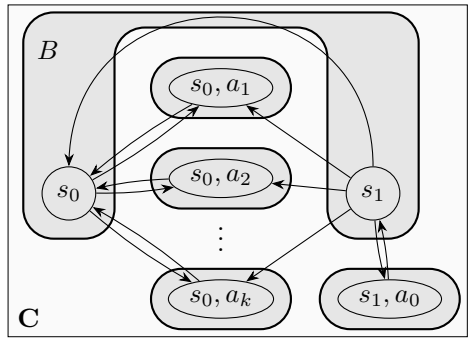

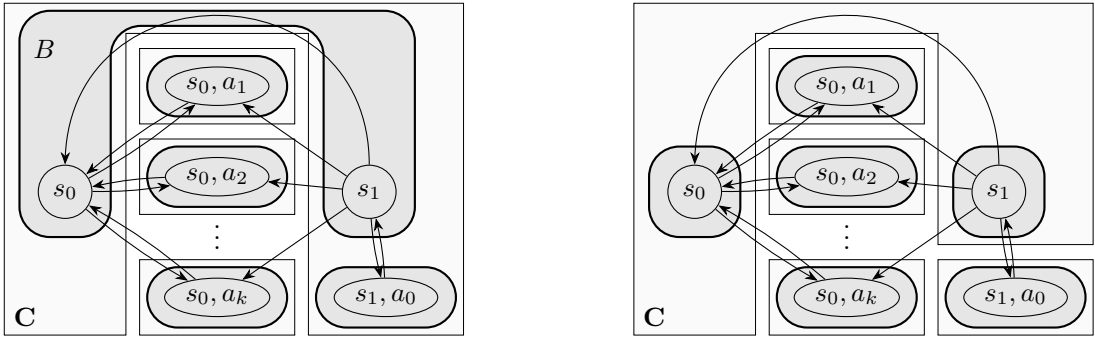
\begin{figure}[tbp]
	\begin{center}
		\begin{tikzpicture}[>=Stealth,scale=4]
			\draw[very thin,fill=black!2]
				(0.245, 0.225) rectangle (0.755, 0.475)
				(0.245,-0.075) rectangle (0.755, 0.175)
				(0.245,-0.475) rectangle (0.755,-0.225)
				(-0.2,-0.475) -| (0.2,0.525) -| (0.8,-0.475) -| (1.31,0.625) -| cycle;

			\draw[rounded corners=4mm,thick,fill=black!10]
				(0.27 , 0.25) rectangle (0.73 , 0.45)
				(0.27 ,-0.05) rectangle (0.73 , 0.15)
				(0.27 ,-0.45) rectangle (0.73 ,-0.25)
				(0.825,-0.45) rectangle (1.285,-0.25)
				(-0.15,0.6) |- (0.15,-0.15) |- (0.85,0.55) |- (1.15,-0.15) -- (1.15,0.6) -- cycle;

			\node[anchor=south west]               at (-0.2 ,-0.475) {$\mathbf{C}$};
			\node[anchor=north west,inner sep=5pt] at (-0.15, 0.55 ) {$B$};

			\node[circle,draw] (s0) at(0,0) {$s_0$};
			\node[circle,draw] (s1) at(1,0) {$s_1$};

			\node[ellipse,draw,inner sep=2pt] (s0a1) at(0.5  , 0.35) {$s_0,a_1$};
			\node[ellipse,draw,inner sep=2pt] (s0a2) at(0.5  , 0.05) {$s_0,a_2$};
			\node[ellipse,draw,inner sep=2pt] (s0ak) at(0.5  ,-0.35) {$s_0,a_k$};
			\node[ellipse,draw,inner sep=2pt] (s1a0) at(1.055,-0.35) {$s_1,a_0$};

			\draw[->] (s1) .. controls (1,0.75) and (0,0.75) .. (s0);
			\draw[->] (s1) edge                (s0a1)
			               edge                (s0a2)
			               edge                (s0ak)
			               edge[bend right=5]  (s1a0)
			          (s0) edge[bend right=5]  (s0a1)
			               edge[bend right=5]  (s0a2)
			               edge[bend right=5]  (s0ak)
			        (s0a1) edge[bend right=5]  (s0)
			        (s0a2) edge[bend right=5]  (s0)
			        (s0ak) edge[bend right=5]  (s0)
			        (s1a0) edge[bend right=5]  (s1);
			\path (s0a2) edge[draw=none] node[yshift=1mm]{$\vdots$} (s0ak);
		\end{tikzpicture}
		\hspace*{2cm}
		\begin{tikzpicture}[>=Stealth,scale=4]
			\draw[very thin,fill=black!2]
				(0.245, 0.225) rectangle (0.755, 0.475)
				(0.245,-0.075) rectangle (0.755, 0.175)
				(0.245,-0.475) rectangle (0.755,-0.225)
				(0.8  ,-0.475) rectangle (1.31 ,-0.225)
				(-0.2,-0.475) -| (0.2,0.525) -| (0.8,-0.175) -| (1.31,0.625) -| cycle;

			\draw[rounded corners=4mm,thick,fill=black!10]
				( 0.27 , 0.25) rectangle (0.73 , 0.45)
				( 0.27 ,-0.05) rectangle (0.73 , 0.15)
				( 0.27 ,-0.45) rectangle (0.73 ,-0.25)
				( 0.825,-0.45) rectangle (1.285,-0.25)
				(-0.15 ,-0.15) rectangle (0.15 , 0.15)
				( 0.85 ,-0.15) rectangle (1.15 , 0.15);

			\node[anchor=south west]               at (-0.2 ,-0.475) {$\mathbf{C}$};

			\node[circle,draw] (s0) at(0,0) {$s_0$};
			\node[circle,draw] (s1) at(1,0) {$s_1$};

			\node[ellipse,draw,inner sep=2pt] (s0a1) at(0.5  , 0.35) {$s_0,a_1$};
			\node[ellipse,draw,inner sep=2pt] (s0a2) at(0.5  , 0.05) {$s_0,a_2$};
			\node[ellipse,draw,inner sep=2pt] (s0ak) at(0.5  ,-0.35) {$s_0,a_k$};
			\node[ellipse,draw,inner sep=2pt] (s1a0) at(1.055,-0.35) {$s_1,a_0$};

			\draw[->] (s1) .. controls (1,0.75) and (0,0.75) .. (s0);
			\draw[->] (s1) edge                (s0a1)
			               edge                (s0a2)
			               edge                (s0ak)
			               edge[bend right=5]  (s1a0)
			          (s0) edge[bend right=5]  (s0a1)
			               edge[bend right=5]  (s0a2)
			               edge[bend right=5]  (s0ak)
			        (s0a1) edge[bend right=5]  (s0)
			        (s0a2) edge[bend right=5]  (s0)
			        (s0ak) edge[bend right=5]  (s0)
			        (s1a0) edge[bend right=5]  (s1);
			\path (s0a2) edge[draw=none] node[yshift=1mm]{$\vdots$} (s0ak);
		\end{tikzpicture}
	\end{center}
	\caption{The Kripke structure of Figure~\ref{fig:KS-initial} after $k$ and after $k+1$ splits\label{fig:KS-after-splits}}
\end{figure}

When we create the $(k+2)$th constellation,
block $B$ will fall apart because $s_1$ has an $a_0$-transition but $s_0$ has none.
State $s_1$ becomes a new bottom state.
Then, as one of the steps in the handling of new bottom states,
the constellations reachable from block $\{ s_1 \}$ are inserted into a balanced tree one-by-one.
If this operation uses time in $\Theta(\log \size{\text{resulting tree}})$,
all insertions together will take time in $\Theta(\log 1 + \log 2 + \cdots + \log (k+2)) = \Theta(\log k!) = \Theta(k \log k)$.
Note that $k$ is in $\Theta(m)$ or in $\Theta(\size{\mathit{Act}})$, but not in $\bigo{n}$.

Therefore, the overall runtime analysis of the algorithm in \cite{GJKW2017} cannot be improved to $\bigo{m \log n}$ without changing the algorithm.

\paragraph{Handling of this example by the new algorithm.}

Let us sketch how the algorithm that uses a partition of transitions would handle the LTS of Figure~\ref{fig:LTS} (left),
so as to convince the reader that it now only uses time in $\bigo{m \log n}$.
Note that, as the example has a constant number of states, the time complexity should degenerate to $\bigo{m \log 2} = \bigo{m}$.

The LTS has state space $S = \{ s_0, s_1 \}$, and every bottom state can do a visible transition, so $\Pi_s = \{ B_\mathrm{vis} \} = \{ S \}$.
The initial $\Pi_t$ puts all visible transitions into a single bunch, containing $k+1$ action-block-slices (one for each action label).
The action-block-slice $T_0 = \{ s_1 \pijl{a_0} s_1 \}$ contains one transition;
all other action-block-slices $T_i = \{ s_1 \pijl{a_i} s_0, s_0 \pijl{a_i} s_0 \}$ (for $i = 1, \ldots, k$) contain two transitions.
The algorithm will, at some time, create a new bunch for $T_0$
and either before or after handle $k-1$ other action-block-slices.

When a large action-block-slice $T_i$ (for $i = 1, \ldots, k$) is moved to its own bunch,
the algorithm creates a new block-bunch-slice for $(T_i)_{B_\mathrm{vis} \pijl{}}$, taking time in $\bigo{\size{T_i}}$,
and adds it to the splitter list.
Block $B_\mathrm{vis}$ is considered as unstable but the split is trivial;
also the split under the secondary splitter is trivial.
Finding out that these splits are trivial takes time in $\bigo{\size{T_i}}$ again
(the subroutine $\mathsf{split}$ will walk through all marked transitions before it finds that $U$ is empty).

When the small action-block-slice $T_0$ is moved to its own bunch,
the algorithm also creates a new block-bunch-slice $(T_0)_{B_\mathrm{vis} \pijl{}}$
and splits $B_\mathrm{vis}$ under it.
This is non-trivial and takes time in $\bigo{ \{s_1\}_{\pijl{}} + \{s_1\}_{\lijp{}} }$.
Now, as $s_1$ is a new bottom state, the algorithm will stabilize $\{ s_1 \}$ under all bunches that are reachable from it;
this takes time in $\bigo{\{ s_1 \}_{\pijl{}}}$ overall.

Handling further large action-block-slices afterwards will only lead to further trivial splits, each in time $\bigo{\size{T_i}}$.

Summing up over these time complexities,
we can see that every transition is handled a constant number of times, giving required time complexity $\bigo{m}$.

\end{document}